\documentclass[british, runningheads]{llncs}
\usepackage[T1]{fontenc}
\usepackage[latin9]{inputenc}
\usepackage{color}
\usepackage{calc}
\usepackage{units}
\usepackage{mathtools}
\usepackage{amsmath}
\usepackage{amssymb}
\usepackage{hyperref}
\usepackage{xspace}
\usepackage{cancel}
\usepackage{rotfloat}
\usepackage{array}
\usepackage{verbatim}
\usepackage{multirow}
\usepackage{bbding} 

\usepackage{ulem}

\usepackage{lmodern}

\newcommand{\subs}[2]{\ensuremath{#1{\mapsto}#2}}

\newcommand{\wrt}{w.r.t.\xspace}

\newcommand{\aka}{a.k.a.\xspace}
\newcommand{\eg}{e.g.\xspace}

\newcommand{\ie}{i.e.\xspace}
\newcommand{\cf}{cf.\xspace}

\newcommand*{\secref}[1]{Sec.~\ref{sec:#1}}
\newcommand*{\defref}[1]{Def.~\ref{def:#1}}

\newcommand{\as}[1]{}
\newcommand{\ron}[1]{}
\newcommand{\rui}[1]{}

\makeatletter

\providecommand{\tabularnewline}{\\}

\newenvironment{lyxlist}[1]
	{\begin{list}{}
		{\settowidth{\labelwidth}{#1}
		 \setlength{\leftmargin}{\labelwidth}
		 \addtolength{\leftmargin}{\labelsep}
		 }}
	{\end{list}}

\usepackage{todonotes}

\newtheorem{idea}{Idea}

\makeatother

\usepackage{babel}
\begin{document}
\title{A Formal Model to Prove Instantiation Termination for E-matching-Based
Axiomatisations\\{\normalsize (Extended Version)}}
\titlerunning{A Formal Model to Prove Instantiation Termination of E-matching}
\author{Rui Ge\textsuperscript{\Envelope}\orcidID{0000-0003-1049-8132} \and Ronald Garcia\orcidID{0000-0002-0982-1118} \and Alexander J. Summers\orcidID{0000-0001-5554-9381}}
\authorrunning{R. Ge et al.}
\institute{Department of Computer Science\\
University of British Columbia, Vancouver, BC, Canada\\
\email{\{rge, rxg\}@cs.ubc.ca \qquad{} alex.summers@ubc.ca}}

\maketitle
\normalem 

\global\long\def\FuncFont#1{\textit{#1}}%

\global\long\def\Member{\FuncFont{member}}%

\global\long\def\Subset{\FuncFont{subset}}%

\global\long\def\Union{\FuncFont{union}}%

\global\long\def\Inter{\FuncFont{inter}}%

\global\long\def\Diff{\FuncFont{diff}}%

\global\long\def\Add{\FuncFont{add}}%

\global\long\def\Remove{\FuncFont{remove}}%

\global\long\def\IsEmpty{\FuncFont{isEmpty}}%

\global\long\def\Empty{\FuncFont{empty}}%

\global\long\def\Card{\FuncFont{card}}%

\global\long\def\SSet{\FuncFont{Set}}%

\global\long\def\BBool{\FuncFont{Bool}}%

\global\long\def\NNat{\FuncFont{Nat}}%

\global\long\def\Skolem{\FuncFont{Sk}}%

\global\long\def\Singleton{\FuncFont{singleton}}%

\global\long\def\Disjoint{\FuncFont{disjoint}}%

\global\long\def\Equal{\FuncFont{equal}}%

\begin{abstract}
  SMT-based program analysis and verification often involve reasoning about program features that have been specified using quantifiers; incorporating quantifiers into SMT-based reasoning is, however, known to be challenging.
  If quantifier instantiation is not carefully controlled, then runtime and  outcomes can be brittle and hard to predict.
  In particular, uncontrolled quantifier instantiation can lead to unexpected incompleteness and even non-termination.
  E-matching is the most widely-used approach for controlling quantifier instantiation, but when axiomatisations are complex, even experts cannot tell if their use of E-matching guarantees completeness or termination.

  This paper presents a new formal model
  that facilitates the proof, once and for all, that giving a complex E-matching-based axiomatisation to an SMT solver, such as Z3 or cvc5, will not cause non-termination.
  Key to our technique is an operational semantics for solver behaviour that models how the E-matching rules common to most solvers are used to determine when quantifier instantiations are enabled, but abstracts over irrelevant details of individual solvers.
  We demonstrate the effectiveness of our technique by presenting a termination proof for a set theory axiomatisation adapted from those used in the Dafny and Viper verifiers.
\end{abstract}


\section{Introduction}
\label{sec:introduction}

SMT-based program analysis and verification have advanced dramatically in the past two decades.
These advances have been partly fuelled by major improvements in SAT and SMT solving techniques, as well as their implementations in \linebreak state-of-the-art solvers such as Z3~\cite{Z3} and cvc5~\cite{cvc5}.
Leveraging these advances in SMT, a huge number of program analysis and verification tools have been based on SMT, including for example Dafny~\cite{Leino2010-Dafny}, Why3~\cite{Filliatre2013-Why3} and Viper~\cite{Mueller2016-Viper}.

Such tools must translate a wide range of program features into SMT queries that model these domain-specific concerns.
While some theories relevant to problem features (\eg linear arithmetic~\cite{Z3}) are natively supported by SMT solvers, most problem features must be modelled by \emph{axiomatisation}.

Axiomatising problem features involves introducing uninterpreted sorts, uninterpreted functions on these sorts, and (crucially) \emph{quantified axioms} that define the intended meaning of these features.
For instance, one can model sets of integers by introducing a sort $\textit{Set}$ for sets, uninterpreted functions $\Member$ and $\Diff$ to represent set membership and set difference respectively, and quantified axioms such as
$\forall s_{1},s_{2}:\SSet,\,x:\textit{Int}.\;\Member(x,s_{2})\rightarrow\neg\Member(x,\Diff(s_{1},s_{2}))$.

Such modelling to SMT is expressive, but makes heavy use of quantifiers that must be instantiated during SMT solving.
But quantifier instantiation in SMT notoriously presents notable challenges, potentially causing slow performance and even non-termination, as well as unexpectedly-failing proofs \cite{Becker2019-Axiom-Profiler,Moskal2009-Programming-with-Triggers}. Worse still, latent quantifier instantiation issues may not surface on all runs, but cause a ``butterfly effect''~\cite{Leino2016-Trigger-Selection-Dafny}, meaning that unrelated changes to an input problem may lead to substantial changes in solver behaviour along these lines.

To manage these issues, solvers allow quantifiers to be annotated with instantiation \emph{triggers} (\aka{} instantiation \emph{patterns}).
Triggers specify (possibly multiple) shapes of ground terms that must be \emph{known} (occur in the current proof context, modulo known equalities) to enable a quantifier instantiation.
This method of guiding quantifier instantiation is referred to as \emph{E-matching}~\cite{Simplify,Nelson1980-PV} and is supported by virtually all modern SMT solvers.

However, selecting appropriate triggers is an art. 
The choice requires expertise in managing a fine balance: 
not too restrictive, to avoid insufficient quantifier instantiations, 
and not too permissive, to prevent excessive instantiations. 
Subtle issues can easily lead to the same hard-to-debug issues even for the most talented of SMT artists~\cite{Leino2016-Trigger-Selection-Dafny,Moskal2009-Programming-with-Triggers}, 
and even when successful it is unclear how one can \emph{know} that the chosen triggers are guaranteed to work in future.

The ideal aim is to achieve both instantiation completeness and instantiation termination.
\emph{Instantiation completeness} means that all necessary quantifier instantiations for a proof can be made by the solver.
\emph{Instantiation termination} means that the solver will never endlessly explore infinitely many quantifier instantiations.
In this paper, we focus on instantiation termination.\footnote{Instantiation termination can be trivially achieved by pathological trigger choices that prevent all instantitions (similar to proving a function terminating under a false precondition). However, such axiomatisations are not useful (or used) in practice.} 

Failures of instantiation termination stem from \emph{matching loops}:
the problematic scenario of a quantifier instantiation (possibly indirectly) leading to learning new terms that cause further instantiations of the same quantifier, leading to a potentially endless loop.
Matching loops \emph{can} cause non-termination, but (problematically, for debugging) may only do so on some runs (in case heuristics in the solver arrive at the necessary facts ``in time'').

Our paper enables proving that matching loops have been avoided altogether.
We present a high-level formal model of E-matching-based quantifier instantiation that suffices to prove \emph{once and for all} that a given set of trigger-annotated quantifiers, when combined with 
\emph{any possible} ground facts, guarantees instantiation termination, thereby ensuring the absence of matching loops. 
Our model is designed to be broadly applicable because it models the core E-matching rules common to most solvers, but abstracts over implementation details where individual solvers make different choices. 
Our model enables a new kind of termination proof, allowing axiomatisation users to independently construct these proofs and confidently pursue terminating responses to ground theory queries. 

Our main technical contributions are as follows:
\begin{enumerate}
\item
  We develop a formal model for reasoning about instantiation termination in E-matching-based axiomatisations.
  The model abstracts from solver implementation details but accounts for essential features for termination proofs.

\item We validate the practical utility of our formal model by using it to prove instantiation termination of a challenging set theory axiomatisation adapted from the cores of those used in the Dafny and Viper verifiers.

\item
  We outline a methodology for constructing instantiation termination proofs using our model.
  Our methodology involves classifying quantifiers according to certain characteristics, using these to incrementally define and refine a progress measure that eventually supports the whole axiomatisation.
\end{enumerate}

Our research draws inspiration from Dross et al.'s~\cite{Dross2016} prior formalisms for quantifier instantiation via E-matching.
To the best of our knowledge, their work represents the sole formal attempt in this space before ours.
However, we find their formalism incompatible with our goals: we elaborate on
this point in~\secref{related-work}.


\section{Problem Statement }
We begin with a basic grounding in E-matching, and use this to lay out the most important challenges a formal model needs to address to be useful in practice.

\subsection{Quantifier Instantiation via E-matching }
Quantifiers\footnote{We use the term \emph{quantifier} (also) as a synonym for quantified formula, and \emph{quantifier body} to refer to the subformula that falls within the scope of a quantifier.} are crucial for effectively modelling external problem features as
an SMT problem.
However, when determining whether such a first-order problem is satisfiable, an SMT solver must contend with quantifiers ranging over infinite sorts.
A successful proof will (and need) only involve finitely many instantiations of the quantifiers, but selecting these is in general undecidable.
Most solvers provide \emph{E-matching} as the main means of guiding instantiation.

E-matching requires each quantifier to be associated with instantiation \emph{triggers} (\aka{} instantiation \emph{patterns}). 
Triggers consist of terms containing the quantified variables, and prescribe that instantiations should only be made when ground terms of matching shape(s) arise in the current proof search.

During a proof search, SMT solvers maintain and update the currently-known ground terms and (dis)equalities on them in an efficient congruence-closure data structure called an \emph{E-graph}.
This information enables \emph{E-matching}~\cite{deMoura2007-Efficient-E-matching,Nelson1980-PV}---matching modulo currently-known equalities---of known terms against quantifier triggers,
which enables new instantiations, and of potential instantiations against previous ones, which prevents redundant instantiations.

\begin{example}
  \label{exa:example-original}Consider the set theory axiom presented early in \secref{introduction}, now annotated with triggers (written comma-separated inside square brackets)%
  \footnote{For brevity, sorts on quantified variables are omitted in this example and hereafter.}:
\[
\forall s_{1},s_{2},x.\left[\Diff(s_{1},s_{2}),\Member(x,s_{2})\right]\Member(x,s_{2})\rightarrow\neg\Member(x,\Diff(s_{1},s_{2}))
\]
The trigger consists of two terms, $\Diff(s_{1},s_{2})$ and $\Member(x,s_{2})$; a multi-term trigger prescribes that terms matching \emph{all} (here, both) patterns must be known for some instantiation of the quantified variables. If so, the corresponding instantiation of the quantifier \emph{itself} will be made: the instantiated quantifier body will be treated as a newly-derived fact (typically, a \emph{clause}), and the solver will also record that this instantiation has been made (to avoid doing so again).

Suppose that an E-graph represents the congruence closure of the facts: $\Member(t,a){=}\top$,
$\Diff(b,c){\neq}b$ and $a{=}c$. E-matching will find a successful match against the trigger above; although it might seem that there is no consistent pair of terms here, the equality $a=c$ means that (modulo equalities) we can consider the terms $\Member(t,a)$ and $\Diff(b,a)$ as known in the E-graph, which match the triggers under the instantiation $\subs{s_{1}}{b}$, $\subs{s_{2}}{a}$ and $\subs{x}{t}$. The corresponding instantiation of the quantifier body yields $\neg\Member(t,a)\vee\neg\Member(t,\Diff(b,a))$. Subsequently, the same quantifier cannot be instantiated with \eg{} $\subs{s_{1}}{b}$, $\subs{s_{2}}{c}$ and $\subs{x}{t}$ since, again modulo equalities, this is an equivalent instantiation.
\end{example}

\begin{example}
\label{exa:example-matching-loop}
Consider the same quantifier, with a different trigger, 
within the context of a different E-graph that represents the congruence closure of the facts: $\Member(t,a)=\top$ and $\Member(t,b)=\top$.
\[
\begin{array}{l}
\forall s_{1},s_{2},x.\left[\Member(x,s_{1}),\Member(x,s_{2})\right]\\
\quad\quad\quad \Member(x,s_{2})\rightarrow\neg\Member(x,\Diff(s_{1},s_{2}))
\end{array}
\]
Now four instantiations are enabled: one for each pair of $\Member$ applications in our current model (and E-graph): \eg{} instantiating $\subs{s_{1}}{a}$, $\subs{s_{2}}{b}$ and $\subs{x}{t}$ or $\subs{s_{1}}{b}$, $\subs{s_{2}}{a}$ and $\subs{x}{t}$. All four will be made: they are different choices since we don't know that $a=b$. The second, for example, causes the new clause (rewritten as a disjunction) $\neg\Member(t,a)\vee\neg\Member(t,\Diff(b,a))$ to be assumed. This doesn't change the E-graph (which is populated only by assumed \emph{literals}); clauses are kept separately in the prover state. However, case-splitting on this clause may lead to the literal $\neg\Member(t,\Diff(a,b))$ being added. At this point, five \emph{new} quantifier instantiations will be enabled; the number of pairs of $\Member$ applications has increased. In fact, by alternately instantiating this quantifier and case-splitting on newly-learned clauses, we can uncover new instantiations indefinitely, in a so-called \emph{matching loop}.
\end{example}

These first examples show that the choice of triggers affects instantiation behaviour, and that modelling instantiations requires considering not only initial terms, but also facts learned during proof search and case-splitting choices.

\begin{example}
\label{exa:example-nested-quantifier}Consider the following ``subset elimination'' axiom (also from the set theory we tackle later) with nested quantifiers:
\[
\begin{array}{l}
\forall s_{1},s_{2}.\left[\Subset(s_{1},s_{2})\right]\;\Subset(s_{1},s_{2})\rightarrow\\
\quad\quad\left(\forall x.[\Member(x,s_{1})][\Member(x,s_{2})]\;\Member(x,s_{1})\rightarrow\Member(x,s_{2})\right)
\end{array}
\]
The inner quantifier has \emph{two} triggers, defining \emph{alternative} conditions for instantiation (a term of either shape is sufficient). Note that these triggers depend on the outer-quantified variables $s_1$ and $s_2$, and thus their instantiations. 

Instantiating an outer quantifier expands the current quantifiers for instantiations. In this example, instantiating the outer quantifier ($\forall s_1, s_2.\dots$) results in a clause that includes a copy of the inner quantifier ($\forall x.\dots$); case-splitting on this clause can cause the copy to be assumed, effectively adding one more quantifier for future instantiations. As such, the instantiation of nested quantifiers \emph{dynamically} introduces new quantifiers, adding complexity to establishing termination arguments---one must be able to identify and predict the quantifiers that will be dynamically introduced. 


\end{example}

\subsection{Objectives for a Formal Model of E-Matching \label{subsec:objectives}}

Given the difficulty of choosing quantifier triggers and \emph{knowing} that their instantiations can \emph{never} continue forever, our objective is to provide formal and usable means of proving such E-matching \emph{termination proofs} once-and-for-all. Rather than attempt to capture the precise behaviour of a specific solver and its configuration, we want a model that abstracts over the behaviours of \emph{any} reasonable implementation of E-matching, while still being sufficiently precise for the proofs to work and be reasonable to construct in practice.

The design of a model for E-matching must address multiple challenges:
\begin{enumerate}
\item How should (intermediate) solver states and the transitions between them be modelled, avoiding over-fitting to specific solver choices while retaining clear and pertinent information suitable for understandable proofs?
\item How should equality-related information and reasoning be captured, given their central nature (for defining enabled matches) but the complexities of the data structures employed in real implementations?
\item How can nested quantifiers (\cf{} Example \ref{exa:example-nested-quantifier}), whose instantiation can introduce new quantifiers on the fly, be supported?
\item How can we make the model extensible to more-complex future applications (\eg{} axiomatisations whose termination depends on theory reasoning)? 
\item How can a formal model enable formal proofs with manageable complexity?
\end{enumerate}

We present our model, designed to address these challenges in the next section; we demonstrate its applicability for termination proofs in \secref{proving}.

\section{An Operational Semantics for E-matching }
We develop our formal model in the style of a \emph{small-step operational
  semantics}, a popular choice for programming languages. In this
operational style, states represent intermediate points of a proof search,
while transitions represent solver steps; non-determinism abstracts over
choices specific solvers make. With this design, our desired notion of instantiation termination can be recast as a familiar style of termination proof, albeit against a semantics with novel core details.

\global\long\def\multi#1{\overrightarrow{#1}}%

\global\long\def\unitriMulti#1#2#3{\forall\multi{#1}.\multi{\left[#2\right]}#3}%

\global\long\def\unitriFlat#1#2#3{\forall#1.\left[#2\right]#3}%

\global\long\def\Tag{\sharp}%

\global\long\def\taggen#1{\widehat{\mathrm{tag}_{\Tag}}\left(#1\right)}%

\global\long\def\taggenOne#1{\mathrm{tag}_{\Tag}\left(#1\right)}%

\global\long\def\taggenOneVar#1#2{\mathrm{tag}_{\Tag}^{#2}\left(#1\right)}%

\global\long\def\filter{\mathrm{filter}}%

\global\long\def\filterUniOp{\filter_{\forall}}%

\global\long\def\filterUni#1{\filterUniOp\left(#1\right)}%

\global\long\def\filterLit#1{\filterLitOp\left(#1\right)}%

\global\long\def\filterLitOp{\filter_{\mathrm{lit}}}%

\global\long\def\state#1#2#3{\left\langle #1,#2,#3\right\rangle }%

\global\long\def\egraph#1{#1^{\mathrm{I}}}%

\global\long\def\ehistory#1{#1^{\mathrm{H}}}%

\global\long\def\tran{\longrightarrow}%

\global\long\def\otran{\tran_{\mathrm{\vee}}}%

\global\long\def\ptran{\tran_{\mathrm{\forall}}}%

\global\long\def\updateEgraph#1#2{#1\triangleleft#2}%

\global\long\def\Egraph#1{#1^{\mathrm{I}}}%

\global\long\def\Ehistory#1{#1^{\mathrm{H}}}%

\global\long\def\Entail#1#2{#1\Vdash#2}%

\global\long\def\NotEntail#1#2{#1\not\Vdash#2}%

\global\long\def\Known#1#2{#1\Vdash_{\mathrm{known}}#2}%

\global\long\def\Class#1#2#3{#1\Vdash_{\mathrm{class}}#2\in#3}%

\global\long\def\Equiv#1#2{#1\sim#2}%

\global\long\def\NotEquiv#1#2{#1\not\sim#2}%

\global\long\def\Amount#1{\left\Vert #1\right\Vert }%

\global\long\def\AmountClass#1{\left\Vert #1\right\Vert _{\mathrm{class}}}%

\global\long\def\AmountLeft#1{\left\Vert #1\right.}%

\global\long\def\AmountRight#1{\left.#1\right\Vert }%

\global\long\def\Inst#1#2{#1\Vdash_{\mathrm{inst}}#2}%

\global\long\def\NotInst#1#2{#1\not\Vdash_{\mathrm{inst}}#2}%

\global\long\def\Select{\vdash_{\mathrm{select}}}%

\global\long\def\Ematching{\vdash_{\mathrm{match}}}%

\global\long\def\MatchSep{\sphericalangle}%

\subsection{Preliminaries \label{subsec:Preliminaries}}

Our syntax for formulas is based around a generalisation of conjunctive normal form, used internally in SMT algorithms; we assume all formulas are pre-converted to this form (existential quantifiers are eliminated by Skolemisation).
\begin{definition}[Formula Syntax]
We assume a pre-defined set of \emph{atoms}\footnote{The pre-defined atoms come from the first-order signature of the problem in question.}, including equalities on terms $t_{1}=t_{2}$. A \emph{literal} $l$ is either an atom or its negation. The grammars of \emph{extended literals} $\phi$, \emph{extended clauses} $C$ and \emph{extended conjunctive normal form
(ECNF) formulas} $A$ are as follows: \[
\begin{array}[t]{rrlrrlrrl}
\phi & \Coloneqq & l \mid (\unitriMulti xTA)^{\Tag\alpha} & \quad\quad C & \Coloneqq & \phi \mid C\vee C & \quad\quad A & \Coloneqq & C \mid A\wedge A
\end{array}
\]
Here, $(\unitriMulti xTA)^{\Tag\alpha}$ denotes a \emph{tagged quantifier}: the (possibly-multiple) variables $\multi x$ are bound, the (possibly multiple) trigger sets $\multi T$ are each marked with square brackets and positioned before the quantifier body $A$, and $\Tag\alpha$ is a \emph{tag} used to identify this particular quantifier.
\end{definition}

As presented in Example \ref{exa:example-original}, a trigger set $T$ is a (non-empty) set of terms, written comma-separated. There are additional requirements: each trigger set must contain each quantified variable at least once, and each term must contain at least one quantified variable. Furthermore, each term must contain at least one uninterpreted function application and no interpreted function symbols such as equalities. These restrictions are common for SMT solvers.

When quantifier tags are not relevant, we omit them for brevity.

\subsection{States \label{subsec:States}}

As illustrated in Examples \ref{exa:example-original} and \ref{exa:example-matching-loop}, both case-splitting and quantifier instantiation steps are crucial to our problem; we define our semantics around these two kinds of transitions. Furthermore, we must abstractly capture information relevant for deciding E-matching questions, tracking in particular which terms and equalities are known (modulo currently known equalities), and which quantifier instantiations have already been made.
\begin{definition}[States]\label{def:states}
States $s\in\textsc{State}$ are defined as follows:
\[
s\Coloneqq\state WAE\mid\lozenge\mid\bot
\]
where $\lozenge$ and $\bot$ are distinguished symbols for \emph{saturated} and \emph{inconsistent} states, $W$ (the \emph{current quantifiers}) is a set of tagged quantifiers, $A$ (the \emph{current clauses}) is a set of extended clauses, and $E$ (the \emph{current E-state}) is explained below.
\end{definition}
For simple applications of our semantics, the set of current quantifiers remains fixed, but for problems with nested quantifiers (\eg{} Example \ref{exa:example-nested-quantifier}), it may grow as a solver runs. As we show, which instantiations are immediately enabled is definable in terms of both the current quantifiers and the current E-state. The current clauses, on the other hand,  generate new literals for the E-state via case-splitting; new extended \emph{clauses} may be added as a consequence of quantifier instantiations.

The inconsistent and saturated states represent two different termination conditions for traces in our semantics: the former due to logical inconsistency, and the latter due to all quantifier instantiations having been exhausted.

\subsection{E-interfaces \label{subsec:E-interfaces}}

Each solver maintains its own implementation of E-graphs to efficiently represent and query the currently-known ground terms modulo congruences and known equalities. Rather than formalising such an implementation,
we devise an abstraction called an \emph{E-interface}, capturing the operations and expected mathematical properties of E-graph implementations.

\begin{definition}[E-interface Judgements] \label{def:E-interface-judgement}An E-interface $\Egraph E$ is a set of equalities and disequalities on terms.\footnote{A positive or negative non-equational literal, $P$, is added to the E-interface via $P=\top$ or $P=\bot$, respectively; $\top \neq \bot$ is preloaded into all E-interfaces.} We write $\Known{\Egraph E}t$ to express that the ground term $t$ is \emph{known} in the E-interface $\Egraph E$; we write $\Entail{\Egraph E}{\Equiv{t_{1}}{t_{2}}}$ to express that the ground terms $t_1$ and $t_2$ are \emph{known equal} in $\Egraph E$. These two judgements are (mutually recursively) defined by (the least fixed-point of) the derivation rules:
\begin{center}
$\dfrac{\Equiv{t_{1}}{t_{2}}\in\Egraph E}{\Entail{\Egraph E}{\Equiv{t_{1}}{t_{2}}}}\textsc{(eq-in)}$\qquad{}$\dfrac{\Entail{\Egraph E}{\Equiv{t_{2}}{t_{1}}}}{\Entail{\Egraph E}{\Equiv{t_{1}}{t_{2}}}}\textsc{(eq-sym)}$\qquad{}$\dfrac{\begin{array}{cc}
\Entail{\Egraph E}{\Equiv{t_{1}}{t_{2}}}\; & \;\Entail{\Egraph E}{\Equiv{t_{2}}{t_{3}}}\end{array}}{\Entail{\Egraph E}{\Equiv{t_{1}}{t_{3}}}}\textsc{(eq-tran)}$\qquad{}$\dfrac{\Known{\Egraph E}t}{\Entail{\Egraph E}{\Equiv tt}}\textsc{(eq-kn-refl)}$\qquad{}$\dfrac{\begin{array}{cc}
\Entail{\Egraph E}{\Equiv{t_{i}}{t_{i}^{\prime}}}\; & \;\Known{\Egraph E}{g\left(t_{1},\dots,t_{i},\dots,t_{n}\right)}\end{array}}{\Entail{\Egraph E}{\Equiv{g\left(t_{1},\dots,t_{i},\dots,t_{n}\right)}{g\left(t_{1},\dots,t_{i}^{\prime},\dots,t_{n}\right)}}}\textsc{(eq-kn-sub)}$\qquad{}$\dfrac{\Entail{\Egraph E}{\Equiv{t_{1}}{t_{2}}}}{\Known{\Egraph E}{t_{1}}}\textsc{(kn-eq)}$\qquad{}$\dfrac{\Known{\Egraph E}{g\left(\dots,t_{i},\dots\right)}}{\Known{\Egraph E}{t_{i}}}\textsc{(kn-sub)}$
\par\end{center}
The judgement $\Entail{\Egraph E}{\NotEquiv{t_{1}}{t_{2}}}$ represents $t_1$ and $t_2$ being \emph{known disequal} in $\Egraph E$; the judgement $\Entail{\Egraph E}{\bot}$ represents that $\Egraph E$ is \emph{inconsistent} (in the logical sense); \cf{} Appx. \ref{sec:appx-formal-model-for-e-matching}.
\end{definition}
E-interfaces are equivalent if they agree on these judgements in all cases. When a proof step adds new literals, we must be able to extend our E-interfaces.

\begin{definition}[E-interface Extension]
For a set of equality and disequality literals $L$, the \emph{update of an E-interface $\Egraph E$ with $L$}, denoted  $\updateEgraph{\Egraph E}{L}$, is a minimal E-interface which satisfies all E-interface judgements that $\Egraph E$ does, while also satisfying $\Entail{\Egraph E}{l}$ for all $l\in L$.
\end{definition}

We call a set of terms a \emph{basis} of $\Egraph E$ if each element is a representative of a different equivalence class\footnote{What we refer to as an \emph{equivalence class} in this paper is known as a \emph{congruence class} in the literature.} induced by the $\Entail{\Egraph E}{\Equiv{t_{1}}{t_{2}}}$ relation on the terms known in $\Egraph E$. As we shall see in the next section, equivalence classes are relevant for defining which quantifier instantiations can be made after which.

\subsection{E-histories, E-states, E-matching  \label{subsec:E-history,-enabled-matches}}

As illustrated in Example \ref{exa:example-original},
E-matching against triggers does not suffice to determine whether a quantifier instantiation should be considered \emph{enabled}; we must also determine whether the instantiation is considered redundant given \emph{previous} ones. We record previous instantiations using our next formal ingredient:

\begin{definition}[E-histories and E-states]
An \emph{E-history} $\Ehistory E$ is a set of pairs (each denoted $(\Tag\alpha:\multi r)$) in our formalism: the first element is a tag (identifying a quantifier), and the second is a vector of ground terms (representing an instantiation of the corresponding quantifier).

An \emph{E-state} (\cf{} \defref{states}) $E$ is a pair $(\Egraph E, \Ehistory E)$ of E-interface and E-history.
\end{definition}

Recall that E-states are a component of the states in our formalism. The E-interface captures the current known terms and equality information, while the E-history represents sufficient information to reject redundant instantiations.

\begin{definition}[History-Enabled E-matches]\label{def:history-enabled-ematches}
Given a candidate pair $\left(\Tag\alpha:\multi r\right)$ 
(of tag $\Tag\alpha$ and vector of terms $\multi r$), the \emph{E-state $E$ enables $\left(\Tag\alpha:\multi r\right)$}, written $\Inst E{\left(\Tag\alpha:\multi r\right)}$, if: for every pair $(\Tag\alpha:\multi{r^{\prime}})\in\Ehistory E$, at least one of the pointwise equalities $\multi {\Equiv{r_i}{r_i^\prime}}$ is \emph{not} known in $\Egraph E$.
\end{definition}

\begin{example}
Revisiting Example \ref{exa:example-original}, suppose the tag
of the quantifier is $\Tag\tau$ and $E$ is the E-state containing the example literals. The first instantiation $\subs{s_{1}}{b}$, $\subs{s_{2}}{a}$ and $\subs{x}{t}$ is represented in our formal model by adding $\left(\Tag\tau:\left(b,a,t\right)\right)$ to the E-history, resulting in a new E-state, say $E'$. The second candidate match $\subs{s_{1}}{b}$, $\subs{s_{2}}{c}$ and $\subs{x}{t}$ is not enabled in $E'$ since the three pointwise equalities between instantiated terms are all known in $E'$.
\end{example}

With the help of the above ingredients, we formally characterise E-matching:  

\begin{definition}[E-matching]
For a given state $\state WAE$, 
the judgement \linebreak $\state WAE\Ematching (\unitriMulti xT{A^{\prime}})^{\Tag\alpha}\MatchSep\multi r$ defines
which instantiations (using terms $\multi r$) of which quantifiers $(\unitriMulti xT{A^{\prime}})^{\Tag\alpha}$ are enabled by E-matching rules, as follows: 
\[
\dfrac{\begin{array}{c}
\begin{array}{cc}
(\unitriMulti xT{A^{\prime}})^{\Tag\alpha}\in W \quad & \quad \multi t \text{ is one trigger set of } \multi {[T]} \\
\Known{\Egraph E}{\multi t\left[\nicefrac{\multi r}{\multi x}\right]} \quad & \quad \Inst E{\left(\Tag\alpha:\multi r\right)} 
\end{array}
\end{array}}{\state WAE\Ematching(\unitriMulti xT{A^{\prime}})^{\Tag\alpha}\MatchSep\multi r}
\]
We write $\state WAE\not\Ematching$ to mean \emph{no} instantiations are enabled in this state.
\end{definition}
E-matching $\Ematching$ requires (1) a quantifier in the current state, 
(2) the trigger set $\multi t$ with replacement terms $\multi r$ for quantified variables $\multi x$ to be known in $\Egraph E$, 
and (3) that this potential match is enabled in the E-state $E$. 
Note that (2) implies the terms $\multi r$ to match against the quantified variables of one trigger set $\multi t$ to be known in the current E-interface $\Egraph E$.

\subsection{State Transitions \label{subsec:State-transit}}

The last main ingredient of our formal model is the definition of state transitions. 

\global\long\def\Verify#1#2{#1\Vdash_{\mathrm{verify}}#2}%

\global\long\def\NotVerify#1#2{#1\not\Vdash_{\mathrm{verify}}#2}%

\global\long\def\updateW#1#2{#1\cup#2}%

\begin{definition}
[State Transitions] The \emph{(single step) state transition relation} $\tran\,\subseteq\textsc{State}\times\textsc{State}$ is defined by the union of the following cases: 
\[
\dfrac{\begin{array}{c}
\emptyset\subset\Phi\subseteq\left\{ \phi_{i}\mid C\in A;\;\NotVerify{W_{1},\Egraph{E_{1}}}C;\;C\text{ is }\cdots\vee\phi_{i}\vee\cdots\right\} \\
W_{2}={W_{1}}\cup{\filterUni{\Phi}}
\quad\Egraph{E_{2}}=\updateEgraph{\Egraph{E_{1}}}{\filterLit{\Phi}}\quad\Ehistory{E_{2}}=\ehistory{E_{1}}
\end{array}}{\state{W_{1}}A{E_{1}}\tran\state{W_{2}}A{E_{2}}}\textsc{(split)}
\]
\[
\dfrac{\Entail{\Egraph E}{\bot}}{\state WAE\tran\bot}\textsc{(bot)}
\]
\[
\dfrac{\begin{array}{ccc}
\NotEntail{\Egraph E}{\bot}\quad & \quad\Verify{W,\Egraph E}C\text{ for every }C\in A\quad & \quad\state WAE\not\Ematching\end{array}}{\state WAE\tran\lozenge}\textsc{(sat)}
\]
\[
\dfrac{\begin{array}{c}
\state{W_{1}}{A_{1}}{E_{1}}\Ematching(\unitriMulti xT{A_{11}})^{\Tag\alpha}\MatchSep\multi r\\
\begin{array}{cc}
A_{12}=A_{11}\left[\nicefrac{\multi r}{\multi x}\right] & A_{12}^{\prime}=\filterUni{A_{12}}\cup\filterLit{A_{12}}\\
A_{2}=A_{1}\cup\left(A_{12}\backslash A_{12}^{\prime}\right) & W_{2}=\updateW{W_{1}}{\filterUni{A_{12}}}{\Egraph{E_{1}}}\\
\Egraph{E_{2}}=\Egraph{E_{1}}\triangleleft\filterLit{A_{12}} & \Ehistory{E_{2}}=\Ehistory{E_{1}}\triangleleft\left(\Tag\alpha:\multi r\right)
\end{array}
\end{array}}{\state{W_{1}}{A_{1}}{E_{1}}\tran\state{W_{2}}{A_{2}}{E_{2}}}\textsc{(inst)}
\]
where $\filterUniOp$ and $\filterLitOp$ filter sets of extended literals into only those which are quantifiers or only those which are simple literals, respectively; 
the judgement $\Verify{W,\Egraph E}{C}$ holds if: for some disjunct $\phi_{i}$ of $C$, either $\phi_{i}$ is a tagged quantifier from $W$, or $\phi_{i}$ is a literal that $\Egraph E$ knows. 
\end{definition}

Our transition relation $\tran$ consists of case-splitting steps, steps that deduce inconsistent states, steps that deduce saturated states, and quantifier instantiation steps, corresponding to rules \textsc{(split)}, \textsc{(bot)}, \textsc{(sat)} and \textsc{(inst)} respectively. 

We allow a case-splitting transition to non-deterministically select \emph{any} non-empty subset of the disjuncts in the \emph{unverified} current clauses---those that have not been proved true yet.
A case-splitting transition must make progress towards satisfying the clauses. 
We do not impose restrictions on the order in which unverified clauses are chosen, nor on the number of disjuncts assumed within a clause, provided that progress is being made.\footnote{Our model  allows simulating efficient propagation-based restrictions of case-splitting, but does not require it; restricting to this case would be possible if needed.}

We model case-splitting as non-deterministic. Recall Example
\ref{exa:example-matching-loop}, where the clause $\neg\Member(t,a)\vee\neg\Member(t,\Diff(a,b))$ is learnt. Subsequently, the solver can choose to assume
either one or both of the disjuncts; generally, it can choose to assume neither disjunct as long as it selects at least one disjunct from some other unsatisfied clause. 
Here, the disjuncts are ground terms (which are added to the E-state); in general, some could be new quantifiers to record.

Our $\Verify{}{}$ judgement checks if a provided clause is satisfied (\ie at least one disjunct is assumed in the current state). 
If all current clauses are satisfied, and the E-interface is not inconsistent, and there are no candidate E-matches, the \textsc{(sat)} rule applies and transitions to the saturated state ($\lozenge$). Conversely, if the current E-interface is inconsistent, the \textsc{(bot)} rule transitions to the inconsistent state ($\bot$); if there are candidate E-matches, the \textsc{(inst)} rule applies. 

The instantiation rule \textsc{(inst)} relies on the $\Ematching$ judgement to select a candidate E-match. The effect of an instantiation transition involves adding quantifiers and literals occurring as unit clauses in the quantifier body to the current quantifiers $W_1$ and E-interface $\Egraph{E_1}$, respectively; 
any remaining non-unit clauses are added to the current clauses $A_{1}$.
Finally, the E-history $\Ehistory {E_1}$ is updated to record this instantiation. 

In practice, common SMT solvers such as cvc5 \cite{cvc5} perform quantifier instantiation both (1) up-front and (2) in phases interleaved with other solver steps. In particular, the latter is essential for many applications: most quantifier instantiations lead to e.g. clauses requiring context-aware case-splitting via DPLL/CDCL. Our model effectively capture both processes through its unrestricted interleavings of quantifier instantiation and case-splitting steps.

In retrospect, Sec. \ref{subsec:States} to \ref{subsec:State-transit} have tackled design challenges \#1 and \#2 (\cf{} Sec. \ref{subsec:objectives}). We address \#3 and \#4 in the next two subsections, respectively.  

\subsection{Nested Quantifiers \label{subsec:Nested-quantifiers}}

Example \ref{exa:example-nested-quantifier} demonstrates that instantiating nested quantifiers can introduce new quantifiers on the fly. 
To effectively argue for termination regarding these instantiations (as will be discussed in Sec. \ref{sec:proving}), one must be able to identify and predict these dynamically introduced quantifiers.  
To facilitate this, we employ a tagging system that is capable of handling nested structures (\cf{} Appx. \ref{sec:appx-formal-model-for-e-matching} for details). 
Each quantifier in an axiomatisation is labelled with a distinct tag. 
The tag for any non-nested quantifier or the outermost quantifier of any nested quantifier is not parameterised. 
An inner quantifier that occurs in a nested quantifier has its tag parameterised by all of its outer-quantified variables. 
Instantiating an outer quantifier produces a copy of the quantifier body in which (among other changes) tags of all inner quantifiers that are parameterised by this outer-quantifier are updated to reflect this instantiation.  In Example \ref{exa:example-nested-quantifier}, we label the outer and inner quantifiers with tags  
$\Tag \text{union-elim}$ and $\Tag \text{union-elim}(s_1, s_2)$, respectively. 
When the outer quantifier is instantiated with $\subs{s_{1}}{a}$ and $\subs{s_{2}}{b}$, 
a copy of the quantifier body in which the inner quantifier is tagged with $\Tag \text{union-elim}(a, b)$ is introduced. 

To further mitigate redundancy in quantifier instantiation, our semantics supports two additional optimisations. 
First, a quantifier is only permitted to join the current quantifiers $W$ if its tag is known to be \emph{distinct} from the tags of existing quantifiers in $W$, \emph{modulo equivalence on the parameters of the tags}, as assessed in the current E-interface. 
This criterion prevents adding redundant quantifiers into $W$. 
Second, the relation of history-enabled E-matches $\Inst {} {}$ leverages the current E-interface to verify the uniqueness of tags---once again, modulo equivalence on tag parameters---before enabling an E-match.  
An E-match is enabled only if no quantifier with an equivalent tag has been instantiated with an equivalent match previously (\cf{} Appx. \ref{sec:appx-formal-model-for-e-matching} for related definitions).

\subsection{Theory-Specific Reasoning \label{subsec:Theory-specific-reasoning}}

Although our rules do not yet account for (interpreted) theory reasoning (as performed by theory solvers in a typical SMT solver design). Our small-step semantics is intentionally chosen to easily accommodate future extensions: ``hot-plugging'' new kinds of primitive transitions is straightforward, and will not disturb the existing formal rules (\eg{} for quantifier instantiations or case-splitting). Similarly to our E-interfaces for abstracting of E-graph details, we plan to do this in a way which abstracts over the \emph{effects} of theory deduction steps, without exposing the solver-specific internals. For example, we can add deduction steps which extend the E-interface with new terms and/or (dis)equalities, based on a valid deduction within, say, an integer theory.

Just as for quantifier instantiations, it may be necessary for some applications to guarantee that theory reasoning is performed under some fairness conditions (\eg{} that inconsistencies detectable by a theory solver are not infinitely postponed). Imposing custom fairness constraints on the traces of our semantics for specific examples can be achieved in a standard way for small-step semantics.

While it is clear that extensions to theory solving will be straightforward, we chose the case study for this paper to be a complex and practically-relevant axiomatisation which nonetheless does not rely on external theory solvers.

\section{Proving Instantiation Termination for E-matching}\label{sec:proving}

We now apply our model to prove instantiation termination for a practical E-matching-based axiomatisation.  First, we briefly present our set theory axiomatisation, adapted from Dafny and Viper.  We then demonstrate our methodology for constructing instantiation termination proofs using our model.

\subsection{Axiomatisation for Set Theory }

To assess our formal model, we tackle formal proofs of instantiation termination for axiomatisations currently employed by state-of-the-art verification tools, specifically targeting set theory in this paper. 
Set theory, despite the known challenges associated with its quantifier instantiation, is extensively used in verifiers. 

Drawing from the axioms used by Dafny~\cite{Dafny2014-Set} and Viper~\cite{Viper2021-Set}, we aim to construct an axiomatisation that (1) faithfully models the core of set theory, (2) supports various encodings of set theory used by verifiers, and (3) strives to maintain a balance on triggers to ensure instantiation termination without harming instantiation completeness. 

Our axiomatisation involves 12 uninterpreted functions, representing a broader range of set operations than our counterparts of Dafny and Viper. Cardinality constraints are entirely removed due to their dependency on external linear arithmetic solvers (\cf{} Sec. \ref{subsec:Theory-specific-reasoning} for explanation). Refer to Appx. \ref{subsec:appx-our-axiomatisation-set-theory} and \ref{subsec:appx-comparison-axiomatisations-set-theory} for a full presentation of our axiomatisation and comparison with theirs.  

Dafny and Viper typically use complex ``iff'' statements to define set operations, restricting trigger flexibility as they must apply in both directions of ``iff''. Inspired by proof systems for formal logic, we redefine set operations using duals of introduction and elimination axioms, allowing for independent triggers for each axiom of the same set operation, thereby enhancing trigger flexibility.

\begin{example} \label{exa:example-union-elim}
Below is our elimination rule for set union, named (union-elim), allowing for more alternative triggers than the counterpart from Dafny and Viper.
\begin{center}
$
\begin{array}{l}
\forall s_{1},s_{2},x.\left[\Member(x,\Union(s_{1},s_{2}))\right]\\
\left[\Union(s_{1},s_{2}),\Member(x,s_{1})\right]
\left[\Union(s_{1},s_{2}),\Member(x,s_{2})\right]\\
\;\Member\left(x,\Union\left(s_{1},s_{2}\right)\right)\rightarrow\Member(x,s_{1})\vee\Member(x,s_{2})
\end{array}
$
\end{center}

Our axiomatisation overall has more permissive triggers,
which provides more flexibility for instantiation, but also increases the risk of non-termination.
That instantiation termination holds for our axiomatisation means that
Dafny and Viper's  more restrictive triggers are not necessary to ensure termination.
\end{example}

\subsection{Progress Measure}

\global\long\def\filterT{\filter_{T}}%

\global\long\def\filterSet{\filter_{Set(T)}}%


\global\long\def\FilterT{O_{2}}%

\global\long\def\FilterSet{O_{1}}%

To prove \emph{instantiation termination} for an axiomatisation, it suffices to prove that querying \emph{any} set of ground literals on the axiomatisation cannot lead to an infinite trace in our formal semantics. The proof argument is parametric with respect to the ground literals\footnote{In fact, the termination argument can be generalised to the ground \emph{clauses} in the initial state.} in the initial state. Drawing inspiration from program reasoning~\cite{Cook2011-CACM-Termination,Turing1949-Checking-a-large-routine}, we identify a suitable measure on solver states and then establish its decrease at appropriate steps in a well-founded manner. 

This method leverages the specific features of the axioms under consideration. 
We analyse our set theory axioms and classify them by two criteria: (1) Would instantiating the axiom potentially generate new equivalence class of terms, \ie new terms modulo equalities? (2) Does the axiom have nested quantifiers?


\paragraph{Non-generative quantifiers.}

We call a quantifier \emph{non-generative} if its instantiations yield neither new quantifiers nor new equivalence classes of terms. 
The majority of our set theory axioms are non-generative. 

For instance, the (union-elim) axiom from Example \ref{exa:example-union-elim}, when instantiated with $\subs{s_{1}}{a}$, $\subs{s_{2}}{b}$ and $\subs{x}{t}$, yields $\neg\Member\left(t,\Union\left(a,b\right)\right) \vee \Member(t,a) \vee \Member(t,b)$, without introducing new equivalence classes of terms. This is because all of $t$, $a$, $b$ and $\Union(a,b)$ are subterms of the matched trigger and hence known. $\BBool$-sorted terms never add new equivalence classes (\cf{} \defref{E-interface-judgement}).

Instantiating a non-generative axiom reduces the amount of enabled E-matches by at least one because, on the one hand, history-enabled E-matches (\cf{} Def. \ref{def:history-enabled-ematches}) prevent instantiating the same quantifier with equivalent matches; on the other hand, instantiating a non-generative axiom does not introduce new quantifiers or new equivalence classes, thereby not expanding the match pool. 
This suggests: 

\begin{idea}
Define the progress measure to be about the amount of enabled E-matches.\label{idea-enabled-e-matches}
\end{idea}

\paragraph{Generative quantifiers.}

A quantifier is \emph{generative} if its instantiations may introduce new equivalence classes of terms. Four of our set theory axioms are generative: each may create new applications of Skolem functions on instantiation. For instance, the following (subset-intro) axiom, when instantiated, may create a new term $\Skolem_{\textit{ss}}(s_{1},s_{2})$ for some sets $s_1$ and $s_2$:
\begin{center}
$
\begin{array}{r}
\forall s_{1},s_{2}.\left[\Subset(s_{1},s_{2})\right]
\left(\Subset(s_{1},s_{2})\vee\Member(\Skolem_{\textit{ss}}(s_{1},s_{2}),s_{1})\right)\wedge\\
\left(\Subset(s_{1},s_{2})\vee\neg\Member(\Skolem_{\textit{ss}}(s_{1},s_{2}),s_{2})\right)
\end{array}
$
\end{center}
Similarly, axioms for introducing extensional quality on sets, establishing set disjointness, and introducing a predicate to check if a provided set is empty---namely (equal-sets-intro), (disjoint-intro), and (isEmpty-intro-1), respectively---can each produce new applications of Skolem functions: $\Skolem_{\textit{eq}}(s_{1},s_{2})$, $\Skolem_{\textit{dj}}(s_{1},s_{2})$, and $\Skolem_{ie}(s)$, respectively (\cf Appx.~\ref{subsec:appx-our-axiomatisation-set-theory} for details).

Generative axioms, by introducing new equivalence classes of terms, may expand the pool of E-matches, including those enabled. 
We thereby suggest: 
\begin{idea}
Predict new equivalence classes of terms introduced by instantiating generative axioms; incorporate these forecasts 
to estimate enabled E-matches. 
\label{idea-forecast-terms}
\end{idea}

\paragraph{Nested quantifiers. }

The third category, nested quantifiers, consists of axioms with nested quantifiers. This category includes three axioms, namely (subset-elim) from Example \ref{exa:example-nested-quantifier}, an axiom named (disjoint-elim) for eliminating set disjointness, and an axiom named (isEmpty-elim-1) for eliminating the predicate that checks if a provided set is empty (\cf Appx. \ref{subsec:appx-our-axiomatisation-set-theory} for their definitions). 

Although nested quantifiers do not introduce new equivalence classes of ground terms, their instantiations can create new quantifiers.
These new quantifiers can each have their own set of enabled E-matches, effectively raising the total amount of enabled E-matches.
To tackle this issue, we suggest: 

\begin{idea}
Anticipate quantifiers that could emerge from instantiating nested quantifiers; include these forecasts to refine the estimation of enabled E-matches. 
\label{idea-forecast-axioms}
\end{idea}

In practice, provided that these ideas are respected, one can often define simpler termination measures via \emph{over-approximations} of these candidate instantiations (provided this over-approximation remains finite and decreasing).

\paragraph{Formalising a practical progress measure.}

A basis of an E-interface is a representation of the known equivalence classes. We define its overapproximation to include potential new equivalence classes introduced by generative axioms.

\begin{definition}
[Overapproximation of Basis for Set Theory] Suppose $B$ is a basis of an E-interface.
The functions $\FilterSet(B)$ and $\FilterT(B)$ denote overapproximations
for the $\SSet(T)$-sorted and $T$-sorted elements within basis $B$,
respectively, to accommodate new expected equivalence classes of terms. 
\begin{align*}
\FilterSet(B) & =\filterSet(B)\\
\FilterT(B) & =\filterT(B)\cup\widehat{\Skolem_{\textit{ss}}}(\FilterSet(B),\FilterSet(B))\cup\widehat{\Skolem_{\textit{eq}}}(\FilterSet(B),\FilterSet(B))\\
 & \quad\cup\widehat{\Skolem_{\textit{dj}}}(\FilterSet(B),\FilterSet(B))\cup\widehat{\Skolem_{\textit{ie}}}(\FilterSet(B))
\end{align*}
Here $\filterSet$ and $\filterT$ take a basis and select its $\SSet(T)$-sorted
and $T$-sorted elements, respectively; each $\widehat{\Skolem}$ is
lifted from the corresponding $\Skolem$ to support sets. 
\end{definition}
The potential new terms introduced by generative axioms are all $T$-sorted Skolem terms. Thus predictions are solely performed by $\FilterT(B)$, not by $\FilterSet(B)$. 

Note that the results of these two overapproximations are guaranteed to be finite. E-interface bases always remain finite: elements are added (at most) for the new terms introduced in a step. Since our construction filters and \eg{} maps Skolem functions over these finite sets, its results are finite.
Leveraging this overapproximation of equivalence classes, we estimate enabled E-matches.

\begin{definition}[Overestimation of Enabled E-matches for Set Theory]
Consider an arbitrary state $\state WAE$. Let $B$ be a basis of the E-interface $\Egraph{E}$. Define an overestimation of the enabled E-matches for $s$ from $B$ as follows:
\begin{center}
    $P(\state WAE, B)=\{ \dots p_{\Tag\tau_{i}},\dots,p_{\Tag\tau_{j}(\multi r)},\dots\}$
\end{center}
where $p_{\Tag\tau_{i}}$ and $p_{\Tag\tau_{j}(\multi r)}$ each denote a set of tuples that overapproximate the enabled E-matches from the basis $B$ to the quantifiers with tags $\Tag\tau_{i}$ and $\Tag\tau_{j}(\multi r)$, respectively; each tag $\Tag\tau_{i}$ identifies an original axiom within $W$, and each $\Tag\tau_{j}(\multi r)$ identifies a quantifier introduced by instantiating the (outer) quantifier of an original axiom $\Tag\tau_{j}$ with terms $\multi r$ from the approximations $\FilterSet(B)$ or $\FilterT(B)$. 

To clarify, examples for each category are presented as follows; the remaining quantifiers shall adhere to the same pattern. 
\begin{itemize}
\item A non-generative axiom:

$p_{\Tag\text{union-elim}}=\left\{  
\left(s_{1},s_{2},x\right)
\;
\middle\vert\;\begin{array}{l}
s_{1},s_{2}\in\FilterSet(B), \,
x\in\FilterT(B),\\
\Inst E{\left(\Tag\text{union-elim}:(s_{1},s_{2},x)\right)}
\end{array}\right\} $

\item A generative axiom:

$p_{\Tag\text{subset-intro}}=\left\{  
\left(s_{1},s_{2}\right)
\;
\middle\vert\;
s_{1},s_{2}\in\FilterSet(B);\;
\Inst E{\left(\Tag\text{subset-intro}:\left(s_{1},s_{2}\right)\right)}
\right\} $

\item A nested quantifier:

$p_{\Tag\text{subset-elim}}=\left\{  
\left(s_{1},s_{2}\right)
\;
\middle\vert\;
s_{1},s_{2}\in\FilterSet(B);\;
\Inst E{\left(\Tag\text{subset-elim}:\left(s_{1},s_{2}\right)\right)}
\right\} $

\item A quantifier introduced by instantiating a nested quantifier: 

$p_{\Tag\text{subset-elim(\ensuremath{a,b})}}=\left\{ x\;\middle\vert\;
x\in\FilterT(B);\;
\Inst E{\left(\Tag\text{subset-elim(\ensuremath{a,b})}:x\right)}
\right\}$

where $a,b\in\FilterSet(B)$. 
\end{itemize}
\end{definition}

We define a progress measure for our set theory axiomatisation. The first and foremost ingredient of our progress measure is an overestimation on the amount of enabled E-matches. We anticipate that this overestimation strictly descents after each instantiation step and does not ascend after each case-splitting step. The second ingredient is the amount of unverified current clauses, which we expect to descent by at least one after each case-splitting step. The result of the progress measure is a lexicographically ordered pair of the above two ingredients.

\begin{definition}[Progress Measure for Set Theory] 
\label{def:progress-measure}
We define the progress measure $M : \textsc{State} \longrightarrow (\mathbb{N}\cup\{-1\})^{2}$, as follows, where $\Amount \cdot$ denotes cardinality. 
\[
M\left(s\right)=\begin{cases}
\left( \underset{p\in P(\state WAE,B)}{\sum}\Amount{p},\; 
\Amount{\left\{ C\in A\;|\;\NotVerify{W,\Egraph E}C\right\} } \right)  &\mathllap{\text{if}}\text{ }s=\state WAE\\
&\mathllap{\text{and }B\text{ is}}\text{ a basis for }\Egraph E\\
\Bigl( -1,\;-1 \Bigr) &\mathllap{\text{if}}\text{ }s=\bot \text{ or } \lozenge
\end{cases}
\]
\end{definition}
Inconsistent or saturated states are assigned (the smallest) measures $(-1, -1)$. The order on $(\mathbb{N}\cup\{-1\})$ is the natural extension of that on $\mathbb{N}$.

\subsection{Invariants and Termination Theorem}

Drawing on program reasoning, we anticipate classical techniques such as induction variants can be employed to termination proofs. 
We maintain two kinds of induction variants: general-purpose and problem-specific invariants. 

General-purpose invariants uphold the integrity of our formal semantics, remaining valid across all applications. 
For example, the E-history $\Ehistory E$ of an arbitrary state $s=\state WAE$ must be up to date \wrt the current quantifiers $W$ and E-interface $\Egraph E$. 
That is, for every pair $\left(\Tag\tau:\multi r\right)$  from $\Ehistory E$, 
there exists a quantifier $\unitriMulti {x}{T}{A}$ from $W$ whose tag is $\Tag\tau$, the dimension of $\multi x$ is equal to that of $\multi r$, $\Known{\Egraph E}{\multi r}$, and $\Known{\Egraph E}{\multi t\left[\nicefrac{\multi r}{\multi x}\right]}$ for some trigger set $\multi t$ from $\multi {[T]}$. 
(\cf Appx.~\ref{sec:appx-formal-model-for-e-matching} for more invariants.)

Problem-specific invariants are tailored to the distinct features of each problem,
focusing on preserving properties of solver states that are reachable from certain initial setups,
and tracing the origins of terms in intermediate states, crucial for complex axiomatisations like set theory. 
For example, consider an arbitrary intermediate state $\state WAE$, 
for each extended clause in $A$ with the form of $\neg\Member\left(t,\Union\left(a,b\right)\right)\vee\Member(t,a)\vee\Member(t,b)$,
$\left(\Tag\text{union-elim}:(a,b,t)\right)\in\Ehistory E$ holds, where axiom (union-elim) is defined as per Example \ref{exa:example-union-elim}.
This invariant concerns the origins of the extended clauses in the current clauses $A$.
Case-splitting on a current clause (e.g. the one above) may seem to introduce a new term, but this invariant indicates that this term is not new---it is equal to a known term that triggered a prior instantiation, as tracked by the E-history $\Ehistory E$.
This ensures a traceable lineage for each clause, linking it back to a specific quantifier in the E-history.
(\cf Appx.~\ref{sec:appx-proving-instantiation-termination-set-theory} for more invariants.)

We finally define the instantiation termination theorem for set theory, proven by induction on traces leveraging both general-purpose and set-theory-specific invariants. 
Note that termination is proved against an \emph{arbitrary} set of ground literals---this works because our progress measure and invariants are defined parametrically with the current state. 
Given these right invariants and termination measure, the proof is straightforward 
(\cf Appx.~\ref{sec:appx-proving-instantiation-termination-set-theory}). 
This theorem guarantees the absence of matching loops in this axiomatisation; users of this axiomatisation hence can confidently seek terminating answers to ground theory queries.

\begin{theorem}[Instantiation Termination for Set Theory] 
Suppose $L$ is an arbitrary set of ground literals.
The initial state is $s_{0}=\state{W_{0}}{A_{0}}{E_{0}}$,
where $W_{0}$ is our axiomatisation for set theory with tags,
$A_{0}=\emptyset$, $\Egraph{E_{0}} = \updateEgraph{\emptyset}{L}$, and $\Ehistory{E_{0}}=\emptyset$.
Any sequence of transitions from the initial state $s_{0}$, where $\tran$ defined in Sec. \ref{subsec:State-transit} represents the transition relation, has a finite length.  \label{thm:inst-term-set-theory}
\end{theorem}

\section{Related Work}
\label{sec:related-work}

For the purpose of program verification, where SMT solvers are used
to prove unsatisfiability, E-matching is widely used to handle quantifiers.
The idea of E-matching dates back to Nelson \cite{Nelson1980-PV}, which was first put into practice in Simplify~\cite{Simplify}.
Since then, efficient handling of E-matching-based quantifier instantiations has been studied by, \eg~de Moura and Bj{\o}rner~\cite{deMoura2007-Efficient-E-matching} for Z3, Ge et al.~\cite{Ge2007-E-matching-CVC3} for CVC3, Bansal et al.~\cite{Bansal2015-E-matching} for Z3 and CVC4, and Moskal et al.~\cite{Moskal2008-E-matching-Fun-Profit-Fx7} for Fx7.
When satisfiable results and their models are of interest, model-based
quantifier instantiation (MBQI)~\cite{Ge2009-MBQI} can be used to handle quantifiers.

Dross et al.~\cite{Dross2014-Thesis,Dross2012,Dross2016} formally define and reason about instantiation termination in a similar context. 
They define a novel \emph{logic} with first-class triggers, introduce \emph{instantiation trees} as algebraic objects to help define termination, and provide an ingenious technique for showing, for their implementation in Alt-Ergo, that finding a \emph{single} finite instantiation tree is sufficient for termination.

Despite being a powerful tool for numerous deep meta-theoretic results \cite{Dross2014-Thesis}, we believe that \emph{applying} a formal inductive construction of instantiation trees for larger examples would be complex in practice: existing examples focus instead on bounds for the sets of terms ever generatable by a solver run. These arguments closely relate to our inductive termination proofs over traces. Our work enables detailed formal proofs based directly on such familiar notions from program reasoning, including inductive invariants and well-founded measures.

The approach of this prior work also requires restrictions on solver behaviour, including \emph{fairness} of quantifier instantiation, and \emph{eager} application of theory deductions (via entailments in their custom logic)\footnote{We explain how to simply add theory steps to our operational model in Sec.~\ref{subsec:Theory-specific-reasoning}.}. Our operational model and termination proofs do not require or build in such assumptions. Still, \emph{restricting} our traces (\eg{} with fairness constraints) would be simple to do if desired for specific applications. Our weak assumptions make our approach (extended with appropriate theory deduction steps) applicable to SMT solvers broadly; solvers such as Z3 \cite{Z3} and CVC5 \cite{cvc5} commonly interleave theory reasoning and quantifier instantiation in (bounded or exhaustive) rounds of multiple steps.

The Axiom Profiler~\cite{Becker2019-Axiom-Profiler} leverages Z3 log files to provide comprehensive support for analysing quantifier instantiations.
The tool focuses on helping users effectively understand and debug problematic solver runs, rather than proving their absence.
It was validated by empirical evidence rather than formal proofs.

Existing works on the termination of SMT transition systems~\cite{Barrett2006-Splitting-on-Demand-SMT,Bonacina2020-CDSAT-Conflict-Driven-SAT,Bonacina2022-CDSAT-Conflict-Driven-SAT,deMoura2013-mcSAT-Model-Constructing-Satisfiability-Calculus} demonstrate that divergence is prevented by ensuring all new terms derive from a finite basis. In contrast, 
in our work, a finite basis does not imply termination---the basis can grow. At a high level these works prove that certain solver aspects always terminate. However,  E-matching cannot have this property; instead it places the onus on the author of an axiomatisation to achieve termination through careful selection of axioms and triggers, motivating a user-facing model.

\section{Conclusion and Future Work}

We have shown a novel model for E-matching as widely employed in SMT solvers, abstracting over solver details while enabling detailed and formal proofs of instantiation termination. Our model has been shown to apply directly and rigorously to the kinds of axiomatisations used in practical verification tools.

In future work, we would like to explore axiomatisations that rely on more-restricted characteristics of a solver, such as fairness of instantiation selection or theory reasoning steps. Similarly to our E-interfaces, we will investigate suitable abstractions over theory solver interactions incorporated into a proof search.

While instantiation termination is a much sought-after property, the complementary problem of guaranteed instantiation completeness is a natural next target to investigate with our novel operational model, which may require us to also explore various fairness restrictions of our model's transition relation.

\begin{credits}
\subsubsection{\ackname} We thank the anonymous reviewers, Mark R. Greenstreet and Yanze Li for their detailed and constructive suggestions. We are very grateful to Claire Dross for putting generous time and energy into thoughtful feedback for us. This work has been partly funded by NSERC Discovery Grants held by Garcia and Summers. 
\end{credits}



%
%
%
\bibliographystyle{splncs04}
\bibliography{references}

\newpage{}

\begin{appendix}

\section{An Operational Semantics for E-matching\label{sec:appx-formal-model-for-e-matching}}

\global\long\def\multi#1{\overrightarrow{#1}}%

\global\long\def\unitriMulti#1#2#3{\forall\multi{#1}.\multi{\left[#2\right]}#3}%

\global\long\def\unitriFlat#1#2#3{\forall#1.\left[#2\right]#3}%

\global\long\def\Tag{\sharp}%

\global\long\def\filter{\mathrm{filter}}%

\global\long\def\filterUni#1{\filter_{\forall}\left(#1\right)}%

\global\long\def\filterLit#1{\filter_{\mathrm{lit}}\left(#1\right)}%

\global\long\def\state#1#2#3{\left\langle #1,#2,#3\right\rangle }%

\global\long\def\egraph#1{#1^{\mathrm{I}}}%

\global\long\def\ehistory#1{#1^{\mathrm{H}}}%

\global\long\def\tran{\longrightarrow}%

\global\long\def\otran{\tran_{\mathrm{\vee}}}%

\global\long\def\ptran{\tran_{\mathrm{\forall}}}%

\begin{definition}
[Terms, Atoms and Literals] Let $\textsc{Var}$ be a set of variables,
$\mathcal{F}$ be a set of function symbols, and $\mathcal{P}$ be
a set of predicate symbols. Function and predicate symbols have their
arities and types of their parameters intrinsically specified. 

Define the set of \emph{terms} as follows: 
\[
t\in\textsc{Term}
\]
\[
t\Coloneqq x\mid c\mid f(t_{1},\dots,t_{n})
\]
where $x$ ranges over variables $\textsc{Var}$, $c$ over nullary
function symbols in $\mathcal{F}$, and $f$ over function symbols
in $\mathcal{F}$ with arity $n\geq1$. 

Define the set of \emph{atoms} (a.k.a. \emph{atomic formulas}) as
follows: 
\[
p\in\textsc{Atom}
\]
\[
p\Coloneqq\bot\mid\left(t_{1}=t_{2}\right)\mid P(t_{1},\dots,t_{n})
\]
where $\bot$ is logical falsehood, $=$ is logical equality\footnote{Logical equality is treated differently from normal predicates and
comes with special properties defined by equality axioms regarding
reflexivity, symmetry, transitivity, substitution for functions, and
substitution for predicates. }, and $P$ ranges over predicate symbols in $\mathcal{P}$ with arity
$n\geq1$. 

Define the set of \emph{literals} as follows: 
\[
l\in\textsc{Literal}
\]
\[
l\Coloneqq p\mid\neg p
\]
As a convention, $\top$ denotes $\neg\bot$, and $t_{1}\neq t_{2}$
denotes $\neg(t_{1}=t_{2})$. 
\end{definition}

\begin{definition}
[Tags] We assume a set of pre-defined \emph{primitive tags}, denoted
$\textsc{PTag}$. Define the set of \emph{tags} as follows: 
\[
\Tag\tau\in\textsc{Tag}
\]
\begin{align*}
\Tag\tau & \Coloneqq\Tag\tau_{0}\mid\Tag\tau(\multi t)
\end{align*}
where $\Tag\tau_{0}\in\textsc{PTag}$.

We typically use $\Tag\tau$ and $\Tag\alpha$ to denote a tag. 
\end{definition}

\begin{remark}
We assume that a global mechanism exists to generate fresh and unique
tags, in accordance with standard practice. 
\end{remark}

\begin{definition}
[Formulas, (Tagged) Quantifiers, and Triggers] Define the set of
\emph{formulas} as follows: 
\[
A\in\textsc{Formula}
\]
\[
A\Coloneqq l\mid A\wedge A\mid A\vee A\mid A\rightarrow A\mid\left(\unitriMulti xTA\right)^{\Tag\tau}
\]
Here, $\left(\unitriMulti xTA\right)^{\Tag\tau}$ represents a \emph{tagged
quantifier}, i.e. a quantifier $\unitriMulti xTA$ identified by its
tag $\Tag\tau\in\textsc{Tag}$. 

A \emph{quantifier} $\unitriMulti xTA$ is short for $\forall x_{1},\dots,x_{n}.[T_{1}]\dots[T_{m}]A$.
The (possibly-multiple) variables $\multi x$ are bound; the (possibly-multiple)
trigger sets $\multi T$ are each marked with square brackets and
positioned before the quantifier body A. 

A \emph{trigger set} $T$ is a (non-empty) set of terms, written comma-separated,
and adhere to the following conditions when occurring in a quantifier
$\unitriMulti xTA$: 
\begin{enumerate}
\item The trigger set $T$ must contain each quantified variable of $\multi x$
at least once; 
\item Each term of the trigger set $T$must contain at least one quantified
variable of $\multi x$;
\item Each term of the trigger set $T$ must contain at least one uninterpreted
function application and no interpreted function symbols such as equalities.
\end{enumerate}
\end{definition}

\begin{remark}
A multi-term trigger set prescribes that terms matching all terms
of the trigger set must be known for some instantiation of the quantified
variables, whereas multiple trigger sets prescribe alternative conditions
for instantiation such that matching one trigger set is sufficient
to trigger an instantiation. 
\end{remark}

\begin{remark}
When quantifier tags are not relevant, we omit them for brevity.
\end{remark}

\begin{definition}
[Extended Literals, Extended Clauses and Extended CNF Formulas] Define
\emph{extended literals}, \emph{extended clauses} and \emph{extended
conjunctive normal form (ECNF) formulas} as follows: 
\[
\phi\in\textsc{ExtLiteral}\quad C\in\textsc{ExtClause}\quad A\in\textsc{ExtCNF}
\]
\[
\begin{array}{rcl}
\phi & \Coloneqq & l\mid\left(\unitriMulti xTA\right)^{\Tag\tau}\\
C & \Coloneqq & \phi\mid C\vee C\\
A & \Coloneqq & C\mid A\wedge A
\end{array}
\]
If $A$ is $C_{1}\wedge\cdots\wedge C_{n}$, we may represent $A$
as a conjunctive set $\bigwedge\{C_{1},\dots,C_{n}\}$ or $\{C_{1},\dots,C_{n}\}$.
If $C$ is $\phi_{1}\vee\dots\vee\phi_{n}$, we may represent $C$
as a disjunctive set $\bigvee\{\phi_{1},\dots,\phi_{n}\}$ or $\{\phi_{1},\dots,\phi_{n}\}$.
\end{definition}

\begin{remark}
We assume all formulas are pre-converted to ECNF. In particular, existential
quantifiers have been eliminated by Skolemisation. 
\end{remark}

\begin{remark}
We sometimes retain logical implications, denoted $\rightarrow$,
in our examples to help clarify the meaning of the formulas, although
they can be eliminated easily. 
\end{remark}

\begin{definition}
[Filters on Extended Literals] Define \emph{filter functions} $\filter_{\forall}$
and $\filter_{\text{lit}}$ to filter sets of extended literals into
only those which are quantifiers or only those which are simple literals,
respectively. 
\[
\filterUni{\Phi}=\left\{ \phi\mid\phi\in\Phi,\phi\text{ is of the form }\left(\unitriMulti xTA\right)^{\Tag\tau}\right\} 
\]
\[
\filterLit{\Phi}=\left\{ \phi\mid\phi\in\Phi,\phi\in\textsc{Literal}\right\} 
\]
\end{definition}

\begin{proposition}
For any set of extended literals $\Phi$, $\filterUni{\Phi}\cup\filterLit{\Phi}=\Phi$
and $\filterUni{\Phi}\cap\filterLit{\Phi}=\emptyset$. 
\end{proposition}

\global\long\def\taggen#1{\mathrm{tag}\left(#1\right)}%

\global\long\def\taggenName{\mathrm{tag}}%

\global\long\def\taggenOne#1{\mathrm{tag_{\iota}}\left(#1\right)}%

\global\long\def\taggenOneName{\mathrm{tag_{\iota}}}%

\global\long\def\taggenOneVar#1#2{\mathrm{tag}_{\iota}^{#2}\left(#1\right)}%

\begin{definition}
[Constraints on Quantifier Tags] An axiomatisation for theory $T$
with tags, denoted $W_{T}$, is a set of tagged quantifiers that adhere
to the following conditions: 
\begin{enumerate}
\item Each quantifier (including each appearing in a nested quantifier)
in the axiomatisation $W_{T}$ is labelled with a distinct tag. 
\item The tag for any non-nested quantifier or the outermost quantifier
of any nested quantifier is not parameterised. That is, for any tagged
quantifier $\left(\unitriMulti xTA\right)^{\Tag\tau}\in W_{T}$, its
tag $\Tag\tau$ is a primitive tag. 
\item An inner quantifier that occurs in a nested quantifier has its tag
parameterised by all of its outer-quantified variables. 
\item Instantiating an outer quantifier produces a copy of the quantifier
body in which (among other changes) tags of all inner quantifiers
that are parameterised by this outer-quantifier are updated to reflect
this instantiation. 
\end{enumerate}
\end{definition}

\global\long\def\updateEgraph#1#2{#1\triangleleft#2}%

\global\long\def\updateW#1#2#3{#1\triangleleft\left(#2,#3\right)}%

\global\long\def\Egraph#1{#1^{\mathrm{I}}}%

\global\long\def\Ehistory#1{#1^{\mathrm{H}}}%

\global\long\def\Entail#1#2{#1\Vdash#2}%

\global\long\def\NotEntail#1#2{#1\not\Vdash#2}%

\global\long\def\Known#1#2{#1\Vdash_{\mathrm{known}}#2}%

\global\long\def\Class#1#2#3{#1\Vdash_{\mathrm{class}}#2\in#3}%

\global\long\def\Equiv#1#2{#1\sim#2}%

\global\long\def\NotEquiv#1#2{#1\not\sim#2}%

\global\long\def\Amount#1{\left\Vert #1\right\Vert }%

\global\long\def\AmountClass#1{\left\Vert #1\right\Vert _{\mathrm{class}}}%

\global\long\def\AmountLeft#1{\left\Vert #1\right.}%

\global\long\def\AmountRight#1{\left.#1\right\Vert }%

\global\long\def\Inst#1#2{#1\Vdash_{\mathrm{inst}}#2}%

\global\long\def\NotInst#1#2{#1\not\Vdash_{\mathrm{inst}}#2}%

\global\long\def\Select{\vdash_{\mathrm{select}}}%

\global\long\def\Ematching{\vdash_{\mathrm{match}}}%

\global\long\def\MatchSep{\sphericalangle}%

\begin{definition}
[States] \emph{States} are defined as follows: 
\[
s\in\textsc{State}
\]
\[
s\Coloneqq\state WAE\mid\lozenge\mid\bot
\]
where $\lozenge$ and $\bot$ are distinguished symbols for saturated
and inconsistent states, $W$ (the \emph{current quantifiers}) is
a set of tagged quantifiers, $A$ (the \emph{current clauses}) is
a set of extended clauses, and E (the \emph{current E-state}) is defined
later. 
\end{definition}

\begin{definition}
[E-interface Judgements] An \emph{E-interface} $\Egraph E\in\textsc{EInter}$
is a set of literals. We write $\Known{\Egraph E}t$ to express that
the ground term $t$ is \emph{known} in the E-interface $\Egraph E$;
we write $\Entail{\Egraph E}{\Equiv{t_{1}}{t_{2}}}$ to express that
the ground terms $t_{1}$ and $t_{2}$ are \emph{known equal} in $\Egraph E$;
we write $\Entail{\Egraph E}{\NotEquiv{t_{1}}{t_{2}}}$ to express
that the ground terms $t_{1}$ and $t_{2}$ are \emph{known disequal}
in $\Egraph E$; we write $\Entail{\Egraph E}{\bot}$ to express that
the E-interface $\Egraph E$ is \emph{inconsistent}. These four judgements
are defined by the following derivation rules: \label{def:appx-e-interface-judgements-full}

$\boxed{\Entail{\Egraph E}{\Equiv tt}}$

$\boxed{\Known{\Egraph E}t}$
\begin{center}
$\dfrac{\Equiv{t_{1}}{t_{2}}\in\Egraph E}{\Entail{\Egraph E}{\Equiv{t_{1}}{t_{2}}}}\textsc{(eq-in)}$\qquad{}$\dfrac{\Entail{\Egraph E}{\Equiv{t_{2}}{t_{1}}}}{\Entail{\Egraph E}{\Equiv{t_{1}}{t_{2}}}}\textsc{(eq-sym)}$
\par\end{center}

\begin{center}
$\dfrac{\begin{array}{cc}
\Entail{\Egraph E}{\Equiv{t_{1}}{t_{2}}}\quad & \quad\Entail{\Egraph E}{\Equiv{t_{2}}{t_{3}}}\end{array}}{\Entail{\Egraph E}{\Equiv{t_{1}}{t_{3}}}}\textsc{(eq-tran)}$
\par\end{center}

\begin{center}
$\dfrac{\Known{\Egraph E}t}{\Entail{\Egraph E}{\Equiv tt}}\textsc{(eq-kn-refl)}$
\par\end{center}

\begin{center}
$\dfrac{\begin{array}{cc}
\Entail{\Egraph E}{\Equiv{t_{i}}{t_{i}^{\prime}}}\quad & \quad\Known{\Egraph E}{g\left(t_{1},\dots,t_{i},\dots,t_{n}\right)}\end{array}}{\Entail{\Egraph E}{\Equiv{g\left(t_{1},\dots,t_{i},\dots,t_{n}\right)}{g\left(t_{1},\dots,t_{i}^{\prime},\dots,t_{n}\right)}}}\textsc{(eq-kn-sub)}$
\par\end{center}

\begin{center}
$\dfrac{\Entail{\Egraph E}{\Equiv{t_{1}}{t_{2}}}}{\Known{\Egraph E}{t_{1}}}\text{\textsc{(kn-eq)}}$\qquad{}$\dfrac{\Known{\Egraph E}{g\left(\dots,t_{i},\dots\right)}}{\Known{\Egraph E}{t_{i}}}\textsc{(kn-sub)}$
\par\end{center}

$\boxed{\Entail{\Egraph E}{\NotEquiv tt}}$
\begin{center}
$\dfrac{\NotEquiv{t_{1}}{t_{2}}\in\Egraph E}{\Entail{\Egraph E}{\NotEquiv{t_{1}}{t_{2}}}}\textsc{(neq-in)}$\qquad{}$\dfrac{\Entail{\Egraph E}{\NotEquiv{t_{2}}{t_{1}}}}{\Entail{\Egraph E}{\NotEquiv{t_{1}}{t_{2}}}}\textsc{(neq-sym)}$
\par\end{center}

\begin{center}
$\dfrac{\begin{array}{cc}
\Entail{\Egraph E}{\NotEquiv{t_{1}}{t_{2}}}\quad & \quad\Entail{\Egraph E}{\Equiv{t_{2}}{t_{3}}}\end{array}}{\Entail{\Egraph E}{\NotEquiv{t_{1}}{t_{3}}}}\textsc{(neq-eq)}$
\par\end{center}

$\boxed{\Entail{\Egraph E}{\bot}}$

\[
\dfrac{\begin{array}{cc}
\Entail{\Egraph E}{\Equiv{t_{1}}{t_{2}}}\quad & \quad\Entail{\Egraph E}{\NotEquiv{t_{1}}{t_{2}}}\end{array}}{\Entail{\Egraph E}{\bot}}\textsc{(bot)}
\]

\end{definition}

\begin{remark}
E-interfaces are equivalent if they agree on these judgements in all
cases. 
\end{remark}

\begin{definition}
[E-interface Extension] For a set of equality and disequality literals
$L$, the \emph{update of an E-interface} $\Egraph E$ \emph{with}
$L$, denoted $\updateEgraph{\Egraph E}L$, is a minimal E-interface
which satisfies all E-interface judgements that $\Egraph E$ does,
while also satisfying $\Entail{\Egraph E}l$ for all $l\in L$. 
\end{definition}

\begin{definition}
[E-interface Judgements (Lifted)] We lift the E-interface judgements
to support vectors of terms as follows: \label{def:appx-e-interface-judgements-lifted}
\begin{description}
\item [{$\boxed{\Entail{\Egraph E}{\Equiv{\multi t}{\multi t}}}$}] Suppose
$\multi{t_{1}}=\left(t_{11},t_{12},\dots,t_{1n}\right)$ and $\multi{t_{2}}=\left(t_{21},t_{22},\dots,t_{2n}\right)$.
$\Entail{\Egraph E}{\Equiv{\multi{t_{1}}}{\multi{t_{2}}}}$ if and
only if $\Entail{\Egraph E}{\Equiv{t_{1i}}{t_{2i}}}$ for every $i\in\{1,\dots,n\}$. 
\item [{$\boxed{\Known{\Egraph E}{\multi t}}$}] Suppose $\multi t=\left(t_{1},t_{2},\dots,t_{n}\right)$.
$\Known{\Egraph E}{\multi t}$ if and only if $\Known{\Egraph E}{t_{i}}$
for every $i\in\{1,\dots,n\}$. 
\item [{$\boxed{\Entail{\Egraph E}{\NotEquiv{\multi t}{\multi t}}}$}] Suppose
$\multi{t_{1}}=\left(t_{11},t_{12},\dots,t_{1n}\right)$ and $\multi{t_{2}}=\left(t_{21},t_{22},\dots,t_{2n}\right)$.
$\Entail{\Egraph E}{\NotEquiv{\multi{t_{1}}}{\multi{t_{2}}}}$ if
and only if $\Entail{\Egraph E}{\NotEquiv{t_{1i}}{t_{2i}}}$ for every
$i\in\{1,\dots,n\}$. 
\end{description}
\end{definition}

\begin{definition}
[E-interface: Candidate Basis] We call a set of terms $B$ a \emph{candidate
basis} of $\Egraph E$ if for every $t$, if $\Known{\Egraph E}t$,
then there exists some $t_{i}\in B$ such that $\Entail{\Egraph E}{\Equiv t{t_{i}}}$. 

We typically use $B_{\Egraph E}$ to represent a candidate basis of
$\Egraph E$ . 
\end{definition}

\begin{definition}
[E-interface: Basis] We call a set of terms $\{t_{1},\dots,t_{n}\}$
a \emph{basis} of $\Egraph E$ if 
\begin{enumerate}
\item for every $t_{i}\in\{t_{1},\dots,t_{n}\}$, $\Known{\Egraph E}{t_{i}}$; 
\item for every $t_{i},t_{j}\in\{t_{1},\dots,t_{n}\}$ where $i\neq j$,
$\Entail{\Egraph E}{\Equiv{t_{i}}{t_{j}}}$ does not hold; 
\item for every $t$, if $\Known{\Egraph E}t$, then there exists some $t_{i}\in\{t_{1},\dots,t_{n}\}$
such that $\Entail{\Egraph E}{\Equiv t{t_{i}}}$. 
\end{enumerate}
A basis for $\Egraph E$ is also referred to as a set of representatives
of equivalence classes of terms known in $\Egraph E$. We (also) use
$B_{\Egraph E}$ to represent a basis of $\Egraph E$. 
\end{definition}

\begin{remark}
Intuitively, a basis of $\Egraph E$ is a minimal set of known terms
in $\Egraph E$ such that any known term in $\Egraph E$ must be equivalent
to exactly one of them. 
\end{remark}

\begin{proposition}
If $\Entail{\Egraph E}{\Equiv{\multi{t_{1}}}{\multi{t_{2}}}}$ and
$\Egraph E\subseteq\Egraph{E_{1}}$, then $\Entail{\Egraph{E_{1}}}{\Equiv{\multi{t_{1}}}{\multi{t_{2}}}}$. 
\end{proposition}

\begin{proposition}
If $\NotEntail{\Egraph E}{\Equiv{\multi{t_{1}}}{\multi{t_{2}}}}$
and $\Egraph E\supseteq\Egraph{E_{1}}$, then $\NotEntail{\Egraph{E_{1}}}{\Equiv{\multi{t_{1}}}{\multi{t_{2}}}}$. 
\end{proposition}

\begin{proposition}
If $B_{E}$ is a candidate basis for $\Egraph E$, then there exists
$B_{E}^{\prime}\subseteq B_{E}$ such that $B_{E}^{\prime}$ is a
basis for $\Egraph E$. 
\end{proposition}

\begin{proposition}
For any E-interface $\Egraph E$, all its bases have the same cardinality.
\end{proposition}

\begin{definition}
[E-interface: Equivalence of Tags] We write $\Entail{\Egraph E}{\Equiv{\Tag\tau_{1}}{\Tag\tau_{2}}}$
to express that the tags $\Tag\tau_{1}$ and $\Tag\tau_{2}$ are \emph{equivalent}
in $\Egraph E$. The judgement is defined as follows: \label{def:appx-E-interface-equiv-of-tags}
\begin{center}
$\dfrac{\Tag\tau_{0}\equiv\Tag\tau_{0}^{\prime}}{\Entail{\Egraph E}{\Equiv{\Tag\tau_{0}}{\Tag\tau_{0}^{\prime}}}}\text{ where \ensuremath{\Tag\tau_{0},\Tag\tau_{0}^{\prime}}\ensuremath{\in\textsc{PTag}}}\;\textsc{(prim)}$\qquad{}$\dfrac{\begin{array}{cc}
\Entail{\Egraph E}{\Equiv{\Tag\tau_{1}}{\Tag\tau_{2}}}\quad & \quad\Entail{\Egraph E}{\Equiv{\multi{t_{1}}}{\multi{t_{2}}}}\end{array}}{\Entail{\Egraph E}{\Equiv{\Tag\tau_{1}(\multi{t_{1}})}{\Tag\tau_{2}(\multi{t_{2}})}}}\textsc{(nest)}$
\par\end{center}

Note that $\Entail{\Egraph E}{\Equiv{\multi{t_{1}}}{\multi{t_{2}}}}$
refers to the E-interface judgement that terms $\multi{t_{1}}$ and
$\multi{t_{2}}$ are known equal in $\Egraph E$. 
\end{definition}

\begin{definition}
[E-histories] An \emph{E-history} $\Ehistory E\in\textsc{EHist}$
is a set of pairs (each denoted $\left(\Tag\alpha:\multi r\right)$
) in our formalism: the first element is a tag (identifying a quantifier),
and the second is a vector of ground terms (representing an instantiation
of the corresponding quantifier). 
\end{definition}

\begin{definition}
[E-states] An \emph{E-state} $E$ is a pair $(\Egraph E,\Ehistory E)$
of E-interface and E-history. 
\end{definition}

\begin{definition}
[History-Enabled E-matches] Given a candidate pair $\left(\Tag\alpha:\multi r\right)$
(of tag $\Tag\alpha$ and vector of terms $\multi r$), the \emph{E-state}
$E$ \emph{enables} $\left(\Tag\alpha:\multi r\right)$, written $\Inst E{\left(\Tag\alpha:\multi r\right)}$,
if: for every pair $\left(\Tag\alpha:\multi{r^{\prime}}\right)\in\Ehistory E$,
$\NotEntail{\Egraph E}{\Equiv{\multi r}{\multi{r^{\prime}}}}$. \label{def:appx-hist-enabled-e-match-not-opt}
\end{definition}

\begin{remark}
$\NotEntail{\Egraph E}{\Equiv{\multi r}{\multi{r^{\prime}}}}$ in
the definition above uses a lifted version of an E-interface judgement
$\Entail{}{}$, defined in Def. \ref{def:appx-e-interface-judgements-lifted}.
An equivalent definition is to replace this with the following: at
least one of the pointwise equalities $\multi{r_{i}\sim r_{i}^{\prime}}$
is not \emph{known} in $\Egraph E$. 
\end{remark}

\begin{definition}
[History-Enabled E-matches (Optimised)] Given a candidate pair $\left(\Tag\alpha:\multi r\right)$
(of tag $\Tag\alpha$ and vector of terms $\multi r$), the \emph{E-state}
$E$ \emph{enables} $\left(\Tag\alpha:\multi r\right)$, written $\Inst E{\left(\Tag\alpha:\multi r\right)}$,
if: for every pair $\left(\Tag\tau:\multi{r^{\prime}}\right)\in\Ehistory E$
such that $\Entail{\Egraph E}{\Equiv{\Tag\alpha}{\Tag\tau}}$, $\NotEntail{\Egraph E}{\Equiv{\multi r}{\multi{r^{\prime}}}}$.
\label{def:appx-hist-enabled-e-match-opt}

Note that the equivalence of tags $\Entail{\Egraph E}{\Equiv{\Tag\alpha}{\Tag\tau}}$
is defined in Def. \ref{def:appx-E-interface-equiv-of-tags}. 
\end{definition}

\begin{remark}
The above optimised version is not included in the main contents of
the paper. 
\end{remark}

\begin{definition}
[E-history Extension] Let $\Ehistory E\in\textsc{EHist}$. Updating
$\Ehistory E$ with a pair $\left(\Tag\alpha:\multi r\right)$ is
defined as follows: 
\[
\updateEgraph{\Ehistory E}{\left(\Tag\alpha:\multi r\right)}=\Ehistory E\cup\left\{ \left(\Tag\alpha:\multi r\right)\right\} 
\]
\end{definition}

\begin{remark}
Since $\updateEgraph{\Ehistory E}{\left(\Tag\alpha:\multi r\right)}$
occurs immediately after verifying $\Inst E{\left(\Tag\alpha:\multi r\right)}$
within a single quantifier instantiation transition, the pair $\left(\Tag\alpha:\multi r\right)$
cannot be redundant for $\Ehistory E$. The definition of being ``redundant''
is in line with the version of $\Inst{}{}$ employed. 
\end{remark}

\global\long\def\Verify#1#2{#1\Vdash_{\mathrm{verify}}#2}%

\global\long\def\NotVerify#1#2{#1\not\Vdash_{\mathrm{verify}}#2}%

\begin{definition}
[Verified Extended Clauses] An extended clause $\phi_{1}\vee\dots\vee\phi_{n}$
is \emph{verified} by a set of quantifiers $W$ and an E-interface
$\Egraph E$, written $\Verify{W,\Egraph E}{\phi_{1}\vee\dots\vee\phi_{n}}$,
if there is some $\phi_{i}$ where $1\leq i\leq n$, such that one
of the following is satisfied: 
\begin{itemize}
\item $\phi_{i}$ is a tagged quantifier $\left(\unitriMulti xT{A_{11}}\right)^{\Tag\alpha}\in W$;
\item $\phi_{i}$ is $t_{1}=t_{2}$ and $\Entail{\Egraph E}{\Equiv{t_{1}}{t_{2}}}$;
\item $\phi_{i}$ is $t_{1}\neq t_{2}$ and $\Entail{\Egraph E}{\NotEquiv{t_{1}}{t_{2}}}$. 
\end{itemize}
\end{definition}

\begin{definition}
[Updating Current Quantifiers] Let $W=\{v_{1}^{\Tag\tau_{1}},\dots v_{i}^{\Tag\tau_{i}},\dots\}$
where each $v_{i}^{\Tag\tau_{i}}$ is a tagged quantifier. Updating
$W$ with a tagged quantifier $v^{\Tag\tau}$ with the help of an
E-interface $\Egraph{}$ is defined as follows: \label{def:appx-updating-current-quantifiers-not-opt}
\[
\updateW W{v^{\Tag\tau}}{\Egraph E}=W\cup\{v^{\Tag\tau}\}
\]

The E-interface $\Egraph E$ a redundant parameter of this update
operation. We may instead write $W\cup\{v^{\Tag\tau}\}$ to specifically
refer to this version of updating current quantifiers. 

By convention of overloading, we lift the above relation to support
$\updateW W{\Phi}{\Egraph E}$ where $\Phi=\left\{ \dots,v^{\Tag\tau},\dots\right\} $
is a set of tagged quantifiers. 
\end{definition}

\begin{definition}
[Updating Current Quantifiers (Optimised)] Let $W=\{v_{1}^{\Tag\tau_{1}},\dots v_{i}^{\Tag\tau_{i}},\dots\}$
where each $v_{i}^{\Tag\tau_{i}}$ is a tagged quantifier. Updating
$W$ with a tagged quantifier $v^{\Tag\tau}$ with the help of an
E-interface $\Egraph{}$ is defined as follows: \label{def:appx-updating-current-quantifiers-opt}
\[
\updateW W{v^{\Tag\tau}}{\Egraph E}=\begin{cases}
W\cup\left\{ v^{\Tag\tau}\right\}  & \text{if there does not exist any }v_{i}^{\Tag\tau_{i}}\in W\\
 & \text{such that }\Entail{\Egraph E}{\Equiv{\Tag\tau}{\Tag\tau_{i}}}\\
W & \text{otherwise}
\end{cases}
\]
By convention of overloading, we lift the above relation to support
$\updateW W{\Phi}{\Egraph E}$ where $\Phi=\left\{ \dots,v^{\Tag\tau},\dots\right\} $
is a set of tagged quantifiers. 
\end{definition}

\begin{remark}
The above optimised version is not included in the main contents of
the paper.  
\end{remark}

\begin{definition}
[E-matching] The judgement $\state WAE\Ematching\left(\unitriMulti xT{A^{\prime}}\right)^{\Tag\alpha}\MatchSep\multi r$
defines, for a given state $\state WAE$, which instantiations (using
terms $\multi r$) of which quantifiers $\left(\unitriMulti xT{A^{\prime}}\right)^{\Tag\alpha}$
are enabled, according to the rules of E-matching, as follows: 

\[
\dfrac{\begin{array}{cc}
\left(\unitriMulti xT{A^{\prime}}\right)^{\Tag\alpha}\in W & \multi t\text{ is one trigger set of }\multi{\left[T\right]}\\
\Known{\Egraph E}{\multi t\left(\multi x\right)\left[\nicefrac{\multi r}{\multi x}\right]} & \Inst E{\left(\Tag\alpha:\multi r\right)}
\end{array}}{\state WAE\Ematching\left(\unitriMulti xT{A^{\prime}}\right)^{\Tag\alpha}\MatchSep\multi r}
\]
We write $\state WAE\not\Ematching$ to mean no instantiations are
enabled in this state $\state WAE$. 
\end{definition}

\begin{definition}
[State Transitions] The (single step) state transition relation $\longrightarrow$
is defined as follows: 
\[
\tran\;\subseteq\;\textsc{State}\times\textsc{State}
\]
\[
\dfrac{\begin{array}{c}
\emptyset\subset\Phi\subseteq\left\{ \phi_{i}\mid C\in A,\;\NotVerify{W_{1},\Egraph{E_{1}}}C,\;C\text{ is }\dots\vee\phi_{i}\vee\dots\right\} \\
W_{2}=\updateW{W_{1}}{\filterUni{\Phi}}{\Egraph{E_{1}}}\quad\Egraph{E_{2}}=\updateEgraph{\Egraph{E_{1}}}{\filterLit{\Phi}}\quad\Ehistory{E_{2}}=\ehistory{E_{1}}
\end{array}}{\state{W_{1}}A{E_{1}}\tran\state{W_{2}}A{E_{2}}}\textsc{(split)}
\]
\[
\dfrac{\Entail{\Egraph E}{\bot}}{\state WAE\tran\bot}\textsc{(bot)}
\]
\[
\dfrac{\begin{array}{ccc}
\NotEntail{\Egraph E}{\bot}\quad & \quad\Verify{W,\Egraph E}C\text{ for every }C\in A\quad & \quad\state WAE\not\Ematching\end{array}}{\state WAE\tran\lozenge}\textsc{(sat)}
\]
\[
\dfrac{\begin{array}{c}
\state{W_{1}}{A_{1}}{E_{1}}\Ematching\left(\unitriMulti xT{A_{11}}\right)^{\Tag\alpha}\MatchSep\multi r\\
\begin{array}{cc}
A_{12}=A_{11}\left[\nicefrac{\multi r}{\multi x}\right] & A_{12}^{\prime}=\filterUni{A_{12}}\cup\filterLit{A_{12}}\\
A_{2}=A_{1}\cup\left(A_{12}\backslash A_{12}^{\prime}\right) & W_{2}=\updateW{W_{1}}{\filterUni{A_{12}}}{\Egraph{E_{1}}}\\
\Egraph{E_{2}}=\Egraph{E_{1}}\triangleleft\filterLit{A_{12}} & \Ehistory{E_{2}}=\Ehistory{E_{1}}\triangleleft\left(\Tag\alpha:\multi r\right)
\end{array}
\end{array}}{\state{W_{1}}{A_{1}}{E_{1}}\tran\state{W_{2}}{A_{2}}{E_{2}}}\textsc{(inst)}
\]
\end{definition}

\begin{remark}
Regarding $W_{2}=\updateW{W_{1}}{\filterUni{\Phi}}{\Egraph{E_{1}}}$
in the $\textsc{(split)}$ rule : Select all tagged quantifiers from
$\Phi$. Joining these tagged quantifiers with the current quantifiers
$W_{1}$ with the help of the current E-interface $\Egraph{E_{1}}$
yields $W_{2}$. We provide two definitions (cf. Def. \ref{def:appx-updating-current-quantifiers-not-opt}
and \ref{def:appx-updating-current-quantifiers-opt}) for updating
the current quantifiers, one of which (cf. Def. \ref{def:appx-updating-current-quantifiers-opt})
is optimised to rule out redundant quantifiers. 

Regarding $W_{2}=\updateW{W_{1}}{\filterUni{A_{12}}}{\Egraph{E_{1}}}$
in the $\textsc{(inst)}$ rule: Select all tagged quantifiers from
$A_{12}$. Joining these tagged quantifiers with the current quantifiers
$W_{1}$ with the help of the current E-interface $\Egraph{E_{1}}$
yields $W_{2}$. We have provided two definitions (cf. Def. \ref{def:appx-updating-current-quantifiers-not-opt}
and \ref{def:appx-updating-current-quantifiers-opt}) for updating
the current quantifiers, one of which (cf. Def. \ref{def:appx-updating-current-quantifiers-opt})
is optimised to rule out redundant quantifiers. 
\end{remark}

\global\long\def\Inj{\mathrm{inj}}%

\global\long\def\InjE{\mathrm{inj}_{\mathrm{E}}}%

\global\long\def\InjFunc#1{\Inj\left(#1\right)}%

\global\long\def\InjEFunc#1{\InjE\left(#1\right)}%

\begin{definition}
[Injection Function] Let $L$ be a set of ground literals. We write
$\InjFunc L$ to denote an initial E-state constructed from $L$.
The \emph{injection} function is defined as follows: 

\[
\Inj:\mathcal{P}\left(\textsc{Literal}\right)\rightarrow\textsc{EState}
\]
\[
\InjFunc L=\left(\updateEgraph{\emptyset}L,\emptyset\right)
\]
\end{definition}

\global\long\def\Invariant#1#2{I_{\mathrm{#1}}\left(#2\right)}%

\begin{definition}
[General-purpose Invariants] Suppose the initial state is $s_{0}=\state{W_{0}}{\emptyset}{\InjFunc L}$,
where $W_{0}$ is an axiomatisation for a theory with tags, and $L$
is an arbitrary set of ground literals from the same theory.

For an arbitrary intermediate state $s=\state WAE$, define the general-purpose
invariant $\Invariant G{s,s_{0}}$ to be a conjunction of the following
predicates: \label{def:appx-inv-general-purpose}
\begin{enumerate}
\item (Quantifiers have distinct tags). $\Invariant{G:QT}{s,s_{0}}$, which
holds iff: quantifiers in current quantifiers $W$ have distinct tags,
i.e., $\Tag\tau_{1}\not\equiv\Tag\tau_{2}$ for every two $v_{1}^{\Tag\tau_{1}},v_{2}^{\Tag\tau_{2}}\in W$
where $v_{1}^{\Tag\tau_{1}}\not\equiv v_{2}^{\Tag\tau_{2}}$. 
\item (No unit clauses in $A$). $\Invariant{G:NA}{s,s_{0}}$, which holds
iff: none of the extended clauses in $A$ is a unit clause, i.e.,
for every $C\in A$, $C$ is not of the form $\{\phi\}$. 
\item (E-history up-to-date with quantifiers). $\Invariant{G:EQ}{s,s_{0}}$,
which holds iff: the current E-history $\ehistory E$ is up to date
w.r.t. the current quantifiers $W$, i.e., for every pair $\left(\Tag\alpha:\multi r\right)$
from E-history $\Ehistory E$, there exists a quantifier $\unitriMulti xT{A^{\prime}}$
from $W$ whose tag is $\Tag\alpha$ and the dimension of $\multi x$
is equal to that of $\multi r$. 
\item (E-history up-to-date with E-interface). $\Invariant{G:EE}{s,s_{0}}$,
which holds iff: the current E-history $\ehistory E$ is update to
date w.r.t. the current E-interface $\egraph E$, i.e., for every
pair $\left(\Tag\alpha:\multi r\right)$ from E-history $\Ehistory E$,
there exists a quantifier $\unitriMulti xT{A^{\prime}}$ from $W$
whose tag is $\Tag\alpha$, $\Known{\Egraph E}{\multi r}$ and $\Known{\Egraph E}{\multi t\left[\nicefrac{\multi r}{\multi x}\right]}$
for some trigger set $\multi t$ from $\multi{\left[T\right]}$. 
\item (History-enabled E-matches with equal terms). $\Invariant{G:HE}{s,s_{0}}$,
which holds iff: if $\Inst E{\left(\Tag\tau:\multi r\right)}$, $\Entail{\Egraph E}{\Equiv{\multi r}{\multi{r^{\prime}}}}$,
then $\Inst E{\left(\Tag\tau:\multi{r^{\prime}}\right)}$. 
\item (Known terms on basis). $\Invariant{G:KB}{s,s_{0}}$, which holds
iff: if $\Known{\Egraph E}{\multi r}$, then for any basis of $\Egraph E$,
denoted $B_{E}$, there exists $\multi{r^{\prime}}=(r_{1}^{\prime},\dots,r_{i}^{\prime},\dots,r_{n}^{\prime})$
with each $r_{i}^{\prime}\in B_{E}$, and $\Entail{\Egraph E}{\Equiv{\multi r}{\multi{r^{\prime}}}}$. 
\item (E-interface grows). $\Invariant{G:IG}{s,s_{0}}$, which holds iff:
if $s\tran s^{\prime}$, then $\Egraph{\left(E^{\prime}\right)}$
is a conservative extension of $\Egraph E$ , i.e., $\Egraph E\subseteq\Egraph{\left(E^{\prime}\right)}$. 
\item (E-history grows). $\Invariant{G:HG}{s,s_{0}}$, which holds iff:
if $s\tran s^{\prime}$, then $\Ehistory{\left(E^{\prime}\right)}$
is a conservative extension of $\Ehistory E$, i.e., $\Ehistory E\subseteq\Ehistory{\left(E^{\prime}\right)}$. 
\item (Quantifiers grow). $\Invariant{G:QG}{s,s_{0}}$, which holds iff:
if $s\tran s^{\prime}$, then $W^{\prime}$ is a conservative extension
of $W$, i.e., $W\subseteq W^{\prime}$. 
\item (Clauses grow). $\Invariant{G:CG}{s,s_{0}}$, which holds iff: if
$s\tran s^{\prime}$, then $A^{\prime}$ is a conservative extension
of $A$, i.e., $A\subseteq A^{\prime}$. 
\item (Verified clauses remain verified). $\Invariant{G:VV}{s,s_{0}}$,
which holds iff: suppose $s\tran s^{\prime}$, if $C\in A$ and $\Verify{W,\Egraph E}C$,
then $C\in A^{\prime}$ and $\Verify{W^{\prime},\Egraph{\left(E^{\prime}\right)}}C$. 
\end{enumerate}
\end{definition}

\begin{proposition}
Suppose the initial state is $s_{0}=\state{W_{0}}{\emptyset}{\InjFunc L}$,
where $W_{0}$ is an axiomatisation for a theory with tags, and $L$
is an arbitrary set of ground literals from the same theory. If $s_{0}\tran^{*}s$,
then $\Invariant G{s,s_{0}}$ holds. \label{prop:appx-prop-validity-of-general-purpose-inv}
\begin{proof}
The proof is straightforward by induction on the trace of $s_{0}\tran^{*}s$. 
\end{proof}

\end{proposition}

\section{Proving Instantiation Termination for Set Theory \label{sec:appx-proving-instantiation-termination-set-theory}}

\global\long\def\FuncFont#1{\textit{#1}}%

\global\long\def\Member{\FuncFont{member}}%

\global\long\def\Subset{\FuncFont{subset}}%

\global\long\def\Union{\FuncFont{union}}%

\global\long\def\Inter{\FuncFont{inter}}%

\global\long\def\Diff{\FuncFont{diff}}%

\global\long\def\Add{\FuncFont{add}}%

\global\long\def\Remove{\FuncFont{remove}}%

\global\long\def\IsEmpty{\FuncFont{isEmpty}}%

\global\long\def\Empty{\FuncFont{empty}}%

\global\long\def\Card{\FuncFont{card}}%

\global\long\def\SSet{\FuncFont{Set}}%

\global\long\def\BBool{\FuncFont{Bool}}%

\global\long\def\NNat{\FuncFont{Nat}}%

\global\long\def\Skolem{\FuncFont{Sk}}%

\global\long\def\Singleton{\FuncFont{singleton}}%

\global\long\def\Disjoint{\FuncFont{disjoint}}%

\global\long\def\Equal{\FuncFont{equal}}%

\begin{figure}
\begin{centering}
\begin{tabular}{l|l}
\hline 
Functions & Types\tabularnewline
\hline 
$\Member$ & $T\times\SSet(T)\rightarrow\BBool$\tabularnewline
$\Subset$ & $\SSet(T)\times\SSet(T)\rightarrow\BBool$\tabularnewline
$\Union$ & $\SSet(T)\times\SSet(T)\rightarrow\SSet(T)$\tabularnewline
$\Inter$ & $\SSet(T)\times\SSet(T)\rightarrow\SSet(T)$\tabularnewline
$\Diff$ & $\SSet(T)\times\SSet(T)\rightarrow\SSet(T)$\tabularnewline
$\Add$ & $T\times\SSet(T)\rightarrow\SSet(T)$\tabularnewline
$\Remove$ & $T\times\SSet(T)\rightarrow\SSet(T)$\tabularnewline
$\IsEmpty$ & $\SSet(T)\rightarrow\BBool$\tabularnewline
$\Empty$ & $()\rightarrow\SSet(T)$\tabularnewline
$\Singleton$ & $T\rightarrow\SSet(T)$\tabularnewline
$\Disjoint$ & $\SSet(T)\times\SSet(T)\rightarrow\BBool$\tabularnewline
$\Equal$ & $\SSet(T)\times\SSet(T)\rightarrow\BBool$\tabularnewline
\hline 
\end{tabular}
\par\end{centering}
\caption{Functions in Our Axiomatisation for Set Theory \label{fig:appx-Functions-in-Our-Set-Theory-Ax}}
\end{figure}

We prove instantiation termination for our axiomatisation of set theory.
Our axiomatisation has employed 12 uninterpreted functions, as illustrated
in Fig. \ref{fig:appx-Functions-in-Our-Set-Theory-Ax}. The full axiomatisation
is available in Appx. \ref{subsec:appx-our-axiomatisation-set-theory}. 
\begin{definition}
[Tagging Quantifiers] In our axiomatisation for set theory, each
axiom is identified by its own name as a tag. The inner quantifiers
of nested quantifiers are tagged as follows: \label{def:appx-tagging-quantifiers-set-theory-axioms}
\begin{itemize}
\item The inner quantifier of axiom (subset-elim)
\[
\begin{array}{l}
\forall s_{1},s_{2}.\left[\Subset(s_{1},s_{2})\right]\;\neg\Subset(s_{1},s_{2})\vee\\
\left(\forall x.[\Member(x,s_{1})][\Member(x,s_{2})]\;\neg\Member(x,s_{1})\vee\Member(x,s_{2})\right)
\end{array}
\]
is tagged with $\Tag\text{subset-elim}(s_{1},s_{2})$. 
\item The inner quantifier of axiom (disjoint-elim) 
\[
\begin{array}{l}
\forall s_{1},s_{2}.\left[\Disjoint(s_{1},s_{2})\right]\neg\Disjoint(s_{1},s_{2})\vee\\
\left(\forall x.\left[\Member(x,s_{1})\right]\left[\Member(x,s_{2})\right]\neg\Member(x,s_{1})\vee\neg\Member(x,s_{2})\right)
\end{array}
\]
is tagged with $\Tag\text{disjoint-elim}(s_{1},s_{2})$. 
\item The inner quantifier of axiom (isEmpty-elim-1) 
\[
\forall s.\left[\IsEmpty(s)\right]\;\neg\IsEmpty(s)\vee\left(\forall x.[\Member(x,s)]\;\neg\Member(x,s)\right)
\]
is tagged with $\Tag\text{isEmpty-elim-1}(s)$. 
\end{itemize}
\end{definition}

\global\long\def\filterT{\filter_{T}}%

\global\long\def\filterSet{\filter_{\SSet(T)}}%

\global\long\def\FilterT{O_{2}}%

\global\long\def\FilterSet{O_{1}}%

\begin{definition}
[Overapproximation of Basis] Suppose $B$ is a basis of an E-interface.
The functions $\FilterSet(B)$ and $\FilterT(B)$ denote overapproximations
for the $\SSet(T)$-typed and $T$-typed elements within basis $B$,
respectively. 
\begin{align*}
\FilterSet(B) & =\filterSet(B)\\
\FilterT(B) & =\filterT(B)\\
 & \quad\cup\widehat{\Skolem_{\textit{ss}}}(\FilterSet(B),\FilterSet(B))\cup\widehat{\Skolem_{\textit{eq}}}(\FilterSet(B),\FilterSet(B))\\
 & \quad\cup\widehat{\Skolem_{\textit{dj}}}(\FilterSet(B),\FilterSet(B))\cup\widehat{\Skolem_{\textit{ie}}}(\FilterSet(B))
\end{align*}
Here $\filterSet$ and $\filterT$ take a basis and select its $\SSet(T)$-typed
and $T$-typed elements, respectively; each $\widehat{\Skolem}$ is
lifted from the corresponding $\Skolem$ to support sets. 
\end{definition}

The $\FilterSet$ function selects terms from $B$ that are of type
$\SSet(T)$; the $\FilterT$ function selects terms from $B$ that
are of type $T$ and uses terms of type $\SSet(T)$ to compose Skolem
terms of type $T$. 

\begin{proposition}
Both $\FilterSet$ and $\FilterT$ are monotonic with respect to $\subseteq$. 
\end{proposition}

\begin{proposition}
If $B$ is a basis of an E-interface $\Egraph E$, then $\FilterSet(B)\cup\FilterT(B)$
is a candidate basis of $\Egraph E$. 
\end{proposition}

\begin{definition}
[Overestimation of Enabled E-matches] Consider an arbitrary state
$s=\state WAE$. Let $B$ be a basis of the E-interface $\Egraph E$.
Define an overestimation of the enabled E-matches for state $s$ from
basis $B$ as follows. We call the defined set a $P$-estimation,
and each of its element a $p$-term. 
\[
P(\state WAE,B)=\left\{ \dots p_{\Tag\tau_{i}},\dots,p_{\Tag\tau_{j}\left(\multi r\right)},\dots\right\} 
\]
Here $p_{\Tag\tau_{i}}$ and $p_{\Tag\tau_{j}(\multi r)}$ each denote
a set of tuples that overapproximate the enabled E-matches from the
basis $B$ to the quantifiers labelled with tags $\Tag\tau_{i}$ and
$\Tag\tau_{j}(\multi r)$, respectively. Each tag $\Tag\tau_{i}$
identifies an original axiom from $W$, and each tag $\Tag\tau_{j}(\multi r)$
identifies a quantifier introduced by instantiating the (outer) quantifier
of an original (nested) axiom $\Tag\tau_{j}$ with terms $\multi r$
from the approximations $\FilterSet(B)$ or $\FilterT(B)$. 

To clarify, examples for each category are presented as follows; the
remaining quantifiers shall adhere to the same pattern. 
\begin{itemize}
\item A non-generative axiom: 

$p_{\Tag\text{union-elim}}=\left\{ \left(s_{1},s_{2},x\right)\;\middle\vert\;\begin{array}{l}
s_{1},s_{2}\in\FilterSet(B),x\in\FilterT(B),\\
\Inst E{\left(\Tag\text{union-elim}:(s_{1},s_{2},x)\right)}
\end{array}\right\} $
\item A generative axiom: 

$p_{\Tag\text{subset-intro}}=\left\{ \left(s_{1},s_{2}\right)\;\middle\vert\;\begin{array}{l}
s_{1},s_{2}\in\FilterSet(B),\\
\Inst E{\left(\Tag\text{subset-intro}:\left(s_{1},s_{2}\right)\right)}
\end{array}\right\} $
\item A nested axiom: 

$p_{\Tag\text{subset-elim}}=\left\{ \left(s_{1},s_{2}\right)\;\middle\vert\;\begin{array}{l}
s_{1},s_{2}\in\FilterSet(B),\\
\Inst E{\left(\Tag\text{subset-elim}:\left(s_{1},s_{2}\right)\right)}
\end{array}\right\} $
\item A quantifier introduced by instantiating a nested axiom: 

$p_{\Tag\text{subset-elim(\ensuremath{a,b})}}=\left\{ x\;\middle\vert\;\begin{array}{l}
x\in\FilterT(B),\Inst E{\left(\Tag\text{subset-elim(\ensuremath{a,b})}:x\right)}\end{array}\right\} $

where $a,b\in\FilterSet(B)$ 

\end{itemize}
To further clarify, there are $C_{\Amount{\FilterSet(B)}}^{2}$ $p$-terms
of the form $p_{\Tag\text{subset-elim(\ensuremath{a,b})}}$ with different
$a,b\in\FilterSet(B)$, each of which corresponds to one quantifier
introduced by instantiating the nested axiom (subset-elim) with $\subs{s_{1}}{a}$
and $\subs{s_{2}}{b}$. Note that $\Amount{\FilterSet(B)}$ means
the cardinality of the set $\FilterSet(B)$, and $C_{\Amount{\FilterSet(B)}}^{2}$
means the 2-combination of $\Amount{\FilterSet(B)}$.
\end{definition}

\begin{definition}
[Overestimation on Amount of Enabled E-matches] Define an overestimation
on the amount of the enabled E-matches, denoted $\Sigma$, as a function
from $\textsc{State}$ to $\mathbb{N}\cup\{-1\}$ as follows: 
\[
\Sigma\left(s\right)=\begin{cases}
\underset{p\in P(\state WAE,B)}{\sum}\Amount p & \text{if }s=\state WAE\text{ and }B\text{ is a basis for }\Egraph E\\
-1 & \text{if }s=\bot\\
-1 & \text{if }s=\lozenge
\end{cases}
\]
where $\Amount{\cdot}$ denotes set cardinality. 
\end{definition}

\begin{definition}
[Amount of Unverified Current Clauses] Define the amount of unverified
current clauses, denoted $\Theta$, as a function from $\textsc{State}$
to $\mathbb{N}\cup\{-1\}$ as follows: 
\[
\Theta\left(s\right)=\begin{cases}
\Amount{\left\{ C\in A\;|\;\NotVerify{W,\Egraph E}C\right\} } & \text{if }s=\state WAE\\
-1 & \text{if }s=\bot\\
-1 & \text{if }s=\lozenge
\end{cases}
\]
where $\Amount{\cdot}$ denotes set cardinality. 
\end{definition}

\begin{definition}
[Progress Measure] Define the progress measure, denoted $M$, as
a function from $\textsc{State}$ to $(\mathbb{N}\cup\{-1\})^{2}$
as follows: 
\[
M(s)=\left(\Sigma(s),\,\Theta(s)\right)
\]
where $\left(\Sigma(s),\,\Theta(s)\right)$ is lexicographically ordered. 
\end{definition}

\begin{sidewaysfigure}
\begin{centering}
{\scriptsize{}}%
\begin{tabular}{|c|c|c|c|}
\hline 
{\scriptsize{}Quantifier tag $\tau$} & {\scriptsize{}$\vee^{+}\phi(\multi x)$} & {\scriptsize{}$\multi x$} & {\scriptsize{}$\multi{\left[T\right]}$}\tabularnewline
\hline 
\hline 
{\scriptsize{}\sout{empty}} & {\scriptsize{}\sout{\mbox{$\neg\Member(x,\Empty())$}}}{\scriptsize{}
(not $\vee^{+}$)} & {\scriptsize{}$x$} & {\scriptsize{}$\left[\Member(x,\Empty())\right]$}\tabularnewline
\hline 
{\scriptsize{}singleton-intro-1} & {\scriptsize{}\sout{\mbox{$\Member(x,\Singleton(x))$}}}{\scriptsize{}
(not $\vee^{+}$)} & {\scriptsize{}$x$} & {\scriptsize{}$\left[\Singleton(x)\right]$}\tabularnewline
\hline 
{\scriptsize{}singleton-intro-2} & {\scriptsize{}$\Member(y,\Singleton(x))\vee x\neq y$} & {\scriptsize{}$x,y$} & {\scriptsize{}$\left[\Member(y,\Singleton(x))\right]$}\tabularnewline
\hline 
{\scriptsize{}singleton-elim} & {\scriptsize{}$\neg\Member(y,\Singleton(x))\vee x=y$} & {\scriptsize{}$x,y$} & {\scriptsize{}$\left[\Member(y,\Singleton(x))\right]$}\tabularnewline
\hline 
{\scriptsize{}add-intro-1} & {\scriptsize{}$\Member(y,\Add(x,s))\vee\neg\Member(y,s)$} & {\scriptsize{}$s,x,y$} & {\scriptsize{}$\begin{array}{c}
\left[\Member(y,s),\Add(x,s)\right]\\
\left[\Member(y,\Add(x,s))\right]
\end{array}$}\tabularnewline
\hline 
{\scriptsize{}\sout{add-intro-2}} & {\scriptsize{}\sout{\mbox{$\Member(x,\Add(x,s))$}}}{\scriptsize{}
(not $\vee^{+}$)} & {\scriptsize{}$s,x$} & {\scriptsize{}$\left[\Add(x,s)\right]$}\tabularnewline
\hline 
{\scriptsize{}add-intro-3} & {\scriptsize{}$\Member(y,\Add(x,s))\vee y\neq x$} & {\scriptsize{}$s,x,y$} & {\scriptsize{}$\begin{array}{c}
\left[\Member(y,\Add(x,s))\right]\\
\left[\Member(y,s),\Add(x,s)\right]
\end{array}$}\tabularnewline
\hline 
{\scriptsize{}add-elim} & {\scriptsize{}$\neg\Member(y,\Add(x,s))\vee(x=y)\vee\Member(y,s)$} & {\scriptsize{}$s,x,y$} & {\scriptsize{}$\begin{array}{c}
\left[\Member(y,\Add(x,s))\right]\\
\left[\Member(y,s),\Add(x,s)\right]
\end{array}$}\tabularnewline
\hline 
{\scriptsize{}union-intro-1} & {\scriptsize{}$\Member\left(x,\Union\left(s_{1},s_{2}\right)\right)\vee\neg\Member\left(x,s_{1}\right)$} & {\scriptsize{}$s_{1},s_{2},x$} & {\scriptsize{}$\begin{array}{c}
\left[\Union(s_{1},s_{2}),\Member(x,s_{1})\right]\\
\left[\Member(x,\Union(s_{1},s_{2}))\right]
\end{array}$}\tabularnewline
\hline 
{\scriptsize{}union-intro-2} & {\scriptsize{}$\Member\left(x,\Union\left(s_{1},s_{2}\right)\right)\vee\neg\Member\left(x,s_{2}\right)$} & {\scriptsize{}$s_{1},s_{2},x$} & {\scriptsize{}$\begin{array}{c}
\left[\Union(s_{1},s_{2}),\Member(x,s_{2})\right]\\
\left[\Member(x,\Union(s_{1},s_{2}))\right]
\end{array}$}\tabularnewline
\hline 
{\scriptsize{}union-elim} & {\scriptsize{}$\neg\Member\left(x,\Union\left(s_{1},s_{2}\right)\right)\vee\Member(x,s_{1})\vee\Member(x,s_{2})$} & {\scriptsize{}$s_{1},s_{2},x$} & {\scriptsize{}$\begin{array}{c}
\left[\Member(x,\Union(s_{1},s_{2}))\right]\\
\left[\Union(s_{1},s_{2}),\Member(x,s_{1})\right]\\
\left[\Union(s_{1},s_{2}),\Member(x,s_{2})\right]
\end{array}$}\tabularnewline
\hline 
\multirow{2}{*}{{\scriptsize{}union-disjoint}} & {\scriptsize{}$\neg\Disjoint(s_{1},s_{2})\vee\left(\Diff(\Union(s_{1},s_{2}),s_{1})=s_{2}\right)$} & \multirow{2}{*}{{\scriptsize{}$s_{1},s_{2}$}} & \multirow{2}{*}{{\scriptsize{}$\left[\Union(s_{1},s_{2})\right]$}}\tabularnewline
\cline{2-2} 
 & {\scriptsize{}$\neg\Disjoint(s_{1},s_{2})\vee\left(\Diff(\Union(s_{1},s_{2}),s_{2})=s_{1}\right)$} &  & \tabularnewline
\hline 
{\scriptsize{}inter-intro} & {\scriptsize{}$\Member(x,\Inter(s_{1},s_{2}))\vee\neg\Member(x,s_{1})\vee\neg\Member(x,s_{2})$} & {\scriptsize{}$s_{1},s_{2},x$} & {\scriptsize{}$\begin{array}{c}
\left[\Member(x,s_{1}),\Inter(s_{1},s_{2})\right]\\
\left[\Member(x,s_{2}),\Inter(s_{1},s_{2})\right]\\
\left[\Member(x,\Inter(s_{1},s_{2}))\right]
\end{array}$}\tabularnewline
\hline 
\multirow{2}{*}{{\scriptsize{}inter-elim}} & {\scriptsize{}$\begin{array}{c}
\neg\Member(x,\Inter(s_{1},s_{2}))\vee\Member(x,s_{1})\\
\\
\end{array}$} & \multirow{2}{*}{{\scriptsize{}$s_{1},s_{2},x$}} & \multirow{2}{*}{{\scriptsize{}$\begin{array}{c}
\left[\Member(x,s_{1}),\Inter(s_{1},s_{2})\right]\\
\left[\Member(x,s_{2}),\Inter(s_{1},s_{2})\right]\\
\left[\Member(x,\Inter(s_{1},s_{2}))\right]
\end{array}$}}\tabularnewline
\cline{2-2} 
 & {\scriptsize{}$\begin{array}{c}
\\
\neg\Member(x,\Inter(s_{1},s_{2}))\vee\Member(x,s_{2})
\end{array}$} &  & \tabularnewline
\hline 
{\scriptsize{}union-right} & {\scriptsize{}\sout{\mbox{$\Union(\Union(s_{1},s_{2}),s_{2})=\Union(s_{1},s_{2})$}}}{\scriptsize{}
(not $\vee^{+}$)} & {\scriptsize{}$s_{1},s_{2}$} & {\scriptsize{}$\left[\Union(\Union(s_{1},s_{2}),s_{2})\right]$}\tabularnewline
\hline 
{\scriptsize{}union-left} & {\scriptsize{}\sout{\mbox{$\Union(s_{1},\Union(s_{1},s_{2}))=\Union(s_{1},s_{2})$}}}{\scriptsize{}
(not $\vee^{+}$)} & {\scriptsize{}$s_{1},s_{2}$} & {\scriptsize{}$\left[\Union(s_{1},\Union(s_{1},s_{2}))\right]$}\tabularnewline
\hline 
{\scriptsize{}inter-right} & {\scriptsize{}\sout{\mbox{$\Inter(\Inter(s_{1},s_{2}),s_{2})=\Inter(s_{1},s_{2})$}}}{\scriptsize{}
(not $\vee^{+}$)} & {\scriptsize{}$s_{1},s_{2}$} & {\scriptsize{}$\left[\Inter(\Inter(s_{1},s_{2}),s_{2})\right]$}\tabularnewline
\hline 
{\scriptsize{}inter-left} & {\scriptsize{}\sout{\mbox{$\Inter(s_{1},\Inter(s_{1},s_{2}))=\Inter(s_{1},s_{2})$}}}{\scriptsize{}
(not $\vee^{+}$)} & {\scriptsize{}$s_{1},s_{2}$} & {\scriptsize{}$\left[\Inter(s_{1},\Inter(s_{1},s_{2}))\right]$}\tabularnewline
\hline 
\end{tabular}{\scriptsize\par}
\par\end{centering}
\caption{Disjunctions Lookup Table (1) \label{fig:Disjunctions-Lookup-Table-1}}
\end{sidewaysfigure}

\begin{sidewaysfigure}
\begin{centering}
{\scriptsize{}}%
\begin{tabular}{|c|c|c|c|}
\hline 
{\scriptsize{}Quantifier tag $\tau$} & {\scriptsize{}$\vee^{+}\phi(\multi x)$} & {\scriptsize{}$\multi x$} & {\scriptsize{}$\multi{\left[T\right]}$}\tabularnewline
\hline 
\hline 
{\scriptsize{}diff-intro} & {\scriptsize{}$\Member(x,\Diff(s_{1},s_{2}))\vee\neg\Member(x,s_{1})\vee\Member(x,s_{2})$} & {\scriptsize{}$s_{1},s_{2},x$} & {\scriptsize{}$\begin{array}{c}
\left[\Member(x,s_{1}),\Diff(s_{1},s_{2})\right]\\{}
[\Member(x,s_{2}),\Diff(s_{1},s_{2})]\\
\left[\Member(x,\Diff(s_{1},s_{2}))\right]
\end{array}$}\tabularnewline
\hline 
\multirow{2}{*}{{\scriptsize{}diff-elim}} & {\scriptsize{}$\begin{array}{c}
\neg\Member(x,\Diff(s_{1},s_{2}))\vee\Member(x,s_{1})\\
\\
\end{array}$} & \multirow{2}{*}{{\scriptsize{}$s_{1},s_{2},x$}} & \multirow{2}{*}{{\scriptsize{}$\begin{array}{c}
\left[\Member(x,s_{1}),\Diff(s_{1},s_{2})\right]\\{}
[\Member(x,s_{2}),\Diff(s_{1},s_{2})]\\
\left[\Member(x,\Diff(s_{1},s_{2}))\right]
\end{array}$}}\tabularnewline
\cline{2-2} 
 & {\scriptsize{}$\begin{array}{c}
\\
\neg\Member(x,\Diff(s_{1},s_{2}))\vee\neg\Member(x,s_{2})
\end{array}$} &  & \tabularnewline
\hline 
\multirow{2}{*}{{\scriptsize{}subset-intro (Sk)}} & {\scriptsize{}$\Subset(s_{1},s_{2})\vee\Member(\Skolem_{ss}(s_{1},s_{2}),s_{1})$} & {\scriptsize{}$s_{1}$, $s_{2}$} & \multirow{2}{*}{{\scriptsize{}$\left[\Subset(s_{1},s_{2})\right]$}}\tabularnewline
\cline{2-3} \cline{3-3} 
 & {\scriptsize{}$\Subset(s_{1},s_{2})\vee\neg\Member(\Skolem_{ss}(s_{1},s_{2}),s_{2})$} & {\scriptsize{}$s_{1}$, $s_{2}$} & \tabularnewline
\hline 
{\scriptsize{}subset-elim (nested)} & {\scriptsize{}$\begin{array}{l}
\neg\Subset(s_{1},s_{2})\vee\\
\left(\forall x.[\Member(x,s_{1})][\Member(x,s_{2})]\;\neg\Member(x,s_{1})\vee\Member(x,s_{2})\right)
\end{array}$} & {\scriptsize{}$s_{1},s_{2}$} & {\scriptsize{}$\left[\Subset(s_{1},s_{2})\right]$}\tabularnewline
\hline 
{\scriptsize{}subset-elim$(a,b)$ where} & \multirow{2}{*}{{\scriptsize{}$\neg\Member(x,a)\vee\Member(x,b)$}} & \multirow{2}{*}{{\scriptsize{}$x$}} & \multirow{2}{*}{{\scriptsize{}$\begin{array}{c}
\left[\Member(x,a)\right]\\
\left[\Member(x,b)\right]
\end{array}$}}\tabularnewline
{\scriptsize{}$a,b\in\FilterSet(B)$} &  &  & \tabularnewline
\hline 
\multirow{2}{*}{{\scriptsize{}equal-sets-intro (Sk)}} & {\scriptsize{}$\Equal(s_{1},s_{2})\vee\Member(\Skolem_{\textit{eq}}(s_{1},s_{2}),s_{1})\vee\Member(\Skolem_{\textit{eq}}(s_{1},s_{2}),s_{2})$} & \multirow{2}{*}{{\scriptsize{}$s_{1},s_{2}$}} & \multirow{2}{*}{{\scriptsize{}$\left[\Equal(s_{1},s_{2})\right]$}}\tabularnewline
\cline{2-2} 
 & {\scriptsize{}$\Equal(s_{1},s_{2})\vee\neg\Member(\Skolem_{\textit{eq}}(s_{1},s_{2}),s_{1})\vee\neg\Member(\Skolem_{\textit{eq}}(s_{1},s_{2}),s_{2})$} &  & \tabularnewline
\hline 
{\scriptsize{}equal-sets-extensionality} & {\scriptsize{}$\neg\Equal(s_{1},s_{2})\vee s_{1}=s_{2}$} & {\scriptsize{}$s_{1},s_{2}$} & {\scriptsize{}$\left[\Equal(s_{1},s_{2})\right]$}\tabularnewline
\hline 
\multirow{2}{*}{{\scriptsize{}disjoint-intro (Sk)}} & {\scriptsize{}$\Disjoint(s_{1},s_{2})\vee\Member(\Skolem_{\textit{dj}}(s_{1},s_{2}),s_{1})$} & \multirow{2}{*}{{\scriptsize{}$s_{1},s_{2}$}} & \multirow{2}{*}{{\scriptsize{}$\left[\Disjoint(s_{1},s_{2})\right]$}}\tabularnewline
\cline{2-2} 
 & {\scriptsize{}$\Disjoint(s_{1},s_{2})\vee\Member(\Skolem_{\textit{dj}}(s_{1},s_{2}),s_{2})$} &  & \tabularnewline
\hline 
{\scriptsize{}disjoint-elim (nested)} & {\scriptsize{}$\begin{array}{l}
\neg\Disjoint(s_{1},s_{2})\vee\\
\left(\forall x.\left[\Member(x,s_{1})\right]\left[\Member(x,s_{2})\right]\neg\Member(x,s_{1})\vee\neg\Member(x,s_{2})\right)
\end{array}$} & {\scriptsize{}$s_{1},s_{2}$} & {\scriptsize{}$\left[\Disjoint(s_{1},s_{2})\right]$}\tabularnewline
\hline 
{\scriptsize{}disjoint-elim$(a,b)$ where} & \multirow{2}{*}{{\scriptsize{}$\neg\Member(x,a)\vee\neg\Member(x,b)$}} & \multirow{2}{*}{{\scriptsize{}$x$}} & \multirow{2}{*}{{\scriptsize{}$\begin{array}{c}
\left[\Member(x,a)\right]\\
\left[\Member(x,b)\right]
\end{array}$}}\tabularnewline
{\scriptsize{}$a,b\in\FilterSet(B)$} &  &  & \tabularnewline
\hline 
{\scriptsize{}remove-intro-1} & {\scriptsize{}$y=x\vee\neg\Member(y,s)\vee\Member(y,\Remove(x,s))$} & {\scriptsize{}$s,x,y$} & {\scriptsize{}$\begin{array}{c}
\left[\Member(y,s),\Remove(x,s)\right]\\
\left[\Member(y,\Remove(x,s))\right]
\end{array}$}\tabularnewline
\hline 
{\scriptsize{}remove-intro-2} & {\scriptsize{}\sout{\mbox{$\neg\Member(x,\Remove(x,s))$}}}{\scriptsize{}
(not $\vee^{+}$)} & {\scriptsize{}$s,x$} & {\scriptsize{}$\left[\Remove(x,s)\right]$}\tabularnewline
\hline 
{\scriptsize{}remove-intro-3} & {\scriptsize{}$\neg\Member(y,\Remove(x,s))\vee y\neq x$} & {\scriptsize{}$s,x$} & {\scriptsize{}$\begin{array}{c}
\left[\Member(y,s),\Remove(x,s)\right]\\
\left[\Member(y,\Remove(x,s))\right]
\end{array}$}\tabularnewline
\hline 
\multirow{2}{*}{{\scriptsize{}remove-elim}} & {\scriptsize{}$\begin{array}{c}
\neg\Member(y,\Remove(x,s))\vee y\neq x\\
\\
\end{array}$} & \multirow{2}{*}{{\scriptsize{}$s,x,y$}} & \multirow{2}{*}{{\scriptsize{}$\begin{array}{c}
\left[\Member(y,s),\Remove(x,s)\right]\\
\left[\Member(y,\Remove(x,s))\right]
\end{array}$}}\tabularnewline
\cline{2-2} 
 & {\scriptsize{}$\begin{array}{c}
\\
\neg\Member(y,\Remove(x,s))\vee\Member(y,s)
\end{array}$} &  & \tabularnewline
\hline 
{\scriptsize{}isEmpty-intro-1 (Sk)} & {\scriptsize{}$\IsEmpty(s)\vee\Member(\Skolem_{ie}(s),s)$} & {\scriptsize{}$s$} & {\scriptsize{}$\left[\IsEmpty(s)\right]$}\tabularnewline
\hline 
{\scriptsize{}isEmpty-intro-2} & {\scriptsize{}$\IsEmpty(s)\vee\neg\Equal(s,\Empty())$} & {\scriptsize{}$s$} & {\scriptsize{}$\begin{array}{c}
\left[\IsEmpty(s)\right]\\
\left[\Equal(s,\Empty())\right]
\end{array}$}\tabularnewline
\hline 
{\scriptsize{}isEmpty-elim-1 (nested)} & {\scriptsize{}$\neg\IsEmpty(s)\vee\forall x.[\Member(x,s)]\;\neg\Member(x,s)$} & {\scriptsize{}$s$} & {\scriptsize{}$\left[\IsEmpty(s)\right]$}\tabularnewline
\hline 
{\scriptsize{}isEmpty-elim-1$(a)$} & \multirow{2}{*}{{\scriptsize{}\sout{\mbox{$\neg\Member(x,a)$}}}{\scriptsize{}
(not $\vee^{+}$)}} & \multirow{2}{*}{{\scriptsize{}$x$}} & \multirow{2}{*}{{\scriptsize{}$[\Member(x,a)]$}}\tabularnewline
{\scriptsize{}where $a\in\FilterSet(B)$} &  &  & \tabularnewline
\hline 
{\scriptsize{}isEmpty-elim-2} & {\scriptsize{}$\neg\IsEmpty(s)\vee\Equal(s,\Empty())$} & {\scriptsize{}$s$} & {\scriptsize{}$\begin{array}{c}
\left[\IsEmpty(s)\right]\\
\left[\Equal(s,\Empty())\right]
\end{array}$}\tabularnewline
\hline 
\end{tabular}{\scriptsize\par}
\par\end{centering}
\caption{Disjunctions Lookup Table (2)\label{fig:Disjunctions-Lookup-Table-2}}
\end{sidewaysfigure}

\begin{definition}
[General-purpose Invariants] Cf. Def. \ref{def:appx-inv-general-purpose}. 
\end{definition}

\begin{definition}
[Problem-specific Invariants] Suppose the initial state is $s_{0}=\state{W_{0}}{\emptyset}{E_{0}}$,
where $W_{0}$ is our axiomatisation for set theory with each quantifier
tagged as specified by Def. \ref{def:appx-tagging-quantifiers-set-theory-axioms},
$E_{0}=\InjFunc L$, and $L$ is an arbitrary set of ground literals
from set theory. \label{def:appx-inv-problem-specific}

For an arbitrary intermediate state $s=\state WAE$, define the problem-specific
invariant $I_{P}(s,s_{0})$ to be a conjunction of the following predicates: 
\begin{enumerate}
\item (Origins of clauses). $I_{\mathrm{P:OC}}(s,s_{0})$, which holds
iff: for each extended clause in $A$ in the form of $\vee^{+}\phi(\multi x)\left[\nicefrac{\multi r}{\multi x}\right]$
where $\vee^{+}\phi(\multi x)$ matches at least one entry of the
$\vee^{+}\phi(\multi x)$ column of Fig. \ref{fig:Disjunctions-Lookup-Table-1}
and \ref{fig:Disjunctions-Lookup-Table-2}, $\left(\Tag\tau:\multi r\right)\in\Ehistory E$
holds for the quantifier identified by $\Tag\tau$ as implied by the
same entry. To clarify, here are some examples. 
\begin{enumerate}
\item For each extended clause in $A$ of the form $\neg\Member\left(t,\Union\left(a,b\right)\right)\vee\Member(t,a)\vee\Member(t,b)$,
$\left(\Tag\text{union-elim}:(a,b,t)\right)\in\Ehistory E$ holds,
where axiom (union-elim) is defined as follows: 
\[
\begin{array}{l}
\forall s_{1},s_{2},x.\\
\left[\Member(x,\Union(s_{1},s_{2}))\right]\\
\left[\Union(s_{1},s_{2}),\Member(x,s_{1})\right]\left[\Union(s_{1},s_{2}),\Member(x,s_{2})\right]\\
\;\neg\Member\left(x,\Union\left(s_{1},s_{2}\right)\right)\vee\Member(x,s_{1})\vee\Member(x,s_{2})
\end{array}
\]
\item For each extended clause in $A$ of the form $\Subset(a,b)\vee\Member(\Skolem_{ss}(a,b),a)$,
$\left(\Tag\text{subset-intro}:(a,b)\right)\in\Ehistory E$ holds,
where axiom (subset-intro) is defined as follows: 
\[
\begin{array}{l}
\forall s_{1},s_{2}.\left[\Subset(s_{1},s_{2})\right]\\
\left(\Subset(s_{1},s_{2})\vee\Member(\Skolem_{ss}(s_{1},s_{2}),s_{1})\right)\wedge\\
\left(\Subset(s_{1},s_{2})\vee\neg\Member(\Skolem_{ss}(s_{1},s_{2}),s_{2})\right)
\end{array}
\]
\item For each extended clause in $A$ of the form 
\[
\begin{array}{l}
\neg\Subset(a,b)\vee\\
\left(\forall x.[\Member(x,a)][\Member(x,b)]\;\neg\Member(x,a)\vee\Member(x,b)\right),
\end{array}
\]
$\left(\Tag\text{subset-elim}:(a,b)\right)\in\Ehistory E$ holds,
where axiom (subset-elim) is defined as follows: 
\[
\begin{array}{l}
\forall s_{1},s_{2}.\left[\Subset(s_{1},s_{2})\right]\;\neg\Subset(s_{1},s_{2})\vee\\
\left(\forall x.[\Member(x,s_{1})][\Member(x,s_{2})]\;\neg\Member(x,s_{1})\vee\Member(x,s_{2})\right)
\end{array}
\]
\item For each extended clause in $A$ of the form $\neg\Member(t,a)\vee\Member(t,b)$,
at least one of the following holds:
\begin{enumerate}
\item $\left(\Tag\text{subset-elim}(a,b):t\right)\in\Ehistory E$ holds,
where quantifier (subset-elim$(a,b)$) is defined as follows: 
\[
\forall x.[\Member(x,a)][\Member(x,b)]\;\neg\Member(x,a)\vee\Member(x,b)
\]
\item $\left(\Tag\text{disjoint-elim}(a,b):t\right)\in\Ehistory E$ holds,
where quantifier (disjoint-elim$(a,b)$) is defined as follows: 
\[
\forall x.\left[\Member(x,a)\right]\left[\Member(x,b)\right]\neg\Member(x,a)\vee\neg\Member(x,b)
\]
\end{enumerate}
\end{enumerate}
\item (Forms of quantifiers in clauses). $I_{\mathrm{P:FQ}}(s,s_{0})$,
which holds iff: every quantifier $\phi$ in any extended clause of
$A$ must be in one of the following forms: 
\begin{enumerate}
\item $\forall x.[\Member(x,a)][\Member(x,b)]\;\neg\Member(x,a)\vee\Member(x,b)$
with tag $\Tag\text{subset-elim}\left(a,b\right)$, and $\left(\Tag\text{subset-elim}:\left(a,b\right)\right)\in\Ehistory E$,
for some sets $a$ and $b$; 
\item $\forall x.\left[\Member(x,a)\right]\left[\Member(x,b)\right]\neg\Member(x,a)\vee\neg\Member(x,b)$
with tag $\Tag\text{disjoint-elim}\left(a,b\right)$, and $\left(\Tag\text{subset-disjoint}:\left(a,b\right)\right)\in\Ehistory E$,
for some sets $a$ and $b$; 
\item $\forall x.[\Member(x,a)]\;\neg\Member(x,a)$, with tag $\Tag\text{isEmpty-elim}(a)$,
and $\left(\Tag\text{isEmpty-elim}:a\right)\in\Ehistory E$, for some
set $a$. 
\end{enumerate}
\item (Instances of clauses). $I_{\mathrm{P:IC}}(s,s_{0})$, which holds
iff: each extended clause $\vee^{+}\phi$ in $A$ must be from the
$\vee^{+}\phi(\multi x)$ column of Figures \ref{fig:Disjunctions-Lookup-Table-1}
and \ref{fig:Disjunctions-Lookup-Table-2} with appropriate substitutions
for variables indicated by the $\multi x$ column.
\item (Inherited quantifiers). $I_{\mathrm{P:IQ}}(s,s_{0})$, which holds
iff: either 
\begin{enumerate}
\item $W=W_{0}$, or
\item $W$ is the disjoint union of $W_{0}$ and $W^{\prime}$, and for
each $w\in W^{\prime}$, one of the following holds: 
\begin{enumerate}
\item $w$ is a quantifier $\forall x.[\Member(x,a)][\Member(x,b)]\;\neg\Member(x,a)\vee\Member(x,b)$,
and $\left(\Tag\text{subset-elim}:\left(a,b\right)\right)\in\Ehistory E$,
for some sets $a$ and $b$. 
\item $w$ is a quantifier $\forall x.\left[\Member(x,a)\right]\left[\Member(x,b)\right]\neg\Member(x,a)\vee\neg\Member(x,b)$,
and $\left(\Tag\text{disjoint-elim}:\left(a,b\right)\right)\in\Ehistory E$,
for some sets $a$ and $b$. 
\item $w$ is a quantifier $\forall x.[\Member(x,a)]\;\neg\Member(x,a)$,
and $(\Tag\text{isEmpty-elim}:a)\in\Ehistory E$, for some set $a$. 
\end{enumerate}
\end{enumerate}
\item (Basis after a step). $I_{\mathrm{P:BS}}(s,s_{0})$, which holds
iff: if $B_{E}$ is a basis for $\Egraph E$, and $s\tran s^{\prime}$,
then $\FilterSet(B_{E})\cup\FilterT(B_{E})$ is a candidate basis
for $\Egraph{\left(E^{\prime}\right)}$. 
\item (Inherited basis). $I_{\mathrm{P:IB}}(s,s_{0})$, which holds iff:
if $B_{E_{0}}$ is a basis for $\Egraph{E_{0}}$, then $\FilterSet(B_{E_{0}})\cup\FilterT(B_{E_{0}})$
is a candidate basis for $\Egraph E$. 
\end{enumerate}
\end{definition}

\begin{remark}
We write $\vee^{+}\phi_{i}$ to denote $\phi_{1}\vee\dots\vee\phi_{i}\vee\dots\vee\phi_{n}$
where $n\geq2$. 
\end{remark}

We prove an interesting problem-specific invariant $\Invariant{P:BS}{s,s_{0}}$
from Def. \ref{def:appx-inv-problem-specific}. Proofs of other invariants
are straightforward or analogous. 
\begin{proposition}
Suppose the initial state is $s_{0}=\state{W_{0}}{\emptyset}{E_{0}}$,
where $W_{0}$ is our axiomatisation for set theory with each quantifier
tagged as specified by Def. \ref{def:appx-tagging-quantifiers-set-theory-axioms},
$E_{0}=\InjFunc L$, and $L$ is an arbitrary set of ground literals
from set theory. Let $B_{E_{0}}$ be a basis for $E_{0}$. \label{prop:appx-prop-P:BS-basis-after-a-step}

Let $I(s,s_{0})=I_{G}(s,s_{0})\wedge I_{\mathrm{P:OC}}(s,s_{0})\wedge I_{\mathrm{P:FQ}}(s,s_{0})\wedge I_{\mathrm{P:IC}}(s,s_{0})\wedge I_{\mathrm{P:IQ}}(s,s_{0})$. 

Suppose $s_{1}\tran s_{2}$ and $I(s_{1},s_{0})$ holds. 

If $B_{E_{1}}$ is a basis for $\Egraph{E_{1}}$, then $\FilterSet(B_{E_{1}})\cup\FilterT(B_{E_{1}})$
is a candidate basis for $\Egraph{E_{2}}$. 
\end{proposition}

\begin{proof}
Proceed by cases on $s_{1}\tran s_{2}$. Both $\textsc{(bot)}$ and
$\textsc{(sat)}$ cases are vacuous. The remaining cases are $\textsc{(split)}$
and $\textsc{(inst)}$. 

We first discuss the $\textsc{(split)}$ case.

Let $s_{1}=\state{W_{1}}{A_{1}}{E_{1}}$.

Let $\Phi\subseteq\left\{ \phi_{i}\mid C\in A_{1},\;\NotVerify{W_{1},\Egraph{E_{1}}}C,\;C\text{ is }\phi_{1}\vee\dots\vee\phi_{i}\vee\dots\vee\phi_{n},\;n\geq2\right\} $.
Refer to the $\vee^{+}\phi(\multi x)$ column of the tables in Figures
\ref{fig:Disjunctions-Lookup-Table-1} and \ref{fig:Disjunctions-Lookup-Table-2}
for possible constructions of $\Phi$. 

Let $s_{2}=\state{W_{2}}{A_{2}}{E_{2}}$, where $\Egraph{E_{2}}=\updateEgraph{\Egraph{E_{1}}}{\filterLit{\Phi}}$. 

Proceed by induction on $\Phi$. 
\begin{case}
($\Phi=\emptyset$). Then $\Egraph{E_{2}}=\Egraph{E_{1}}$. Since
$B_{E_{1}}$ is a basis for $\Egraph{E_{1}}$, then $B_{E_{1}}$ must
be a basis (thus also a candidate basis) for $\Egraph{E_{2}}$. 
\end{case}

\begin{case}
($\Phi=\left\{ \left\{ \phi_{1}\right\} \right\} $). 

Proceed by cases on $\phi_{1}$. We choose $\phi_{1}=\Member(\Skolem_{ss}(s_{1}',s_{2}'),s_{1}')$
as an example; the other cases are analogous to this case. 
\begin{enumerate}
\item $\phi_{1}=\Member(\Skolem_{ss}(s_{1}',s_{2}'),s_{1}')$ must be chosen
from \\
$\Subset(s_{1}',s_{2}')\vee\Member(\Skolem_{ss}(s_{1}',s_{2}'),s_{1}')\in A_{1}$,
by $I_{\mathrm{P:IC}}(s_{1},s_{0})$. 
\item $\left(\Tag\text{subset-intro}:\left(s_{1}',s_{2}'\right)\right)\in\Ehistory{E_{1}}$
by $I_{\mathrm{P:OC}}(s_{1},s_{0})$. 
\item $\Known{\Egraph{E_{1}}}{\left(s_{1}',s_{2}'\right)}$ by $\Invariant{G:EE}{s_{1},s_{0}}$.
Then, $\Known{\Egraph{E_{1}}}{s_{1}'}$ and $\Known{\Egraph{E_{1}}}{s_{2}'}$. 
\item Since $B_{E_{1}}$ is a basis for $\Egraph{E_{1}}$, there are $b_{1},b_{2}\in B_{E_{1}}$
such that $\Entail{\Egraph{E_{1}}}{\Equiv{s_{1}'}{b_{1}}}$ and $\Entail{\Egraph{E_{1}}}{\Equiv{s_{2}'}{b_{2}}}$.
Note $b_{1}$ and $b_{2}$ may or may not be identical to each other. 
\item Now $\Egraph{E_{2}}=\Egraph{E_{1}}\triangleleft\phi_{1}=\Egraph{E_{1}}\cup\{\Member(\Skolem_{ss}(s_{1}',s_{2}'),s_{1}')\sim\top\}$. 
\item There exists $b\in\FilterSet(B_{E_{1}})\cup\FilterT(B_{E_{1}})$ such
that $\Entail{\Egraph{E_{2}}}{\Skolem_{ss}(s_{1}',s_{2}')}\sim b$.
In particular, $\Entail{\Egraph{E_{2}}}{\Skolem_{ss}(s_{1}',s_{2}')\sim\Skolem_{ss}(b_{1},b_{2})}$
and $\Skolem_{ss}(b_{1},b_{2})\in\FilterSet(B_{E_{1}})\cup\FilterT(B_{E_{1}})$. 
\item Our goal is to show that $\FilterSet(B_{E_{1}})\cup\FilterT(B_{E_{1}})$
is a candidate basis for $\Egraph{E_{2}}$. That is, for every $t$,
if $\Known{\Egraph{E_{2}}}t$, then there exists some $b\in\FilterSet(B_{E_{1}})\cup\FilterT(B_{E_{1}})$
such that $\Entail{\Egraph{E_{2}}}{\Equiv tb}$. \\
Assume $\Known{\Egraph{E_{2}}}t$, that is $\Known{\Egraph{E_{1}}\cup\{\Member(\Skolem_{ss}(s_{1}',s_{2}'),s_{1}')\sim\top\}}t$.
Then either $\Known{\Egraph{E_{1}}}t$, or $t$ is $\Skolem_{ss}(s_{1}',s_{2}')$.
The former case is immediate, and the latter case has been handled. 
\end{enumerate}
\end{case}

\begin{case}
($\Phi=\Phi^{\prime}\cup\left\{ \left\{ \phi_{1}\right\} \right\} $
and $\left\{ \phi_{1}\right\} \notin\Phi^{\prime}$). Let $\Egraph{E_{12}}=\updateEgraph{\Egraph{E_{1}}}{\filterLit{\Phi^{\prime}}}$.
Since $B_{E_{1}}$ is a basis for $\Egraph{E_{1}}$, by the inductive
hypothesis, $\FilterSet(B_{E_{1}})\cup\FilterT(B_{E_{1}})$ must be
a candidate basis for $\Egraph{E_{12}}$. Let $B_{E_{12}}\subseteq\FilterSet(B_{E_{1}})\cup\FilterT(B_{E_{1}})$. 

Proceed by cases on $\phi_{1}$. We choose $\phi_{1}=\Member(\Skolem_{ss}(s_{1}',s_{2}'),s_{1}')$
as an example; the other cases are analogous to this case.
\begin{enumerate}
\item $\phi_{1}=\Member(\Skolem_{ss}(s_{1}',s_{2}'),s_{1}')$ must be chosen
from \\
$\Subset(s_{1}',s_{2}')\vee\Member(\Skolem_{ss}(s_{1}',s_{2}'),s_{1}')\in A_{1}$,
by $I_{\mathrm{P:IC}}(s_{1},s_{0})$. 
\item $\left(\Tag\text{subset-intro}:\left(s_{1}',s_{2}'\right)\right)\in\Ehistory{E_{1}}$
by $I_{\mathrm{P:OC}}(s_{1},s_{0})$. 
\item $\Known{\Egraph{E_{1}}}{\left(s_{1}',s_{2}'\right)}$ by $\Invariant{G:EE}{s_{1},s_{0}}$.
Then, $\Known{\Egraph{E_{1}}}{s_{1}'}$ and $\Known{\Egraph{E_{1}}}{s_{2}'}$. 
\item Then, $\Known{\Egraph{E_{12}}}{s_{1}'}$ and $\Known{\Egraph{E_{12}}}{s_{2}'}$. 
\item Since $\FilterSet(B_{E_{1}})\cup\FilterT(B_{E_{1}})$ is a candidate
basis for $\Egraph{E_{12}}$, there are $b_{1},b_{2}\in\FilterSet(B_{E_{1}})\cup\FilterT(B_{E_{1}})$
such that $\Entail{\Egraph{E_{12}}}{\Equiv{s_{1}'}{b_{1}}}$ and $\Entail{\Egraph{E_{12}}}{\Equiv{s_{2}'}{b_{2}}}$.
Note $b_{1}$ and $b_{2}$ may or may not be identical to each other. 
\item Now $\Egraph{E_{2}}=\Egraph{E_{12}}\triangleleft\phi_{1}=\Egraph{E_{12}}\cup\{\Member(\Skolem_{ss}(s_{1}',s_{2}'),s_{1}')\sim\top\}$. 
\item There exists $b\in\FilterSet(B_{E_{1}})\cup\FilterT(B_{E_{1}})$ such
that $\Entail{\Egraph{E_{2}}}{\Skolem_{ss}(s_{1}',s_{2}')}\sim b$.
In particular, $\Entail{\Egraph{E_{2}}}{\Skolem_{ss}(s_{1}',s_{2}')\sim\Skolem_{ss}(b_{1},b_{2})}$
and $\Skolem_{ss}(b_{1},b_{2})\in\FilterSet(B_{E_{1}})\cup\FilterT(B_{E_{1}})$. 
\item Our goal is to show that $\FilterSet(B_{E_{1}})\cup\FilterT(B_{E_{1}})$
is a candidate basis for $\Egraph{E_{2}}$. That is, for every $t$,
if $\Known{\Egraph{E_{2}}}t$, then there exists some $b\in\FilterSet(B_{E_{1}})\cup\FilterT(B_{E_{1}})$
such that $\Entail{\Egraph{E_{2}}}{\Equiv tb}$. \\
Assume $\Known{\Egraph{E_{2}}}t$, that is $\Known{\Egraph{E_{12}}\cup\{\Member(\Skolem_{ss}(s_{1}',s_{2}'),s_{1}')\sim\top\}}t$.
Then either $\Known{\Egraph{E_{12}}}t$, or $t$ is $\Skolem_{ss}(s_{1}',s_{2}')$.
The former case is immediate from the inductive hypothesis, and the
latter case has been handled. 
\end{enumerate}
\end{case}

We then discuss the $\textsc{(inst)}$ case. Let $s_{1}=\state{W_{1}}{A_{1}}{E_{1}}$
and $s_{2}=\state{W_{2}}{A_{2}}{E_{2}}$, where $\state{W_{1}}{A_{1}}{E_{1}}\Ematching\left(\unitriMulti xT{A_{11}}\right)^{\Tag\alpha}\MatchSep\multi r$,
$A_{12}=A_{11}\left[\nicefrac{\multi r}{\multi x}\right]$, $\Egraph{E_{2}}=\Egraph{E_{1}}\triangleleft\filterLit{A_{12}}$. 

By $I_{\mathrm{P:IQ}}(s_{1},s_{0})$, proceed by cases on $\left(\unitriMulti xT{A_{11}}\right)^{\Tag\alpha}\in W_{1}$.
For the majority cases---quantifiers whose $\vee^{+}\phi$ columns
in Fig. \ref{fig:Disjunctions-Lookup-Table-1} and \ref{fig:Disjunctions-Lookup-Table-2}
are not crossed out, $\filterLit{A_{12}}=\emptyset$. Thus $\Egraph{E_{2}}=\Egraph{E_{1}}$.
It is immediate that $\FilterSet(B_{E_{1}})\cup\FilterT(B_{E_{1}})$
is a candidate basis for $\Egraph{E_{2}}$. 

For the remaining cases---quantifiers whose $\vee^{+}\phi$ columns
in Fig. \ref{fig:Disjunctions-Lookup-Table-1} and \ref{fig:Disjunctions-Lookup-Table-2}
are crossed out, we choose (add-intro-2) as an example; the other
cases are analogous to this case. 
\begin{enumerate}
\item $\state{W_{1}}{A_{1}}{E_{1}}\Ematching\left(\forall s,x.\left[\Add(x,s)\right]\;\Member(x,\Add(x,s))\right)^{\Tag\text{add-intro-2}}\MatchSep(s',x')$. 
\item Then, $\Known{\Egraph{E_{1}}}{\Add(x',s')}$, $\Known{\Egraph{E_{1}}}{x'}$
and $\Known{\Egraph{E_{1}}}{s'}$. 
\item $A_{12}=\Member(x',\Add(x',s'))$. 
\item $\filterLit{A_{12}}=\{\Member(x',\Add(x',s'))\sim\top\}$. 
\item $\Egraph{E_{2}}=\Egraph{E_{1}}\triangleleft\filterLit{A_{12}}=\Egraph{E_{1}}\cup\{\Member(x',\Add(x',s'))\sim\top\}$. 
\item Since $B_{E_{1}}$ is a basis for $\Egraph{E_{1}}$, there are $b_{1},b_{2},b_{3}\in B_{E_{1}}$
such that $\Entail{\Egraph{E_{1}}}{\Equiv{s'}{b_{1}}}$, $\Entail{\Egraph{E_{1}}}{\Equiv{x'}{b_{2}}}$
and $\Entail{\Egraph{E_{1}}}{\Equiv{\Add(x',s')}{b_{3}}}$. 
\item There are $b_{1},b_{2},b_{3}\in B_{E_{1}}\subseteq\FilterSet(B_{E_{1}})\cup\FilterT(B_{E_{1}})$
such that $\Entail{\Egraph{E_{2}}}{\Equiv{s'}{b_{1}}}$, $\Entail{\Egraph{E_{2}}}{\Equiv{x'}{b_{2}}}$
and $\Entail{\Egraph{E_{2}}}{\Equiv{\Add(x',s')}{b_{3}}}$. 
\item Our goal is to show that $\FilterSet(B_{E_{1}})\cup\FilterT(B_{E_{1}})$
is a candidate basis for $\Egraph{E_{2}}$. That is, for every $t$,
if $\Known{\Egraph{E_{2}}}t$, then there exists some $b\in\FilterSet(B_{E_{1}})\cup\FilterT(B_{E_{1}})$
such that $\Entail{\Egraph{E_{2}}}{\Equiv tb}$. \\
Assume $\Known{\Egraph{E_{2}}}t$, that is $\Known{\Egraph{E_{1}}\cup\{\Member(x',\Add(x',s'))\sim\top\}}t$.
It must be the case that $\Known{\Egraph{E_{1}}}t$, which has been
handled. 
\end{enumerate}
\end{proof}

\begin{proposition}
[Validity of Invariants] Suppose the initial state is $s_{0}=\state{W_{0}}{\emptyset}{E_{0}}$,
where $W_{0}$ is our axiomatisation for set theory with each quantifier
tagged as specified by Def. \ref{def:appx-tagging-quantifiers-set-theory-axioms},
$E_{0}=\InjFunc L$, and $L$ is an arbitrary set of ground literals
from set theory. Let $B_{E_{0}}$ be a basis for $E_{0}$. \label{prop:appx-set-validity-of-invariants}

If $s_{0}\tran^{*}s$, then $\Invariant{}{s,s_{0}}=\Invariant G{s,s_{0}}\wedge\Invariant P{s,s_{0}}$
holds. 
\end{proposition}

\begin{proof}
Since $\Invariant{}{s,s_{0}}$ is defined to be a conjunction of all
invariants involving states $s$ and $s_{0}$, proving the validity
of $\Invariant{}{s,s_{0}}$ boils down to establishing the validity
of each individual general-purpose and problem-specific invariants,
such as Proposition \ref{prop:appx-prop-P:BS-basis-after-a-step}
regarding the invariant $\Invariant{P:BS}{s,s_{0}}$. 
\end{proof}

\global\long\def\updateW#1#2{#1\cup#2}%

\begin{lemma}
[Descent of Measure] Suppose the initial state is $s_{0}=\state{W_{0}}{\emptyset}{E_{0}}$,
where $W_{0}$ is our axiomatisation for set theory with each quantifier
tagged as specified by Def. \ref{def:appx-tagging-quantifiers-set-theory-axioms},
$E_{0}=\InjFunc L$, and $L$ is an arbitrary set of ground literals
from set theory. Let $B_{E_{0}}$ be a basis for $E_{0}$. \label{lem:set-descent-of-measure}

Let $\Invariant{}{s,s_{0}}=\Invariant G{s,s_{0}}\wedge\Invariant P{s,s_{0}}$. 

Suppose $s_{1}\tran s_{2}$ and $I(s_{1},s_{0})$ holds. Then $M(s_{2})<M(s_{1})$
by the lexicographical order. 
\end{lemma}

\begin{proof}
By definition of $\tran$, proceed by cases on $s_{1}\tran s_{2}$.
There are four cases, $\textsc{(split)}$, $\textsc{(bot)}$, $\textsc{(sat)}$
and $\textsc{(inst)}$. Both $\textsc{(bot)}$ and $\textsc{(sat)}$
cases are straightforward by the definitions of $\Sigma$ and $\Theta$.
We instead focus on the other two cases. For the $\textsc{(split)}$
case, we demonstrate that $M(s_{2})<M(s_{1})$ because $\Sigma(s_{2})\leq\Sigma(s_{1})$
and $\Theta(s_{2})<\Theta(s_{1})$. For the $\textsc{(inst)}$ case,
we demonstrate that $M(s_{2})<M(s_{1})$ because $\Sigma(s_{2})<\Sigma(s_{1})$. 

We first work on the $\textsc{(split)}$ case. Let $s_{1}=\state{W_{1}}{A_{1}}{E_{1}}$.
By $\Invariant{P:IB}{s_{1},s_{0}}$, $\FilterSet(B_{E_{0}})\cup\FilterT(B_{E_{0}})$
is a candidate basis for $\Egraph{E_{1}}$. Let $B_{E_{1}}\subseteq\FilterSet(B_{E_{0}})\cup\FilterT(B_{E_{0}})$
be a basis for $\Egraph{E_{1}}$. 

Let $\Phi\subseteq\left\{ \phi_{i}\mid C\in A_{1},\;\NotVerify{W_{1},\Egraph{E_{1}}}C,\;C\text{ is }\phi_{1}\vee\dots\vee\phi_{i}\vee\dots\vee\phi_{n},\;n\geq2\right\} $.

Let $s_{2}=\state{W_{2}}{A_{2}}{E_{2}}$. We have: 
\begin{enumerate}
\item $A_{2}=A_{1}$. Let $A_{2}=A_{1}=A$. 
\item $\Egraph{E_{2}}=\updateEgraph{\Egraph{E_{1}}}{\filterLit{\Phi}}$.
Then, $\Egraph{E_{2}}\supseteq\Egraph{E_{1}}$. 
\item Since $B_{E_{1}}$ is a basis for $\Egraph{E_{1}}$, by $\Invariant{P:BS}{s_{1},s_{0}}$,
$\FilterT\left(B_{E_{1}}\right)\cup\FilterSet\left(B_{E_{1}}\right)$
is a candidate basis for $\Egraph{E_{2}}$. Let $B_{E_{2}}\subseteq\FilterSet\left(B_{E_{1}}\right)\cup\FilterT\left(B_{E_{1}}\right)$
be a basis for $\Egraph{E_{2}}$. 
\item $\Ehistory{E_{2}}=\ehistory{E_{1}}$.
\item By $\Invariant{G:VV}{s_{1},s_{0}}$, if $C\in A$ and $\Verify{W_{1},\Egraph{E_{1}}}C$,
then $\Verify{W_{2},\Egraph{E_{2}}}C$. Taking its contrapositive,
if $C\in A$ and $\NotVerify{W_{2},\Egraph{E_{2}}}C$, then $\NotVerify{W_{1},\Egraph{E_{1}}}C$. 
\end{enumerate}
We compute $\Sigma$ on both $s_{1}$ and $s_{2}$: 

\begin{align*}
\Sigma\left(s_{2}\right)= & \underset{p\in P\left(\state{W_{2}}{A_{2}}{E_{2}},B_{E_{2}}\right)}{\sum}\Amount p\\
= & \Amount{p_{\Tag\text{union-elim}}}+\Amount{p_{\Tag\text{subset-intro}}}+\Amount{p_{\Tag\text{subset-elim}}}+\\
 & \underset{a,b\in\FilterSet\left(B_{E_{2}}\right)}{\sum}\Amount{p_{\Tag\text{subset-elim}(a,b)}}+\cdots\\
= & \Amount{\left\{ \left(s_{1}',s_{2}',x'\right)\;\middle\vert\;\begin{array}{l}
s_{1}',s_{2}'\in\FilterT(B_{E_{2}}),x'\in\FilterT\left(B_{E_{2}}\right),\\
\Inst{E_{2}}{\left(\Tag\text{union-elim}:(s_{1}',s_{2}',x')\right)}
\end{array}\right\} }\\
 & +\Amount{\left\{ \left(s_{1}',s_{2}'\right)\;\middle\vert\;\begin{array}{l}
s_{1}',s_{2}'\in\FilterSet\left(B_{E_{2}}\right),\\
\Inst{E_{2}}{\left(\Tag\text{subset-intro}:\left(s_{1}',s_{2}'\right)\right)}
\end{array}\right\} }\\
 & +\Amount{\left\{ \left(s_{1}',s_{2}'\right)\;\middle\vert\;\begin{array}{l}
s_{1}',s_{2}'\in\FilterSet\left(B_{E_{2}}\right),\\
\Inst{E_{2}}{\left(\Tag\text{subset-elim}:\left(s_{1}',s_{2}'\right)\right)}
\end{array}\right\} }\\
 & +\underset{a,b\in\FilterSet\left(B_{E_{2}}\right)}{\sum}\Amount{\left\{ x'\in T\;\middle\vert\;\begin{array}{l}
x'\in\FilterT\left(B_{E_{2}}\right),\\
\Inst{E_{2}}{\left(\Tag\text{subset-elim(\ensuremath{a,b})}:x'\right)}
\end{array}\right\} }\\
 & +\cdots
\end{align*}
\begin{align*}
\Sigma\left(s_{1}\right)= & \underset{p\in P\left(\state{W_{1}}{A_{1}}{E_{1}},B_{E_{1}}\right)}{\sum}\Amount p\\
= & \Amount{p_{\Tag\text{union-elim}}}+\Amount{p_{\Tag\text{subset-intro}}}+\Amount{p_{\Tag\text{subset-elim}}}+\\
 & \underset{a,b\in\FilterSet\left(B_{E_{1}}\right)}{\sum}\Amount{p_{\Tag\text{subset-elim}(a,b)}}+\cdots\\
= & \Amount{\left\{ \left(s_{1}',s_{2}',x'\right)\;\middle\vert\;\begin{array}{l}
s_{1}',s_{2}'\in\FilterT(B_{E_{1}}),x'\in\FilterT\left(B_{E_{1}}\right),\\
\Inst{E_{1}}{\left(\Tag\text{union-elim}:(s_{1}',s_{2}',x')\right)}
\end{array}\right\} }\\
 & +\Amount{\left\{ \left(s_{1}',s_{2}'\right)\;\middle\vert\;\begin{array}{l}
s_{1}',s_{2}'\in\FilterSet\left(B_{E_{1}}\right),\\
\Inst{E_{1}}{\left(\Tag\text{subset-intro}:\left(s_{1}',s_{2}'\right)\right)}
\end{array}\right\} }\\
 & +\Amount{\left\{ \left(s_{1}',s_{2}'\right)\;\middle\vert\;\begin{array}{l}
s_{1}',s_{2}'\in\FilterSet\left(B_{E_{1}}\right),\\
\Inst{E_{1}}{\left(\Tag\text{subset-elim}:\left(s_{1}',s_{2}'\right)\right)}
\end{array}\right\} }\\
 & +\underset{a,b\in\FilterSet\left(B_{E_{1}}\right)}{\sum}\Amount{\left\{ x'\in T\;\middle\vert\;\begin{array}{l}
x'\in\FilterT\left(B_{E_{1}}\right),\\
\Inst{E_{1}}{\left(\Tag\text{subset-elim(\ensuremath{a,b})}:x'\right)}
\end{array}\right\} }\\
 & +\cdots
\end{align*}

To show $\Sigma\left(s_{2}\right)\leq\Sigma\left(s_{1}\right)$, it
suffices to show the following propositions. 
\begin{enumerate}
\item If $s_{1}',s_{2}'\in\FilterT(B_{E_{2}})$, $x'\in\FilterT\left(B_{E_{2}}\right)$,
$\Inst{E_{2}}{\left(\Tag\text{union-elim}:(s_{1}',s_{2}',x')\right)}$,
then $x'\in\FilterT\left(B_{E_{1}}\right)$, $s_{1}',s_{2}'\in\FilterT(B_{E_{1}})$,
$\Inst{E_{1}}{\left(\Tag\text{union-elim}:(s_{1}',s_{2}',x')\right)}$. 
\item If $s_{1}',s_{2}'\in\FilterSet\left(B_{E_{2}}\right)$ and $\Inst{E_{2}}{\left(\Tag\text{subset-intro}:\left(s_{1}',s_{2}'\right)\right)}$,
then $s_{1}',s_{2}'\in\FilterSet\left(B_{E_{1}}\right)$ and $\Inst{E_{1}}{\left(\Tag\text{subset-intro}:\left(s_{1}',s_{2}'\right)\right)}$. 
\item If $s_{1}',s_{2}'\in\FilterSet\left(B_{E_{2}}\right)$ and $\Inst{E_{2}}{\left(\Tag\text{subset-elim}:\left(s_{1}',s_{2}'\right)\right)}$,
then $s_{1}',s_{2}'\in\FilterSet\left(B_{E_{1}}\right)$ and $\Inst{E_{1}}{\left(\Tag\text{subset-elim}:\left(s_{1}',s_{2}'\right)\right)}$. 
\item (1) If $a,b\in\FilterSet\left(B_{E_{2}}\right)$, then $a,b\in\FilterSet\left(B_{E_{1}}\right)$.
 (2) If $x'\in\FilterT\left(B_{E_{2}}\right)$ and $\Inst{E_{2}}{\left(\Tag\text{subset-elim(\ensuremath{a,b})}:x'\right)}$,
then $x'\in\FilterT\left(B_{E_{1}}\right)$ and $\Inst{E_{1}}{\left(\Tag\text{subset-elim(\ensuremath{a,b})}:x'\right)}$. 
\item (other cases omitted). 
\end{enumerate}
We demonstrate that the second proposition holds; the rest cases can
be proved analogously.
\begin{enumerate}
\item Since $B_{E_{2}}\subseteq\FilterSet\left(B_{E_{1}}\right)\cup\FilterT\left(B_{E_{1}}\right)$,
if $s_{1}',s_{2}'\in\FilterSet\left(B_{E_{2}}\right)$, then $s_{1}',s_{2}'\in\FilterSet\left(B_{E_{1}}\right)$. 
\item If $\Inst{E_{2}}{\left(\Tag\text{subset-intro}:\left(s_{1}',s_{2}'\right)\right)}$,
then for every $\left(\Tag\text{subset-intro}:\left(r_{1},r_{2}\right)\right)\in\Ehistory{E_{2}}$,
$\NotEntail{\Egraph{E_{2}}}{\Equiv{\left(s_{1}',s_{2}'\right)}{\left(r_{1},r_{2}\right)}}$. 
\item Since $\Ehistory{E_{2}}=\ehistory{E_{1}}$ and $\Egraph{E_{2}}\supseteq\Egraph{E_{1}}$,
for every $\left(\Tag\text{subset-intro}:\left(r_{1},r_{2}\right)\right)\in\Ehistory{E_{1}}$,
$\NotEntail{\Egraph{E_{1}}}{\Equiv{\left(s_{1}',s_{2}'\right)}{\left(r_{1},r_{2}\right)}}$.
That is, $\Inst{E_{1}}{\left(\Tag\text{subset-intro}:\left(s_{1}',s_{2}'\right)\right)}$. 
\end{enumerate}
Therefore, $\Sigma\left(s_{2}\right)\leq\Sigma\left(s_{1}\right)$. 

We compute $\Theta$ on both $s_{1}$ and $s_{2}$: 
\begin{align*}
\Theta\left(s_{2}\right)= & \Amount{\left\{ C\in A\;|\;\NotVerify{W_{2},\Egraph{E_{2}}}C\right\} }
\end{align*}
\begin{align*}
\Theta\left(s_{1}\right)= & \Amount{\left\{ C\in A\;|\;\NotVerify{W_{1},\Egraph{E_{1}}}C\right\} }
\end{align*}

To show $\Theta\left(s_{2}\right)<\Theta\left(s_{1}\right)$, it suffices
to show the following propositions. 
\begin{enumerate}
\item For every $C\in A$, if $\NotVerify{W_{2},\Egraph{E_{2}}}C$, then
$\NotVerify{W_{1},\Egraph{E_{1}}}C$. 
\item There exists some $C\in A$ such that $\NotVerify{W_{1},\Egraph{E_{1}}}C$
and $\Verify{W_{2},\Egraph{E_{2}}}C$. 
\end{enumerate}
The first proposition has been established. To prove the second proposition,
assume $\NotVerify{W_{1},\Egraph{E_{1}}}C$, $\phi_{i}\in C$, and
$\phi_{i}\subseteq\Phi$. If $\phi_{i}$ is a tagged quantifier, then
$\phi_{i}\in W_{2}$; if $\phi_{i}$ is $t_{1}=t_{2}$, then $\Entail{\Egraph{E_{2}}}{\Equiv{t_{1}}{t_{2}}}$;
if $\phi_{i}$ is $t_{1}\neq t_{2}$, then $\Entail{\Egraph{E_{1}}}{\NotEquiv{t_{1}}{t_{2}}}$.
In all these three cases, $\Verify{W_{2},\Egraph{E_{2}}}C$. 

We then work on the $\textsc{(inst)}$ case. Let $s_{1}=\state{W_{1}}{A_{1}}{E_{1}}$.
By $\Invariant{P:IB}{s_{1},s_{0}}$, $\FilterSet(B_{E_{0}})\cup\FilterT(B_{E_{0}})$
is a candidate basis for $\Egraph{E_{1}}$. Let $B_{E_{1}}\subseteq\FilterSet(B_{E_{0}})\cup\FilterT(B_{E_{0}})$
be a basis for $\Egraph{E_{1}}$. 

Let $s_{2}=\state{W_{2}}{A_{2}}{E_{2}}$. By $\Invariant{P:BS}{s_{1},s_{0}}$,
$\FilterSet(B_{E_{1}})\cup\FilterT(B_{E_{1}})$ is a candidate basis
for $\Egraph{E_{2}}$. Let $B_{E_{2}}\subseteq\FilterSet(B_{E_{1}})\cup\FilterT(B_{E_{1}})$
be a basis for $\Egraph{E_{2}}$. 

Proceed by cases on $s_{1}\Ematching\cdot\MatchSep\cdot$. We choose
the following case as an example; the other cases are analogous to
this case. 
\[
\begin{array}{r}
\state{W_{1}}{A_{1}}{E_{1}}\Ematching\left(\begin{array}{l}
\forall s_{1},s_{2}.\;\left[\Subset(s_{1},s_{2})\right]\\
\left(\Subset(s_{1},s_{2})\vee\Member(\Skolem_{ss}(s_{1},s_{2}),s_{1})\right)\wedge\\
\left(\Subset(s_{1},s_{2})\vee\neg\Member(\Skolem_{ss}(s_{1},s_{2}),s_{2})\right)
\end{array}\right)^{\Tag\text{subset-intro}}\\
\MatchSep\left(s_{1}^{\prime},s_{2}^{\prime}\right)
\end{array}
\]
We have: 
\begin{enumerate}
\item $\Known{\Egraph{E_{1}}}{\Subset(s_{1}^{\prime},s_{2}^{\prime})}$
and $\Known{\Egraph{E_{1}}}{\left(s_{1}^{\prime},s_{2}^{\prime}\right)}$.
\item $\Inst{E_{1}}{\left(\Tag\text{subset-intro}:\left(s_{1}',s_{2}'\right)\right)}$
\item $\Egraph{E_{1}}=\Egraph{E_{2}}$. 
\item $\Ehistory{E_{2}}=\Ehistory{E_{1}}\cup\left\{ \left(\Tag\text{subset-intro}:\left(s_{1}^{\prime},s_{2}^{\prime}\right)\right)\right\} $. 
\end{enumerate}
We compute $\Sigma$ on both $s_{1}$ and $s_{2}$: 
\begin{align*}
\Sigma\left(s_{2}\right)= & \underset{p\in P\left(\state{W_{2}}{A_{2}}{E_{2}},B_{E_{2}}\right)}{\sum}\Amount p\\
= & \Amount{p_{\Tag\text{union-elim}}}+\Amount{p_{\Tag\text{subset-intro}}}+\Amount{p_{\Tag\text{subset-elim}}}+\\
 & \underset{a,b\in\FilterSet\left(B_{E_{2}}\right)}{\sum}\Amount{p_{\Tag\text{subset-elim}(a,b)}}+\cdots\\
= & \Amount{\left\{ \left(s_{1}',s_{2}',x'\right)\;\middle\vert\;\begin{array}{l}
s_{1}',s_{2}'\in\FilterT(B_{E_{2}}),x'\in\FilterT\left(B_{E_{2}}\right),\\
\Inst{E_{2}}{\left(\Tag\text{union-elim}:(s_{1}',s_{2}',x')\right)}
\end{array}\right\} }\\
 & +\Amount{\left\{ \left(s_{1}',s_{2}'\right)\;\middle\vert\;\begin{array}{l}
s_{1}',s_{2}'\in\FilterSet\left(B_{E_{2}}\right),\\
\Inst{E_{2}}{\left(\Tag\text{subset-intro}:\left(s_{1}',s_{2}'\right)\right)}
\end{array}\right\} }\\
 & +\Amount{\left\{ \left(s_{1}',s_{2}'\right)\;\middle\vert\;\begin{array}{l}
s_{1}',s_{2}'\in\FilterSet\left(B_{E_{2}}\right),\\
\Inst{E_{2}}{\left(\Tag\text{subset-elim}:\left(s_{1}',s_{2}'\right)\right)}
\end{array}\right\} }\\
 & +\underset{a,b\in\FilterSet\left(B_{E_{2}}\right)}{\sum}\Amount{\left\{ x'\in T\;\middle\vert\;\begin{array}{l}
x'\in\FilterT\left(B_{E_{2}}\right),\\
\Inst{E_{2}}{\left(\Tag\text{subset-elim(\ensuremath{a,b})}:x'\right)}
\end{array}\right\} }\\
 & +\cdots
\end{align*}
\begin{align*}
\Sigma\left(s_{1}\right)= & \underset{p\in P\left(\state{W_{1}}{A_{1}}{E_{1}},B_{E_{1}}\right)}{\sum}\Amount p\\
= & \Amount{p_{\Tag\text{union-elim}}}+\Amount{p_{\Tag\text{subset-intro}}}+\Amount{p_{\Tag\text{subset-elim}}}+\\
 & \underset{a,b\in\FilterSet\left(B_{E_{1}}\right)}{\sum}\Amount{p_{\Tag\text{subset-elim}(a,b)}}+\cdots\\
= & \Amount{\left\{ \left(s_{1}',s_{2}',x'\right)\;\middle\vert\;\begin{array}{l}
s_{1}',s_{2}'\in\FilterT(B_{E_{1}}),x'\in\FilterT\left(B_{E_{1}}\right),\\
\Inst{E_{1}}{\left(\Tag\text{union-elim}:(s_{1}',s_{2}',x')\right)}
\end{array}\right\} }\\
 & +\Amount{\left\{ \left(s_{1}',s_{2}'\right)\;\middle\vert\;\begin{array}{l}
s_{1}',s_{2}'\in\FilterSet\left(B_{E_{1}}\right),\\
\Inst{E_{1}}{\left(\Tag\text{subset-intro}:\left(s_{1}',s_{2}'\right)\right)}
\end{array}\right\} }\\
 & +\Amount{\left\{ \left(s_{1}',s_{2}'\right)\;\middle\vert\;\begin{array}{l}
s_{1}',s_{2}'\in\FilterSet\left(B_{E_{1}}\right),\\
\Inst{E_{1}}{\left(\Tag\text{subset-elim}:\left(s_{1}',s_{2}'\right)\right)}
\end{array}\right\} }\\
 & +\underset{a,b\in\FilterSet\left(B_{E_{1}}\right)}{\sum}\Amount{\left\{ x'\in T\;\middle\vert\;\begin{array}{l}
x'\in\FilterT\left(B_{E_{1}}\right),\\
\Inst{E_{1}}{\left(\Tag\text{subset-elim(\ensuremath{a,b})}:x'\right)}
\end{array}\right\} }\\
 & +\cdots
\end{align*}

To show $\Sigma\left(s_{2}\right)<\Sigma\left(s_{1}\right)$, it suffices
to show the following propositions. 
\begin{enumerate}
\item (union-elim). 
\begin{enumerate}
\item There exists $s_{1}',s_{2}'\in\FilterSet\left(B_{E_{1}}\right)$ and
$x'\in\FilterT\left(B_{E_{1}}\right)$ such that \\
$\Inst{E_{1}}{\left(\Tag\text{union-elim}:\left(s_{1}',s_{2}',x'\right)\right)}$,
but either $s_{1}'\notin\FilterSet\left(B_{E_{2}}\right)$, or $s_{2}'\notin\FilterSet\left(B_{E_{2}}\right),$
or $x'\in\FilterT\left(B_{E_{2}}\right)$, or $\NotInst{E_{2}}{\left(\Tag\text{union-elim}:(s_{1}',s_{2}',x')\right)}$. 
\item If $s_{1}',s_{2}'\in\FilterT(B_{E_{2}})$, $x'\in\FilterT\left(B_{E_{2}}\right)$,
$\Inst{E_{2}}{\left(\Tag\text{union-elim}:(s_{1}',s_{2}',x')\right)}$,
then $x'\in\FilterT\left(B_{E_{1}}\right)$, $s_{1}',s_{2}'\in\FilterT(B_{E_{1}})$,
$\Inst{E_{1}}{\left(\Tag\text{union-elim}:(s_{1}',s_{2}',x')\right)}$. 
\end{enumerate}
\item (subset-intro). 
\begin{enumerate}
\item There exists $s_{1}',s_{2}'\in\FilterSet\left(B_{E_{1}}\right)$ such
that \\
$\Inst{E_{1}}{\left(\Tag\text{subset-intro}:\left(s_{1}',s_{2}'\right)\right)}$,
but either $s_{1}'\notin\FilterSet\left(B_{E_{2}}\right)$, or $s_{2}'\notin\FilterSet\left(B_{E_{2}}\right),$
or $\NotInst{E_{2}}{\left(\Tag\text{subset-intro}:\left(s_{1}',s_{2}'\right)\right)}$. 
\item If $s_{1}',s_{2}'\in\FilterSet\left(B_{E_{2}}\right)$ and $\Inst{E_{2}}{\left(\Tag\text{subset-intro}:\left(s_{1}',s_{2}'\right)\right)}$,
then $s_{1}',s_{2}'\in\FilterSet\left(B_{E_{1}}\right)$ and $\Inst{E_{1}}{\left(\Tag\text{subset-intro}:\left(s_{1}',s_{2}'\right)\right)}$. 
\end{enumerate}
\item (subset-elim). 
\begin{enumerate}
\item There exists $s_{1}',s_{2}'\in\FilterSet\left(B_{E_{1}}\right)$ such
that \\
$\Inst{E_{1}}{\left(\Tag\text{subset-elim}:\left(s_{1}',s_{2}'\right)\right)}$,
but either $s_{1}'\notin\FilterSet\left(B_{E_{2}}\right)$, or $s_{2}'\notin\FilterSet\left(B_{E_{2}}\right),$
or $\NotInst{E_{2}}{\left(\Tag\text{subset-elim}:\left(s_{1}',s_{2}'\right)\right)}$. 
\item If $s_{1}',s_{2}'\in\FilterSet\left(B_{E_{2}}\right)$ and $\Inst{E_{2}}{\left(\Tag\text{subset-elim}:\left(s_{1}',s_{2}'\right)\right)}$,
then $s_{1}',s_{2}'\in\FilterSet\left(B_{E_{1}}\right)$ and $\Inst{E_{1}}{\left(\Tag\text{subset-elim}:\left(s_{1}',s_{2}'\right)\right)}$. 
\end{enumerate}
\item (subset-elim$(a,b)$). 
\begin{enumerate}
\item There exists some $x'\in\FilterT\left(B_{E_{1}}\right)$ such that
$\Inst{E_{1}}{\left(\Tag\text{subset-elim(\ensuremath{a,b})}:x'\right)}$,
but either $x'\notin\FilterT\left(B_{E_{2}}\right)$ or $\NotInst{E_{2}}{\left(\Tag\text{subset-elim(\ensuremath{a,b})}:x'\right)}$. 
\item If $a,b\in\FilterSet\left(B_{E_{2}}\right)$, then $a,b\in\FilterSet\left(B_{E_{1}}\right)$.
\item If $x'\in\FilterT\left(B_{E_{2}}\right)$ and $\Inst{E_{2}}{\left(\Tag\text{subset-elim(\ensuremath{a,b})}:x'\right)}$,
then $x'\in\FilterT\left(B_{E_{1}}\right)$ and $\Inst{E_{1}}{\left(\Tag\text{subset-elim(\ensuremath{a,b})}:x'\right)}$. 
\end{enumerate}
\item (other cases omitted)
\end{enumerate}
We focus on the (subset-intro) case here; the remaining cases can
be proved analogously.

We first work on the first proposition of the (subset-intro) case.
\begin{enumerate}
\item Given $\Known{\Egraph{E_{1}}}{\left(s_{1}^{\prime},s_{2}^{\prime}\right)}$,
by $\Invariant{G:KB}{s_{1},s_{0}}$, there exists $b_{1},b_{2}\in B_{E_{1}}$
such that $\Entail{\Egraph{E_{1}}}{\Equiv{\left(s_{1}^{\prime},s_{2}^{\prime}\right)}{\left(b_{1},b_{2}\right)}}$. 
\item Proceed on whether $b_{1},b_{2}\in B_{E_{2}}$. If either $b_{1}\notin B_{E_{2}}$
or $b_{2}\notin B_{E_{2}}$, which implies either $b_{1}\notin\FilterSet\left(B_{E_{2}}\right)$
or $b_{2}\notin\FilterSet\left(B_{E_{2}}\right)$---the proposition
holds. Assume $b_{1},b_{2}\in B_{E_{2}}$. 
\item Given $\Inst{E_{1}}{\left(\Tag\text{subset-intro}:\left(s_{1}^{\prime},s_{2}^{\prime}\right)\right)}$,
by $\Invariant{G:HE}{s_{1},s_{0}}$, $\Inst{E_{1}}{\left(\Tag\text{subset-intro}:\left(b_{1},b_{2}\right)\right)}$.
\item Our goal is to show $\NotInst{E_{2}}{\left(\Tag\text{subset-intro}:\left(b_{1},b_{2}\right)\right)}$. 
\begin{enumerate}
\item Given $\Egraph{E_{1}}\subseteq\Egraph{E_{2}}$ and $\Entail{\Egraph{E_{1}}}{\Equiv{\left(s_{1}^{\prime},s_{2}^{\prime}\right)}{\left(b_{1},b_{2}\right)}}$,
$\Entail{\Egraph{E_{2}}}{\Equiv{\left(s_{1}^{\prime},s_{2}^{\prime}\right)}{\left(b_{1},b_{2}\right)}}$. 
\item Given $\left(\Tag\text{subset-intro}:\left(s_{1}^{\prime},s_{2}^{\prime}\right)\right)\in\Ehistory{E_{2}}$
and $\Entail{\Egraph{E_{2}}}{\Equiv{\left(s_{1}^{\prime},s_{2}^{\prime}\right)}{\left(b_{1},b_{2}\right)}}$,
$\NotInst{E_{2}}{\left(\Tag\text{subset-intro}:\left(b_{1},b_{2}\right)\right)}$. 
\end{enumerate}
\end{enumerate}
We then work on the second proposition of the (subset-intro) case. 
\begin{enumerate}
\item Since $B_{E_{2}}\subseteq\FilterSet\left(B_{E_{1}}\right)\cup\FilterT\left(B_{E_{1}}\right)$,
if $s_{1}',s_{2}'\in\FilterSet\left(B_{E_{2}}\right)$, then $s_{1}',s_{2}'\in\FilterSet\left(B_{E_{1}}\right)$. 
\item If $\Inst{E_{2}}{\left(\Tag\text{subset-intro}:\left(s_{1}',s_{2}'\right)\right)}$,
then there exists no $\left(\Tag\text{subset-intro}:\multi r\right)\in\Ehistory{E_{2}}$
such that $\Entail{\Egraph{E_{2}}}{\Equiv{\left(s_{1}',s_{2}'\right)}{\multi r}}$.
\item Since $\Ehistory{E_{2}}=\Ehistory{E_{1}}\cup\left\{ \left(\Tag\text{subset-intro}:\left(s_{1}^{\prime},s_{2}^{\prime}\right)\right)\right\} $,
there exists no $\left(\Tag\text{subset-intro}:\multi r\right)\in\Ehistory{E_{1}}$
such that $\Entail{\Egraph{E_{2}}}{\Equiv{\left(s_{1}',s_{2}'\right)}{\multi r}}$. 
\item Since $\Egraph{E_{1}}\subseteq\Egraph{E_{2}}$, then there exists
no $\left(\Tag\text{subset-intro}:\multi r\right)\in\Ehistory{E_{1}}$
such that $\Entail{\Egraph{E_{1}}}{\Equiv{\left(s_{1}',s_{2}'\right)}{\multi r}}$. 
\item Hence $\Inst{E_{1}}{\left(\Tag\text{subset-intro}:\left(s_{1}',s_{2}'\right)\right)}$. 
\end{enumerate}
Therefore, $\Sigma\left(s_{2}\right)<\Sigma\left(s_{1}\right)$. 

Overall, we have proved that $M\left(s_{2}\right)<M\left(s_{1}\right)$. 
\end{proof}

\begin{theorem}
[Instantiation Termination for Set Theory] Suppose the initial state
is $s_{0}=\state{W_{0}}{\emptyset}{E_{0}}$, where $W_{0}$ is our
axiomatisation for set theory with each quantifier tagged as specified
by Def. \ref{def:appx-tagging-quantifiers-set-theory-axioms}, $E_{0}=\InjFunc L$,
and $L$ is an arbitrary set of ground literals from set theory.

Any sequence of transitions from the initial sate $s_{0}$, where
$\tran$ represents the state transition relation, has a finite length. 
\end{theorem}

\begin{proof}
Suppose there exists an infinite path, $s_{0}\tran s\tran s^{\prime}\tran s^{\prime\prime}\tran\dots$. 
\begin{lyxlist}{00.00.0000}
\item [{Conjecture}] If $s_{0}\tran^{*}s_{1}\tran s_{2}$, then $M(s_{2})<M(s_{1})$. 
\end{lyxlist}
By the above conjecture, we have $M(s_{0})>M(s)>M(s^{\prime})>M(s^{\prime\prime})>\dots$.
Given that the result of $M$ is a lexicographical order on $(\mathbb{N}\cup\{-1\})^{2}$,
the path $s_{0}\tran s\tran s^{\prime}\tran s^{\prime\prime}\longrightarrow\dots$
must be finite. 

We now prove the conjecture. Let $\Invariant{}{s_{1},s_{0}}=\Invariant G{s_{1},s_{0}}\wedge\Invariant P{s_{1},s_{0}}$.
It suffices to prove the following propositions. 
\begin{enumerate}
\item If $s_{0}\tran^{*}s_{1}$, then $\Invariant{}{s_{1},s_{0}}$ holds.
\item If $s_{1}\tran s_{2}$ and $\Invariant{}{s_{1},s_{0}}$ holds, then
$M(s_{2})<M(s_{1})$. 
\end{enumerate}
The first proposition is implied by Proposition \ref{prop:appx-set-validity-of-invariants}.
The second proposition is implied by Lemma \ref{lem:set-descent-of-measure}. 
\end{proof}

\section{Set Theory Axiomatisations }

We present our axiomatisation for set theory in Appx. \ref{subsec:appx-our-axiomatisation-set-theory},
and compare our axiomatisation with the counterparts from Dafny and
Viper in Appx. \ref{subsec:appx-comparison-axiomatisations-set-theory}. 

\subsection{Our Axiomatisation for Set Theory \label{subsec:appx-our-axiomatisation-set-theory}}

We present our axiomatisation by groups of axioms, with each group
focusing on one operation. For each axiom, we provide two versions:
one without triggers for a straightforward understanding, and another
with triggers, formatted in ECNF for consistency with our formal model. 

\subsubsection{Empty \\}

empty
\[
\forall x.\;\neg\Member(x,\Empty())
\]
\[
\forall x.\left[\Member(x,\Empty())\right]\;\neg\Member(x,\Empty())
\]

\subsubsection{Singleton \\}

\noindent singleton-intro-1
\[
\forall x.\Member(x,\Singleton(x))
\]
\[
\forall x.\left[\Singleton(x)\right]\Member(x,\Singleton(x))
\]

\noindent singleton-intro-2
\[
\forall x,y.\Member(y,\Singleton(x))\leftarrow x=y
\]
\[
\forall x,y.\left[\Member(y,\Singleton(x))\right]\Member(y,\Singleton(x))\vee x\neq y
\]

\noindent singleton-elim
\[
\forall x,y.\Member(y,\Singleton(x))\rightarrow x=y
\]
\[
\forall x,y.\left[\Member(y,\Singleton(x))\right]\neg\Member(y,\Singleton(x))\vee x=y
\]

\subsubsection{Add \\}

add-intro-1
\[
\forall s,x,y.\;\Member(y,\Add(x,s))\leftarrow\Member(y,s)
\]
\[
\begin{array}{r}
\forall s,x,y.\left[\Member(y,s),\Add(x,s)\right]\left[\Member(y,\Add(x,s))\right]\\
\;\Member(y,\Add(x,s))\vee\neg\Member(y,s)
\end{array}
\]

\noindent add-intro-2
\[
\forall s,x.\;\Member(x,\Add(x,s))
\]
\[
\forall s,x.\left[\Add(x,s)\right]\;\Member(x,\Add(x,s))
\]

\noindent add-intro-3 
\[
\forall s,x,y.\;\Member(y,\Add(x,s))\leftarrow y=x
\]
\[
\begin{array}{r}
\forall s,x,y.\left[\Member(y,\Add(x,s))\right]\left[\Member(y,s),\Add(x,s)\right]\\
\Member(y,\Add(x,s))\vee y\neq x
\end{array}
\]

\noindent add-elim
\[
\forall s,x,y.\;\Member(y,\Add(x,s))\rightarrow(x=y)\vee\Member(y,s)
\]
\[
\begin{array}{r}
\forall s,x,y.\left[\Member(y,\Add(x,s))\right]\left[\Member(y,s),\Add(x,s)\right]\\
\neg\Member(y,\Add(x,s))\vee(x=y)\vee\Member(y,s)
\end{array}
\]

\subsubsection{Union \\}

\noindent union-intro-1
\[
\forall s_{1},s_{2},x.\Member\left(x,\Union\left(s_{1},s_{2}\right)\right)\leftarrow\Member\left(x,s_{1}\right)
\]
\[
\begin{array}{r}
\forall s_{1},s_{2},x.\left[\Union(s_{1},s_{2}),\Member(x,s_{1})\right]\left[\Member\left(x,\Union\left(s_{1},s_{2}\right)\right)\right]\\
\Member\left(x,\Union\left(s_{1},s_{2}\right)\right)\vee\neg\Member\left(x,s_{1}\right)
\end{array}
\]

\noindent union-intro-2
\[
\forall s_{1},s_{2},x.\Member\left(x,\Union\left(s_{1},s_{2}\right)\right)\leftarrow\Member\left(x,s_{2}\right)
\]
\[
\begin{array}{r}
\forall s_{1},s_{2},x.\left[\Union(s_{1},s_{2}),\Member(x,s_{2})\right]\left[\Member\left(x,\Union\left(s_{1},s_{2}\right)\right)\right]\\
\Member\left(x,\Union\left(s_{1},s_{2}\right)\right)\vee\neg\Member\left(x,s_{2}\right)
\end{array}
\]

\noindent union-elim
\[
\forall s_{1},s_{2},x.\Member\left(x,\Union\left(s_{1},s_{2}\right)\right)\rightarrow\Member(x,s_{1})\vee\Member(x,s_{2})
\]
\[
\begin{array}{l}
\forall s_{1},s_{2},x.\\
\left[\Member(x,\Union(s_{1},s_{2}))\right]\\
\left[\Union(s_{1},s_{2}),\Member(x,s_{1})\right]\left[\Union(s_{1},s_{2}),\Member(x,s_{2})\right]\\
\;\neg\Member\left(x,\Union\left(s_{1},s_{2}\right)\right)\vee\Member(x,s_{1})\vee\Member(x,s_{2})
\end{array}
\]

\noindent union-disjoint
\[
\begin{array}{l}
\forall s_{1},s_{2}.\left[\Union(s_{1},s_{2})\right]\\
\Disjoint(s_{1},s_{2})\rightarrow\left(\Diff(\Union(s_{1},s_{2}),s_{1})=s_{2}\right)\wedge\left(\Diff(\Union(s_{1},s_{2}),s_{2})=s_{1}\right)
\end{array}
\]

\subsubsection{Intersection \\}

\noindent inter-intro
\[
\forall s_{1},s_{2},x.\Member(x,\Inter(s_{1},s_{2}))\leftarrow\Member(x,s_{1})\wedge\Member(x,s_{2})
\]
\[
\begin{array}{l}
\forall s_{1},s_{2},x.\\
\left[\Member(x,s_{1}),\Inter(s_{1},s_{2})\right]\left[\Member(x,s_{2}),\Inter(s_{1},s_{2})\right]\\
\left[\Member(x,\Inter(s_{1},s_{2}))\right]\\
\;\Member(x,\Inter(s_{1},s_{2}))\vee\neg\Member(x,s_{1})\vee\neg\Member(x,s_{2})
\end{array}
\]

\noindent inter-elim
\[
\forall s_{1},s_{2},x.\Member(x,\Inter(s_{1},s_{2}))\rightarrow\Member(x,s_{1})\wedge\Member(x,s_{2})
\]
\[
\begin{array}{l}
\forall s_{1},s_{2},x.\left[\Member(x,\Inter(s_{1},s_{2}))\right]\\
\left[\Inter(s_{1},s_{2}),\Member(x,s_{1})\right]\left[\Inter(s_{1},s_{2}),\Member(x,s_{2})\right]\\
\left(\neg\Member(x,\Inter(s_{1},s_{2}))\vee\Member(x,s_{1})\right)\wedge\\
\left(\neg\Member(x,\Inter(s_{1},s_{2}))\vee\Member(x,s_{2})\right)
\end{array}
\]

\subsubsection{Properties on Union and Intersection \\}

\noindent union-right
\[
\forall s_{1},s_{2}.\Union(\Union(s_{1},s_{2}),s_{2})=\Union(s_{1},s_{2})
\]
\[
\forall s_{1},s_{2}.\left[\Union(\Union(s_{1},s_{2}),s_{2})\right]\Union(\Union(s_{1},s_{2}),s_{2})=\Union(s_{1},s_{2})
\]

\noindent union-left
\[
\forall s_{1},s_{2}.\Union(s_{1},\Union(s_{1},s_{2}))=\Union(s_{1},s_{2})
\]
\[
\forall s_{1},s_{2}.\left[\Union(s_{1},\Union(s_{1},s_{2}))\right]\Union(s_{1},\Union(s_{1},s_{2}))=\Union(s_{1},s_{2})
\]

\noindent inter-right
\[
\forall s_{1},s_{2}.\Inter(\Inter(s_{1},s_{2}),s_{2})=\Inter(s_{1},s_{2})
\]
\[
\forall s_{1},s_{2}.\left[\Inter(\Inter(s_{1},s_{2}),s_{2})\right]\Inter(\Inter(s_{1},s_{2}),s_{2})=\Inter(s_{1},s_{2})
\]

\noindent inter-left
\[
\forall s_{1},s_{2}.\Inter(s_{1},\Inter(s_{1},s_{2}))=\Inter(s_{1},s_{2})
\]
\[
\forall s_{1},s_{2}.\left[\Inter(s_{1},\Inter(s_{1},s_{2}))\right]\Inter(s_{1},\Inter(s_{1},s_{2}))=\Inter(s_{1},s_{2})
\]

\subsubsection{Difference \\}

\noindent diff-intro
\[
\forall s_{1},s_{2},x.\Member(x,\Diff(s_{1},s_{2}))\leftarrow\Member(x,s_{1})\wedge\neg\Member(x,s_{2})
\]
\[
\begin{array}{l}
\forall s_{1},s_{2},x.\left[\Member(x,s_{1}),\Diff(s_{1},s_{2})\right]\\
\left[\Member(x,\Diff(s_{1},s_{2}))\right]\left[\Member(x,s_{2}),\Diff(s_{1},s_{2})\right]\\
\Member(x,\Diff(s_{1},s_{2}))\vee\neg\Member(x,s_{1})\vee\Member(x,s_{2})
\end{array}
\]

\noindent diff-elim
\[
\forall s_{1},s_{2},x.\Member(x,\Diff(s_{1},s_{2}))\rightarrow\Member(x,s_{1})\wedge\neg\Member(x,s_{2})
\]
\[
\begin{array}{r}
\forall s_{1},s_{2},x.\left[\Member(x,\Diff(s_{1},s_{2}))\right]\left[\Member(x,s_{2}),\Diff(s_{1},s_{2})\right]\\
\left[\Member(x,s_{1}),\Diff(s_{1},s_{2})\right]\\
\left(\neg\Member(x,\Diff(s_{1},s_{2}))\vee\Member(x,s_{1})\right)\wedge\\
\left(\neg\Member(x,\Diff(s_{1},s_{2}))\vee\neg\Member(x,s_{2})\right)
\end{array}
\]

\subsubsection{Subset \\}

\noindent subset-intro 
\[
\forall s_{1},s_{2}.\Subset(s_{1},s_{2})\leftarrow\left(\forall x.\Member(x,s_{1})\rightarrow\Member(x,s_{2})\right)
\]
\[
\begin{array}{l}
\forall s_{1},s_{2}.\left[\Subset(s_{1},s_{2})\right]\\
\left(\Subset(s_{1},s_{2})\vee\Member(\Skolem_{ss}(s_{1},s_{2}),s_{1})\right)\wedge\\
\left(\Subset(s_{1},s_{2})\vee\neg\Member(\Skolem_{ss}(s_{1},s_{2}),s_{2})\right)
\end{array}
\]

\noindent subset-elim 
\[
\forall s_{1},s_{2}.\Subset(s_{1},s_{2})\rightarrow\left(\forall x.\Member(x,s_{1})\rightarrow\Member(x,s_{2})\right)
\]
\[
\begin{array}{l}
\forall s_{1},s_{2}.\left[\Subset(s_{1},s_{2})\right]\;\neg\Subset(s_{1},s_{2})\vee\\
\left(\forall x.[\Member(x,s_{1})][\Member(x,s_{2})]\;\neg\Member(x,s_{1})\vee\Member(x,s_{2})\right)
\end{array}
\]

\subsubsection{Extensionality \\}

\noindent equal-sets-intro
\[
\begin{array}{l}
\forall s_{1},s_{2}.\Equal(s_{1},s_{2})\leftarrow\\
\left(\forall x.\Member(x,s_{1})\leftrightarrow\Member(x,s_{2})\right)
\end{array}
\]
\[
\begin{array}{l}
\forall s_{1},s_{2}.\left[\Equal(s_{1},s_{2})\right]\\
\left(\Equal(s_{1},s_{2})\vee\Member(\Skolem_{\textit{eq}}(s_{1},s_{2}),s_{1})\vee\Member(\Skolem_{\textit{eq}}(s_{1},s_{2}),s_{2})\right)\wedge\\
\left(\Equal(s_{1},s_{2})\vee\neg\Member(\Skolem_{\textit{eq}}(s_{1},s_{2}),s_{1})\vee\neg\Member(\Skolem_{\textit{eq}}(s_{1},s_{2}),s_{2})\right)
\end{array}
\]

\noindent equal-sets-extensionality 
\[
\forall s_{1},s_{2}.\left[\Equal(s_{1},s_{2})\right]\Equal(s_{1},s_{2})\rightarrow s_{1}=s_{2}
\]
\[
\forall s_{1},s_{2}.\left[\Equal(s_{1},s_{2})\right]\neg\Equal(s_{1},s_{2})\vee s_{1}=s_{2}
\]

\subsubsection{Disjoint \\}

\noindent disjoint-intro 
\[
\forall s_{1},s_{2}.\Disjoint(s_{1},s_{2})\leftarrow\left(\forall x.\neg\Member(x,s_{1})\vee\neg\Member(x,s_{2})\right)
\]
\[
\begin{array}{r}
\forall s_{1},s_{2}.\left[\Disjoint(s_{1},s_{2})\right]\\
\left(\Disjoint(s_{1},s_{2})\vee\Member(\Skolem_{\textit{dj}}(s_{1},s_{2}),s_{1})\right)\wedge\\
\left(\Disjoint(s_{1},s_{2})\vee\Member(\Skolem_{\textit{dj}}(s_{1},s_{2}),s_{2})\right)
\end{array}
\]

\noindent disjoint-elim  
\[
\forall s_{1},s_{2}.\Disjoint(s_{1},s_{2})\rightarrow\left(\forall x.\neg\Member(x,s_{1})\vee\neg\Member(x,s_{2})\right)
\]
\[
\begin{array}{l}
\forall s_{1},s_{2}.\left[\Disjoint(s_{1},s_{2})\right]\neg\Disjoint(s_{1},s_{2})\vee\\
\left(\forall x.\left[\Member(x,s_{1})\right]\left[\Member(x,s_{2})\right]\neg\Member(x,s_{1})\vee\neg\Member(x,s_{2})\right)
\end{array}
\]

\subsubsection{Remove \\}

\noindent remove-intro-1
\[
\forall s,x,y.\Member(y,\Remove(x,s))\leftarrow y\neq x\wedge\Member(y,s)
\]
\[
\begin{array}{r}
\forall s,x,y.\left[\Member(y,s),\Remove(x,s)\right]\left[\Member(y,\Remove(x,s))\right]\\
y=x\vee\neg\Member(y,s)\vee\Member(y,\Remove(x,s))
\end{array}
\]

\noindent remove-intro-2 
\[
\forall s,x.\neg\Member(x,\Remove(x,s))
\]
\[
\forall s,x.\left[\Remove(x,s)\right]\;\neg\Member(x,\Remove(x,s))
\]

\noindent remove-intro-3
\[
\forall s,x.\neg\Member(y,\Remove(x,s))\leftarrow y=x
\]
\[
\begin{array}{r}
\forall s,x.\left[\Member(y,\Remove(x,s))\right]\left[\Member(y,s),\Remove(x,s)\right]\\
\neg\Member(y,\Remove(x,s))\vee y\neq x
\end{array}
\]

\noindent remove-elim
\[
\forall s,x,y.\Member(y,\Remove(x,s))\rightarrow y\neq x\wedge\Member(y,s)
\]
\[
\begin{array}{r}
\forall s,x,y.\left[\Member(y,\Remove(x,s))\right]\left[\Member(y,s),\Remove(x,s)\right]\\
\left(\neg\Member(y,\Remove(x,s))\vee y\neq x\right)\wedge\\
\left(\neg\Member(y,\Remove(x,s))\vee\Member(y,s)\right)
\end{array}
\]

\subsubsection{IsEmpty \\}

\noindent isEmpty-intro-1 
\[
\forall s.\left(\IsEmpty(s)\leftarrow\forall x.\neg\Member(x,s)\right)
\]
\[
\forall s.\left[\IsEmpty(s)\right]\;\IsEmpty(s)\vee\Member(\Skolem_{ie}(s),s)
\]

\noindent isEmpty-intro-2 
\[
\forall s.\;\IsEmpty(s)\leftarrow\Equal(s,\Empty())
\]
\[
\forall s.\left[\IsEmpty(s)\right]\left[\Equal(s,\Empty())\right]\;\IsEmpty(s)\vee\neg\Equal(s,\Empty())
\]

\noindent isEmpty-elim-1  
\[
\forall s.\left(\IsEmpty(s)\rightarrow\forall x.\neg\Member(x,s)\right)
\]
\[
\forall s.\left[\IsEmpty(s)\right]\;\neg\IsEmpty(s)\vee\forall x.[\Member(x,s)]\;\neg\Member(x,s)
\]

\noindent isEmpty-elim-2 
\[
\forall s.\;\IsEmpty(s)\rightarrow\Equal(s,\Empty())
\]
\[
\forall s.\left[\IsEmpty(s)\right]\left[\Equal(s,\Empty())\right]\;\neg\IsEmpty(s)\vee\Equal(s,\Empty())
\]

\subsection{Comparison of Our Axiomatisation for Set Theory with Dafny's and
Viper's \label{subsec:appx-comparison-axiomatisations-set-theory}}

We compare our axiomatisation for set theory with the counterparts
from Dafny and Viper. We perform the comparison by groups of axioms,
with each group focusing on one operation. For each axiom of ours,
we typically include two versions: one without triggers, and one in
extended CNF with triggers. 

We use the following labels to indicate where each axiom comes from. 
\begin{enumerate}
\item Label {[}all{]} means the axiom of question is part of our axiomatisation,
Dafny's and Viper's. 
\item Label {[}dafny, viper{]} means the axiom of question is only included
in the axiomatisations from Dafny and Viper. 
\item Label {[}dafny{]} means the axiom of question is only included in
the axiomatisation from Dafny. 
\item Label {[}viper{]} means the axiom of question is only included in
the axiomatisation from Viper. 
\item Label {[}dafny, ours{]} means the axiom of question is only included
in our axiomatisation, and the axiomatisation from Dafny. 
\item No presence of labels indicates the axiom of question is only included
in our axiomatisation. 
\end{enumerate}

\subsubsection{Empty \\ }

\noindent empty {[}all{]} 
\[
\forall x.\;\neg\Member(x,\Empty())
\]
\[
\forall x.\left[\Member(x,\Empty())\right]\;\neg\Member(x,\Empty())
\]

\subsubsection{Singleton \\}

\noindent singleton-id {[}dafny, viper{]} 
\[
\forall x.\left[\Singleton(x)\right]\Member(x,\Singleton(x))
\]

\noindent singleton-bi {[}dafny, viper{]} 
\[
\forall x,y.\left[\Member(y,\Singleton(x))\right]\Member(y,\Singleton(x))\leftrightarrow x=y
\]

\noindent singleton-intro-1
\[
\forall x.\Member(x,\Singleton(x))
\]
\[
\forall x.\left[\Singleton(x)\right]\Member(x,\Singleton(x))
\]

\noindent singleton-intro-2
\[
\forall x,y.\Member(y,\Singleton(x))\leftarrow x=y
\]
\[
\forall x,y.\left[\Member(y,\Singleton(x))\right]\Member(y,\Singleton(x))\vee x\neq y
\]

\noindent singleton-elim
\[
\forall x,y.\Member(y,\Singleton(x))\rightarrow x=y
\]
\[
\forall x,y.\left[\Member(y,\Singleton(x))\right]\neg\Member(y,\Singleton(x))\vee x=y
\]

\subsubsection{Add \\}

\noindent add-bi {[}dafny, viper{]}
\[
\forall s,x,y.\left[\Member(y,\Add(x,s))\right]\Member(y,\Add(x,s))\leftrightarrow\left(y=x\right)\vee\Member(y,s)
\]

\noindent add-intro-1 {[}all{]}
\[
\forall s,x,y.\;\Member(y,\Add(x,s))\leftarrow\Member(y,s)
\]
\[
\begin{array}{r}
\forall s,x,y.\left[\Member(y,s),\Add(x,s)\right]\left[\Member(y,\Add(x,s))\right]\\
\;\Member(y,\Add(x,s))\vee\neg\Member(y,s)
\end{array}
\]
Note that we added the trigger $\left[\Member(y,\Add(x,s))\right]$. 

\noindent add-intro-2 {[}all{]}
\[
\forall s,x.\;\Member(x,\Add(x,s))
\]
\[
\forall s,x.\left[\Add(x,s)\right]\;\Member(x,\Add(x,s))
\]

\noindent add-intro-3 
\[
\forall s,x,y.\;\Member(y,\Add(x,s))\leftarrow y=x
\]
\[
\begin{array}{r}
\forall s,x,y.\left[\Member(y,\Add(x,s))\right]\left[\Member(y,s),\Add(x,s)\right]\\
\Member(y,\Add(x,s))\vee y\neq x
\end{array}
\]

\noindent add-elim
\[
\forall s,x,y.\;\Member(y,\Add(x,s))\rightarrow(x=y)\vee\Member(y,s)
\]
\[
\begin{array}{r}
\forall s,x,y.\left[\Member(y,\Add(x,s))\right]\left[\Member(y,s),\Add(x,s)\right]\\
\neg\Member(y,\Add(x,s))\vee(x=y)\vee\Member(y,s)
\end{array}
\]

\subsubsection{Union \\}

\noindent union-bi {[}dafny, viper{]}
\[
\begin{array}{l}
\forall s_{1},s_{2},x.\left[\Member\left(x,\Union\left(s_{1},s_{2}\right)\right)\right]\\
\Member\left(x,\Union\left(s_{1},s_{2}\right)\right)\leftrightarrow\Member(x,s_{1})\vee\Member(x,s_{2})
\end{array}
\]

\noindent union-intro-1 {[}all{]}
\[
\forall s_{1},s_{2},x.\Member\left(x,\Union\left(s_{1},s_{2}\right)\right)\leftarrow\Member\left(x,s_{1}\right)
\]
\[
\begin{array}{r}
\forall s_{1},s_{2},x.\left[\Union(s_{1},s_{2}),\Member(x,s_{1})\right]\left[\Member\left(x,\Union\left(s_{1},s_{2}\right)\right)\right]\\
\Member\left(x,\Union\left(s_{1},s_{2}\right)\right)\vee\neg\Member\left(x,s_{1}\right)
\end{array}
\]
Note that we added the trigger $\left[\Member\left(x,\Union\left(s_{1},s_{2}\right)\right)\right]$. 

\noindent union-intro-2 {[}all{]}
\[
\forall s_{1},s_{2},x.\Member\left(x,\Union\left(s_{1},s_{2}\right)\right)\leftarrow\Member\left(x,s_{2}\right)
\]
\[
\begin{array}{r}
\forall s_{1},s_{2},x.\left[\Union(s_{1},s_{2}),\Member(x,s_{2})\right]\left[\Member\left(x,\Union\left(s_{1},s_{2}\right)\right)\right]\\
\Member\left(x,\Union\left(s_{1},s_{2}\right)\right)\vee\neg\Member\left(x,s_{2}\right)
\end{array}
\]
Note that we added the trigger $\left[\Member\left(x,\Union\left(s_{1},s_{2}\right)\right)\right]$. 

\noindent union-elim
\[
\forall s_{1},s_{2},x.\Member\left(x,\Union\left(s_{1},s_{2}\right)\right)\rightarrow\Member(x,s_{1})\vee\Member(x,s_{2})
\]
\[
\begin{array}{l}
\forall s_{1},s_{2},x.\\
\left[\Member(x,\Union(s_{1},s_{2}))\right]\\
\left[\Union(s_{1},s_{2}),\Member(x,s_{1})\right]\left[\Union(s_{1},s_{2}),\Member(x,s_{2})\right]\\
\;\neg\Member\left(x,\Union\left(s_{1},s_{2}\right)\right)\vee\Member(x,s_{1})\vee\Member(x,s_{2})
\end{array}
\]

\noindent union-disjoint {[}dafny, ours{]}
\[
\begin{array}{l}
\forall s_{1},s_{2}.\left[\Union(s_{1},s_{2})\right]\\
\Disjoint(s_{1},s_{2})\rightarrow\left(\Diff(\Union(s_{1},s_{2}),s_{1})=s_{2}\right)\wedge\left(\Diff(\Union(s_{1},s_{2}),s_{2})=s_{1}\right)
\end{array}
\]

\subsubsection{Intersection \\}

\noindent inter-bi {[}dafny{]}
\[
\begin{array}{l}
\forall s_{1},s_{2},x.\left[\Member(x,\Inter(s_{1},s_{2}))\right]\\
\Member(x,\Inter(s_{1},s_{2}))\leftrightarrow\Member(x,s_{1})\wedge\Member(x,s_{2})
\end{array}
\]
inter-bi {[}viper{]}
\[
\begin{array}{l}
\forall s_{1},s_{2},x.\\
\left[\Member(x,\Inter(s_{1},s_{2}))\right]\\
\left[\Inter(s_{1},s_{2}),\Member(x,s_{1})\right]\left[\Inter(s_{1},s_{2}),\Member(x,s_{2})\right]\\
\;\Member(x,\Inter(s_{1},s_{2}))\leftrightarrow\Member(x,s_{1})\wedge\Member(x,s_{2})
\end{array}
\]

\noindent inter-intro
\[
\forall s_{1},s_{2},x.\Member(x,\Inter(s_{1},s_{2}))\leftarrow\Member(x,s_{1})\wedge\Member(x,s_{2})
\]
\[
\begin{array}{l}
\forall s_{1},s_{2},x.\\
\left[\Member(x,s_{1}),\Inter(s_{1},s_{2})\right]\left[\Member(x,s_{2}),\Inter(s_{1},s_{2})\right]\\
\left[\Member(x,\Inter(s_{1},s_{2}))\right]\\
\;\Member(x,\Inter(s_{1},s_{2}))\vee\neg\Member(x,s_{1})\vee\neg\Member(x,s_{2})
\end{array}
\]

\noindent inter-elim
\[
\forall s_{1},s_{2},x.\Member(x,\Inter(s_{1},s_{2}))\rightarrow\Member(x,s_{1})\wedge\Member(x,s_{2})
\]
\[
\begin{array}{l}
\forall s_{1},s_{2},x.\left[\Member(x,\Inter(s_{1},s_{2}))\right]\\
\left[\Inter(s_{1},s_{2}),\Member(x,s_{1})\right]\left[\Inter(s_{1},s_{2}),\Member(x,s_{2})\right]\\
\left(\neg\Member(x,\Inter(s_{1},s_{2}))\vee\Member(x,s_{1})\right)\wedge\\
\left(\neg\Member(x,\Inter(s_{1},s_{2}))\vee\Member(x,s_{2})\right)
\end{array}
\]

\subsubsection{Properties on Union and Intersection \\}

\noindent union-right {[}all{]}
\[
\forall s_{1},s_{2}.\left[\Union(\Union(s_{1},s_{2}),s_{2})\right]\Union(\Union(s_{1},s_{2}),s_{2})=\Union(s_{1},s_{2})
\]

\noindent union-left {[}all{]}
\[
\forall s_{1},s_{2}.\left[\Union(s_{1},\Union(s_{1},s_{2}))\right]\Union(s_{1},\Union(s_{1},s_{2}))=\Union(s_{1},s_{2})
\]

\noindent inter-right {[}all{]}
\[
\forall s_{1},s_{2}.\left[\Inter(\Inter(s_{1},s_{2}),s_{2})\right]\Inter(\Inter(s_{1},s_{2}),s_{2})=\Inter(s_{1},s_{2})
\]

\noindent inter-left {[}all{]} 
\[
\forall s_{1},s_{2}.\left[\Inter(s_{1},\Inter(s_{1},s_{2}))\right]\Inter(s_{1},\Inter(s_{1},s_{2}))=\Inter(s_{1},s_{2})
\]

\subsubsection{Difference \\}

\noindent diff-bi {[}dafny{]}
\[
\begin{array}{l}
\forall s_{1},s_{2},x.\left[\Member(x,\Diff(s_{1},s_{2}))\right]\\
\Member(x,\Diff(s_{1},s_{2}))\leftrightarrow\Member(x,s_{1})\wedge\neg\Member(x,s_{2})
\end{array}
\]

\noindent diff-bi {[}viper{]}
\[
\begin{array}{r}
\forall s_{1},s_{2},x.\left[\Member(x,\Diff(s_{1},s_{2}))\right]\left[\Diff(s_{1},s_{2}),\Member(x,s_{1})\right]\\
\Member(x,\Diff(s_{1},s_{2}))\leftrightarrow\Member(x,s_{1})\wedge\neg\Member(x,s_{2})
\end{array}
\]

\noindent diff-notin {[}dafny, viper{]} 
\[
\begin{array}{l}
\forall s_{1},s_{2},x.\left[\Diff(s_{1},s_{2}),\Member(x,s_{2})\right]\\
\Member(x,s_{2})\rightarrow\neg\Member(x,\Diff(s_{1},s_{2}))
\end{array}
\]

\noindent diff-intro
\[
\forall s_{1},s_{2},x.\Member(x,\Diff(s_{1},s_{2}))\leftarrow\Member(x,s_{1})\wedge\neg\Member(x,s_{2})
\]
\[
\begin{array}{l}
\forall s_{1},s_{2},x.\left[\Member(x,s_{1}),\Diff(s_{1},s_{2})\right]\\
\left[\Member(x,\Diff(s_{1},s_{2}))\right]\left[\Member(x,s_{2}),\Diff(s_{1},s_{2})\right]\\
\Member(x,\Diff(s_{1},s_{2}))\vee\neg\Member(x,s_{1})\vee\Member(x,s_{2})
\end{array}
\]

\noindent diff-elim
\[
\forall s_{1},s_{2},x.\Member(x,\Diff(s_{1},s_{2}))\rightarrow\Member(x,s_{1})\wedge\neg\Member(x,s_{2})
\]
\[
\begin{array}{r}
\forall s_{1},s_{2},x.\left[\Member(x,\Diff(s_{1},s_{2}))\right]\left[\Member(x,s_{2}),\Diff(s_{1},s_{2})\right]\\
\left[\Member(x,s_{1}),\Diff(s_{1},s_{2})\right]\\
\left(\neg\Member(x,\Diff(s_{1},s_{2}))\vee\Member(x,s_{1})\right)\wedge\\
\left(\neg\Member(x,\Diff(s_{1},s_{2}))\vee\neg\Member(x,s_{2})\right)
\end{array}
\]

\subsubsection{Subset  \\}

\noindent subset-bi {[}dafny, viper{]} 
\[
\begin{array}{l}
\forall s_{1},s_{2}.\left[\Subset(s_{1},s_{2})\right]\Subset(s_{1},s_{2})\leftrightarrow\\
\left(\forall x.\left[\Member(x,s_{1})\right]\left[\Member(x,s_{2})\right]\Member(x,s_{1})\rightarrow\Member(x,s_{2})\right)
\end{array}
\]

\noindent subset-intro 
\[
\forall s_{1},s_{2}.\Subset(s_{1},s_{2})\leftarrow\left(\forall x.\Member(x,s_{1})\rightarrow\Member(x,s_{2})\right)
\]
\[
\begin{array}{l}
\forall s_{1},s_{2}.\left[\Subset(s_{1},s_{2})\right]\\
\left(\Subset(s_{1},s_{2})\vee\Member(\Skolem_{ss}(s_{1},s_{2}),s_{1})\right)\wedge\\
\left(\Subset(s_{1},s_{2})\vee\neg\Member(\Skolem_{ss}(s_{1},s_{2}),s_{2})\right)
\end{array}
\]

\noindent subset-elim 
\[
\forall s_{1},s_{2}.\Subset(s_{1},s_{2})\rightarrow\left(\forall x.\Member(x,s_{1})\rightarrow\Member(x,s_{2})\right)
\]
\[
\begin{array}{l}
\forall s_{1},s_{2}.\left[\Subset(s_{1},s_{2})\right]\;\neg\Subset(s_{1},s_{2})\vee\\
\left(\forall x.[\Member(x,s_{1})][\Member(x,s_{2})]\;\neg\Member(x,s_{1})\vee\Member(x,s_{2})\right)
\end{array}
\]

\subsubsection{Extensionality \\}

\noindent equal-sets-bi {[}dafny, viper{]} 
\[
\begin{array}{l}
\forall s_{1},s_{2}.\left[\Equal(s_{1},s_{2})\right]\Equal(s_{1},s_{2})\leftrightarrow\\
\left(\forall x.\left[\Member(x,s_{1})\right]\left[\Member(x,s_{2})\right]\Member(x,s_{1})\leftrightarrow\Member(x,s_{2})\right)
\end{array}
\]

\noindent equal-sets-extensionality {[}dafny, viper{]} 
\[
\forall s_{1},s_{2}.\left[\Equal(s_{1},s_{2})\right]\Equal(s_{1},s_{2})\rightarrow s_{1}=s_{2}
\]

\noindent equal-sets-intro
\[
\begin{array}{l}
\forall s_{1},s_{2}.\Equal(s_{1},s_{2})\leftarrow\\
\left(\forall x.\Member(x,s_{1})\leftrightarrow\Member(x,s_{2})\right)
\end{array}
\]
\[
\begin{array}{l}
\forall s_{1},s_{2}.\left[\Equal(s_{1},s_{2})\right]\\
\left(\Equal(s_{1},s_{2})\vee\Member(\Skolem_{\textit{eq}}(s_{1},s_{2}),s_{1})\vee\Member(\Skolem_{\textit{eq}}(s_{1},s_{2}),s_{2})\right)\wedge\\
\left(\Equal(s_{1},s_{2})\vee\neg\Member(\Skolem_{\textit{eq}}(s_{1},s_{2}),s_{1})\vee\neg\Member(\Skolem_{\textit{eq}}(s_{1},s_{2}),s_{2})\right)
\end{array}
\]

\noindent equal-sets-extensionality 
\[
\forall s_{1},s_{2}.\left[\Equal(s_{1},s_{2})\right]\Equal(s_{1},s_{2})\rightarrow s_{1}=s_{2}
\]
\[
\forall s_{1},s_{2}.\left[\Equal(s_{1},s_{2})\right]\neg\Equal(s_{1},s_{2})\vee s_{1}=s_{2}
\]

\subsubsection{Disjoint \\}

\noindent disjoint-bi {[}dafny{]} 
\[
\begin{array}{l}
\forall s_{1},s_{2}.\left[\Disjoint(s_{1},s_{2})\right]\Disjoint(s_{1},s_{2})\leftrightarrow\\
\left(\forall x.\left[\Member(x,s_{1})\right]\left[\Member(x,s_{2})\right]\neg\Member(x,s_{1})\vee\neg\Member(x,s_{2})\right)
\end{array}
\]

\noindent disjoint-intro 
\[
\forall s_{1},s_{2}.\Disjoint(s_{1},s_{2})\leftarrow\left(\forall x.\neg\Member(x,s_{1})\vee\neg\Member(x,s_{2})\right)
\]
\[
\begin{array}{r}
\forall s_{1},s_{2}.\left[\Disjoint(s_{1},s_{2})\right]\\
\left(\Disjoint(s_{1},s_{2})\vee\Member(\Skolem_{\textit{dj}}(s_{1},s_{2}),s_{1})\right)\wedge\\
\left(\Disjoint(s_{1},s_{2})\vee\Member(\Skolem_{\textit{dj}}(s_{1},s_{2}),s_{2})\right)
\end{array}
\]

\noindent disjoint-elim  
\[
\forall s_{1},s_{2}.\Disjoint(s_{1},s_{2})\rightarrow\left(\forall x.\neg\Member(x,s_{1})\vee\neg\Member(x,s_{2})\right)
\]
\[
\begin{array}{l}
\forall s_{1},s_{2}.\left[\Disjoint(s_{1},s_{2})\right]\neg\Disjoint(s_{1},s_{2})\vee\\
\left(\forall x.\left[\Member(x,s_{1})\right]\left[\Member(x,s_{2})\right]\neg\Member(x,s_{1})\vee\neg\Member(x,s_{2})\right)
\end{array}
\]

\subsubsection{Remove \\}

\noindent remove-intro-1
\[
\forall s,x,y.\Member(y,\Remove(x,s))\leftarrow y\neq x\wedge\Member(y,s)
\]
\[
\begin{array}{r}
\forall s,x,y.\left[\Member(y,s),\Remove(x,s)\right]\left[\Member(y,\Remove(x,s))\right]\\
y=x\vee\neg\Member(y,s)\vee\Member(y,\Remove(x,s))
\end{array}
\]

\noindent remove-intro-2 
\[
\forall s,x.\neg\Member(x,\Remove(x,s))
\]
\[
\forall s,x.\left[\Remove(x,s)\right]\;\neg\Member(x,\Remove(x,s))
\]

\noindent remove-intro-3
\[
\forall s,x.\neg\Member(y,\Remove(x,s))\leftarrow y=x
\]
\[
\begin{array}{r}
\forall s,x.\left[\Member(y,\Remove(x,s))\right]\left[\Member(y,s),\Remove(x,s)\right]\\
\neg\Member(y,\Remove(x,s))\vee y\neq x
\end{array}
\]

\noindent remove-elim
\[
\forall s,x,y.\Member(y,\Remove(x,s))\rightarrow y\neq x\wedge\Member(y,s)
\]
\[
\begin{array}{r}
\forall s,x,y.\left[\Member(y,\Remove(x,s))\right]\left[\Member(y,s),\Remove(x,s)\right]\\
\left(\neg\Member(y,\Remove(x,s))\vee y\neq x\right)\wedge\\
\left(\neg\Member(y,\Remove(x,s))\vee\Member(y,s)\right)
\end{array}
\]

\subsubsection{IsEmpty \\}

\noindent isEmpty-intro-1 
\[
\forall s.\left(\IsEmpty(s)\leftarrow\forall x.\neg\Member(x,s)\right)
\]
\[
\forall s.\left[\IsEmpty(s)\right]\;\IsEmpty(s)\vee\Member(\Skolem_{ie}(s),s)
\]

\noindent isEmpty-intro-2 
\[
\forall s.\;\IsEmpty(s)\leftarrow\Equal(s,\Empty())
\]
\[
\forall s.\left[\IsEmpty(s)\right]\left[\Equal(s,\Empty())\right]\;\IsEmpty(s)\vee\neg\Equal(s,\Empty())
\]

\noindent isEmpty-elim-1  
\[
\forall s.\left(\IsEmpty(s)\rightarrow\forall x.\neg\Member(x,s)\right)
\]
\[
\forall s.\left[\IsEmpty(s)\right]\;\neg\IsEmpty(s)\vee\forall x.[\Member(x,s)]\;\neg\Member(x,s)
\]

\noindent isEmpty-elim-2 
\[
\forall s.\;\IsEmpty(s)\rightarrow\Equal(s,\Empty())
\]
\[
\forall s.\left[\IsEmpty(s)\right]\left[\Equal(s,\Empty())\right]\;\neg\IsEmpty(s)\vee\Equal(s,\Empty())
\]

\end{appendix}

\end{document}